%% file: main.tex
\documentclass[11pt,pdfa,letterpaper]{article}
\input{headers.tex}

\date{}
\title{$\class{stateQIP} = \class{statePSPACE}$}

\author[1]{Tony Metger\footnote{Email: \href{mailto:tmetger@ethz.ch}{tmetger@ethz.ch}}}
\author[2]{Henry Yuen\footnote{Email: \href{mailto:hyuen@cs.columbia.edu}{hyuen@cs.columbia.edu}}}
\affil[1]{Institute for Theoretical Physics, ETH Zurich}
\affil[2]{Department of Computer Science, Columbia University}

\begin{document}
\maketitle

\vspace{-0.7cm}
\begin{abstract}
Complexity theory traditionally studies the hardness of solving classical computational problems. 
In the quantum setting, it is also natural to consider a different notion of complexity, namely the complexity of physically preparing a certain quantum state.
We study the relation between two such \emph{state complexity classes}: $\statePSPACE$, which contains states that can be generated by space-uniform polynomial-space quantum circuits, and $\stateQIP$, which contains states that a polynomial-time quantum verifier can generate by interacting with an all-powerful untrusted quantum prover.
The latter class was recently introduced by Rosenthal and Yuen (ITCS 2022), who proved that $\statePSPACE \subseteq \stateQIP$.

Our main result is the reverse inclusion, $\stateQIP \subseteq \statePSPACE$, thereby establishing equality of the two classes and providing a natural state-complexity analogue to the celebrated $\class{QIP}=\class{PSPACE}$ theorem of Jain, et al. (J.~ACM 2011).
To prove this, we develop a polynomial-space \emph{quantum} algorithm for solving a large class of exponentially large ``\pspace-computable'' semidefinite programs (SDPs), which also prepares an optimiser \emph{encoded in a quantum state}. 
Our SDP solver relies on recent block-encoding techniques from quantum algorithms, demonstrating that these techniques are also useful for complexity theory.

Using similar techniques, we also show that optimal prover strategies for general quantum interactive protocols can be implemented in quantum polynomial space. We prove this by studying an \emph{algorithmic} version of Uhlmann's theorem and establishing an upper bound on the complexity of implementing Uhlmann transformations.

\end{abstract}

\newpage 
{
\hypersetup{linkcolor=black}
\setcounter{tocdepth}{2}
\tableofcontents
}

\newpage

\section{Introduction}

Classical complexity theory studies the hardness of computational problems on classical computers.
Quantum complexity theory so far has been mostly focused on the hardness of solving \emph{classical} computational problems, i.e.~problems with classical inputs and outputs, on \emph{quantum} computers.
However, quantum computers are not restricted to solving classical problems; they can also perform \emph{inherently quantum} tasks such as quantum state synthesis, where the goal is to physically synthesize a quantum state satisfying certain properties.
Examples of such tasks include preparing ground states of a local Hamiltonian or trying to clone a quantum money state.
Recent explorations of the complexity of such state synthesis problems have often yielded surprising results without classical analogues~\cite{aaronson2016complexity,irani2021quantum}.

Here, we focus on \emph{interactive proofs} for state synthesis, a notion recently studied by Rosenthal and Yuen~\cite{rosenthal2022interactive}. In this model, a polynomial-time quantum verifier is given some implicit description of a family of quantum states $(\ket{\psi_n})_{n \in \N}$ and, given an index $n \in \N$, has to synthesize an approximation to $\ket{\psi_n}$ with the help of an untrusted but all-powerful prover. On the one hand, the prover in principle has the ability to help the verifier synthesize the state $\ket{\psi_n}$, which may be extremely complex and require more than polynomial time to synthesize by oneself. On the other hand, the verifier now has to ensure that the prover is not maliciously misdirecting the verifier to synthesize some other state that is far from $\ket{\psi_n}$. 
The question raised by~\cite{rosenthal2022interactive} is the following: what states are synthesizable in this interactive model?

To study this question formally, we need to introduce the \emph{state complexity classes} $\stateQIP$ and $\statePSPACE$ from~\cite{rosenthal2022interactive}. 
Unlike traditional complexity classes, which are sets of decision languages (each of which is a  set of binary strings), state complexity classes are collections of infinite sequences of quantum states. We give somewhat informal definitions of these classes here; for formal (and slightly more general) definitions, see \Cref{sec:pspace} and \Cref{sec:stateqip}. 

For a function $\delta: \N \to \R$, we define $\statePSPACE_{\delta(n)}$ to denote the class of state sequences $(\ket{\psi_n})_{n \in \N}$ such that there is a space-uniform family of polynomial-space quantum circuits $C_n$ (a ``synthesis algorithm'') where for sufficiently large $n$, the output of $C_n$ is a state that is $\delta(n)$-close in trace distance to $\ket{\psi_n}$. We call $\delta(n)$ the \emph{closeness guarantee} of the synthesis algorithm.\footnote{Closeness guarantees of state complexity classes are an important aspect of their definitions. Tolerating some error in the output state allows for different state complexity classes to be meaningfully compared with each other. Furthermore, unlike with many models of randomized computation, it is unclear whether the closeness guarantee can be generically amplified. }

Similarly, the class $\stateQIP_{\delta(n)}$ is defined as the class of state sequences $(\ket{\psi_n})_{n \in \N}$ such that there is a quantum interactive protocol between a (quantum) prover and (quantum) verifier satisfying the following for all sufficiently large $n$: if the verifier on input $1^n$ accepts with probability at least $\frac{1}{2}$, its output state conditioned on accepting is guaranteed to be $\delta(n)$-close to $\ket{\psi_n}$ (this is the \emph{soundness} condition). Furthermore, there exists a prover that is accepted by the protocol with probability $1$ (this is the \emph{completeness} condition).

The main result of \cite{rosenthal2022interactive}
is the following:

\begin{theorem}[\cite{rosenthal2022interactive}]
\label{thm:ry-informal}
For all functions $\delta(n)$ and polynomials $q(n)$, it holds that $\statePSPACE_{\delta(n)} \subseteq \stateQIP_{\delta(n) + 1/q(n)}$.
\end{theorem}

Their protocol takes advantage of the $\IP = \PSPACE$ protocol of~\cite{lund1992algebraic,shamir1992ip} by running \emph{exponentially} many different invocations of it in superposition, and carefully performing checks to ensure that a malicious prover is not surreptitiously entangling himself with the desired output state. The protocol and its analysis highlight many of the challenges that one encounters in the state synthesis setting, such as the unclonability of quantum states, the difficulty of testing the equality of quantum states, and the general lack of search-to-decision reductions~\cite{irani2021quantum}.

\subsection{Our results}

Our main result is the reverse containment $\stateQIP \subseteq \statePSPACE$ -- in other words, the statement that all states that can be verifiably synthesized in an interactive proof can also be synthesized in polynomial space.
More precisely, we prove the following.

\begin{theorem}[Main theorem]
\label{thm:main-informal}
 For all functions $\delta(n)$ and polynomials $q(n)$, it holds that $\stateQIP_{\delta(n)} \subseteq \statePSPACE_{\delta(n) + 1/q(n)}$.
\end{theorem}

We note that both inclusions of \Cref{thm:ry-informal,thm:main-informal} incur an inverse polynomial loss in the closeness parameters. This appears to be inherent to comparing state complexity classes; simulating one computational model with another to perform state synthesis is likely to introduce additional error. However \Cref{thm:ry-informal,thm:main-informal} show that the additional error can be made to be an arbitrarily small inverse polynomial. 

By defining $\stateQIP$ and $\statePSPACE$ as the intersections over all polynomials $p(n)$ of $\stateQIP_{1/p(n)}$ and $\statePSPACE_{1/p(n)}$, respectively, we can combine \cref{thm:ry-informal} and \cref{thm:main-informal} to establish the equality $\stateQIP = \statePSPACE$.\footnote{However, one could also consider taking the intersection over all \emph{exponentially small} closeness guarantees (in fact, this is how $\statePSPACE$ is defined in~\cite{rosenthal2022interactive}) -- but then it is not clear whether there are nontrivial state sequences in $\stateQIP_{\exp(-n)}$, for example. This suggests that allowing inverse polynomial error is the more natural notion in state complexity. We leave determining the properties of $\stateQIP_{\delta(n)}$ for negligibly-small $\delta(n)$ for future work.}
This can be viewed as the state synthesis analogue to the celebrated $\QIP = \PSPACE$-theorem of Jain, Ji, Upadhyay, and Watrous~\cite{jain2011qip}, which showed that the class of polynomial-space computable \emph{decision} problems is the same as the class of \emph{decision} problems decidable by a quantum interactive proof.

The difficult direction in~\cite{jain2011qip} is showing the inclusion $\QIP \subseteq \PSPACE$, i.e.~showing that a decision problem that can be decided with a quantum interactive proof can also be decided in \pspace.
This is the same direction as our~\cref{thm:main-informal}.
To prove $\QIP \subseteq \PSPACE$,~\cite{jain2011qip} first show that the acceptance probability of a $\QIP$ protocol can be computed by a particular exponentially large semidefinite program (SDP).
Given its exponential size, this SDP cannot even be written down in \pspace.
However,~\cite{jain2011qip} show that its value can be approximated by an algorithm using the \emph{matrix multiplicative weights update (MMWU)} framework~\cite{arora2016combinatorial,kale2007efficient}, and that the output of this algorithm can be approximated using polynomial-depth (but exponentially wide) classical circuits, described by the complexity class $\class{NC(\poly)}$.
This analysis is quite involved because it requires careful treatment of the approximation error and leverages several highly non-trivial results on implementing matrix functions with depth-bounded classical circuits~\cite{csanky1975fast,borodin1982fast,borodin1983parallel,ben1986fast,neff1994specified}.
Having established an $\class{NC(\poly)}$-algorithm for computing the particular SDP arising from QIP protocols, the result of~\cite{jain2011qip} follows because $\PSPACE = \class{NC(\poly)}$~\cite{borodin1977relating}.

Similarly to~\cite{jain2011qip}, we can express a given $\stateQIP$-protocol as an exponentially large SDP.
However, when trying to adapt the SDP algorithm of~\cite{jain2011qip} to this problem, one faces several obstacles.
The main difficulty is that the SDP from~\cite{jain2011qip} is meant to approximate the maximum acceptance probability of a quantum verifier, meaning that the SDP algorithm is solving a classical decision problem (``is the value of this SDP higher or lower than some threshold?'').
This makes it natural to use tools from classical complexity theory such as $\class{NC(\poly)} = \PSPACE$.
In contrast, in the state synthesis setting we want to synthesize the output state of the verifier, which is generally a highly entangled state on many qubits. 
Such a state corresponds to an (exponentially large) feasible solution of an SDP, but it is not sufficient to approximate the value of that SDP or compute individual entries of this feasible solution: instead we want to physically generate a quantum state whose (exponentially large) density matrix corresponds to a feasible solution.

A secondary difficulty is that the SDP algorithm from~\cite{jain2011qip} requires the SDP to have a special form.
For $\QIP$-protocols, this special form can always be achieved by compiling the protocol into a \emph{three-round} format. 
It is an open question whether the same holds for $\stateQIP$, i.e.~whether all $\stateQIP$ protocols can be parallelized to three rounds (or even any constant) -- straightforward applications of the parallelization strategy of~\cite{kitaev2000parallelization} do not seem to work. 
Hence, the SDP of a $\stateQIP$ protocol does not have the special form required in~\cite{jain2011qip}.

It seems plausible that the $\class{NC(\poly)}$-based approach of~\cite{jain2011qip} could be adapted to the $\stateQIP$-setting, too, e.g.~by using a generalization of their approach presented in~\cite{gutoski2012parallel} and finding a way of converting that algorithm into one that physically prepares a quantum state corresponding to a feasible solution.
However, faced with the above difficulties, we use an entirely different approach and develop a \emph{quantum} polynomial space algorithm for solving exponentially large SDPs.
This provides a much more natural approach in the state complexity setting because our algorithm solves the SDP by actually \emph{constructing} a quantum state corresponding to a feasible solution instead of using classical computation to approximate the SDP value, and we give a direct analysis that does not use classical circuit complexity results.

Using this approach, we prove \Cref{thm:main-informal} as a consequence of a more general result: a general polynomial-space \emph{quantum} algorithm for solving a large class of exponentially-sized SDPs and preparing a nearly optimal solution to the SDP as a quantum state.
Most of this paper is devoted to the construction and analysis of this algorithm and we apply it to prove \cref{thm:main-informal} in \cref{sec:state_classes}.
As a bonus, this algorithm also gives an alternative proof of $\QIP = \PSPACE$ that requires no parallelization of the $\QIP$ protocol and no direct use of classical circuit complexity results. We now describe this quantum algorithm and its analysis in more detail.

\paragraph{Quantum algorithms for exponentially large SDPs.} 
It is a well-known fact that optimizing an SDP can be efficiently reduced to finding a feasible solution to an SDP. Thus in this paper we focus solely on feasibility SDPs, for which we use the following standard form: a (feasibility) SDP instance is a pair $(\Phi,B)$, where $\Phi: \linear(\C^D) \to \linear(\C^D)$ is a superoperator that maps $D \times D$ Hermitian matrices to $D \times D$ Hermitian matrices and $B$ is a $D \times D$ Hermitian matrix.\footnote{Of course, more generally $\Phi$ could have differing input and output dimensions. However, by adding spurious dimensions, we can always ensure that the input and output have the same dimension, so this assumption is without loss of generality.}
A feasible solution for $(\Phi,B)$ is a positive semidefinite matrix $X \in \linear(\C^D)$ that satisfies $\Phi(X) = B$. An $\eps$-feasible solution $X$ is one where $\| \Phi(X) - B \|_1 \leq \eps$, where $\| \cdot \|_1$ denotes the trace norm. We assume for normalization that $\Tr(X) = 1$; this can be done by rescaling the instance and rescaling the error parameter $\eps$. 
As usual, in order to study the asymptotic complexity we implicitly consider families of such SDPs indexed by a size parameter $n$ and let the dimension $D = 2^{\poly(n)}$.\footnote{When applying this algorithm to SDPs derived from $\stateQIP$ protocols, this parameter $n$ corresponds to the number of qubits used by the verifier in the protocol.}
Our goal is to design a quantum algorithm that solves such SDPs using $\poly(n)$ qubits.

Our algorithm can solve the class of feasibility SDP instances $(\Phi,B)$ that are both (a) \emph{small-width} and (b) \emph{\pspace-computable}. The small-width property means that $\|B\|_\infty \leq 1$ and the adjoint map $\Phi^*$ is \emph{contracting}, i.e., $\| \Phi^*(Y) \|_\infty \leq 1$ whenever $\|Y \|_\infty \leq 1$. The \pspace-computability of $(\Phi,B)$ means that each matrix entry of $B$ is computable in polynomial-space (i.e.~there exists a polynomial-space algorithm that, given $(\log D = \poly(n))$-bit indices $i, j$, outputs $B_{ij}$), and for all inputs $X$ the matrix entries of $\Phi(X)$ are computable by a polynomial-space algorithm that has oracle access to matrix entries of the input $X$.
When we say that an algorithm is provided a \pspace-computable SDP $(\Phi, B)$ as input, we mean that it has access to the \pspace-procedures for computing these entries.
We now state our main technical result:

\begin{theorem}[Quantum solver for small-width, $\PSPACE$-computable SDPs]
\label{thm:sdp-informal}
For all polynomials $q(n)$, there exists a polynomial-space quantum algorithm that, given as input a small-width $\PSPACE$-computable SDP instance $(\Phi,B)$, outputs a $1/q(n)$-feasible solution $X$ as a quantum state.
\end{theorem}

We make some remarks about the ``small-width'' property. This is very likely necessary as the class of \pspace-computable SDPs are expressive enough to capture $\mathsf{EXP}$-hard computations -- this is a consequence of the fact that solving LPs and SDPs is complete for $\mathsf{P}$ under logspace reductions~\cite{serna1991approximating}. It is widely believed that $\mathsf{PSPACE} \neq \mathsf{EXP}$, so therefore we do not expect general \pspace-computable SDPs to be solved in polynomial space (even with a quantum algorithm). Furthermore, many works on fast parallel algorithms for solving SDPs assume that some notion of ``width'' associated with the SDP is small~\cite{arora2005fast,klein1996efficient,jain2009two,jain2011qip}. While the definition of width differs slightly from paper to paper they all roughly measure the extent to which the constraints can be violated. The small-width property as we have defined it is also a measure of how far the primal constraints can be violated by a density matrix $X$: using the variational characterization of the spectral norm and H\"older's inequality, we have that $\| \Phi(X) - B \|_\infty \leq \| B \|_\infty +\max_{Y : \|Y \|_1 \leq 1} \langle \Phi^*(Y), X\rangle \leq \| B \|_\infty +\max_{Y : \|Y \|_1 \leq 1} \norm{\Phi^*(Y)}_\infty \norm{X}_1$, which is at most $2$ since we assumed $X$ is a density matrix. It turns out that this bound controls the number of iterations required by the algorithm.

Our SDP solver is inspired by recent works on solving SDPs on quantum computers~\cite{brandao2017quantum,brandao2017quantum2,van2020quantum}, which implement the Arora-Kale MMWU-based SDP solver~\cite{arora2016combinatorial} using quantum techniques. However, we cannot directly use their SDP solvers because they target SDPs with polynomially-many \emph{trace constraints} (i.e.~constraints of the form $\Tr(A_j X) = b_j$ for some collection of matrices $\{ A_j \}_j$ and scalars $\{ b_j \}_j$). Our notion of $\PSPACE$-computable SDPs, when expressed in terms of trace constraints, can generally involve \emph{exponentially} many trace constraints (because $\Phi$ is an exponentially-large superoperator).
The solvers of ~\cite{brandao2017quantum,brandao2017quantum2,van2020quantum} find an approximate solution $X$ that can violate each of the exponentially many constraints by an additive $\eps$. This is weaker than what we demand, which is that the \emph{trace distance} between $\Phi(X)$ and $B$ is at most $\eps$.\footnote{A simple illustration of the gap between the two notions of approximate feasibility is the following: consider an SDP instance $(\Phi,B)$ that enforces the constraint that $X = B$. This corresponds to exponentially many trace constraints, e.g., forcing equality of each matrix entry of $X$ and $B$; however deviating by $\eps$ in every entry is a far more lenient condition than deviating by $\eps$ in trace distance.}

We instead design a solver, also based on the MMWU framework, for solving feasibility SDPs of the form $\Phi(X) = B$ for density matrices $X$. If the dimension of the SDP is $D = 2^{\poly(n)}$ and the SDP has the small-width property, then the algorithm runs in $\poly(n)$ iterations and uses $\poly(n)$ qubits, and constructs a potential solution $X$ at each iteration.\footnote{However, implementing each iteration as a quantum procedure requires exponential time, so the algorithms as a whole requires exponential time (but, importantly, only a polynomial number of qubits).} The solution is encoded in the form of a \emph{block encoding}~\cite{gilyen2019qsvt,gilyen2019quantum,low2019hamiltonian}, which is a $\poly(n)$-qubit unitary $U$ whose top left corner approximates $X$.
Importantly this unitary $U$ is computed by an explicitly-described polynomial-space quantum algorithm. 
To implement each step of the MMWU algorithm, we make use of recent techniques for transforming block encodings in quantum algorithms (see~\cite{martyn2021grand} for a survey). 
In this way, our SDP solver is fully ``quantized'' in that every step of the algorithm relies heavily on quantum algorithmic techniques, providing an arguably more direct and intuitive approach than the $\class{NC(\poly)}$-based method of~\cite{jain2011qip}.
We give a more detailed overview of our SDP solver in \cref{sec:overview}.

\paragraph{Closure of $\statePSPACE$ under purification.} Rosenthal and Yuen~\cite{rosenthal2022interactive} defined the state complexity classes $\statePSPACE$ and $\stateQIP$ to be classes of sequences of \emph{pure} states. Their construction of interactive proofs for state synthesis to prove \Cref{thm:ry-informal} required that the states being synthesized are pure, and left open the question of whether the result could be extended to families of \emph{mixed states}.

More precisely, define a sequence of mixed states $(\rho_n)_{n\in \N}$ to be polynomial-space computable if there exists a space-uniform family of general quantum circuits $C_n$ (which may have mid-circuit measurements as well as the ability to reset qubits) that output $\rho_n$. Can such a sequence also be interactively and verifiably synthesized? 

One might be tempted to ``purify'' the output of the general quantum circuits $C_n$ by appealing to the principle of deferred measurement. However, the circuit $C_n$ may in general make an exponential number of intermediate measurements, and the standard way of deferring the measurements is to add an additional ancilla qubit for each intermediate measurement, yielding an exponential blow up of the space complexity. The recent results of~\cite{fefferman2021eliminating,girish2021eliminating} on eliminating intermediate measurements do not immediately apply here (as far as we can tell) because~\cite{fefferman2021eliminating} only deals with \emph{decision problems} and~\cite{girish2021eliminating} does not eliminate qubit reset operations. 

Our next result is that $\statePSPACE$ is closed under purification, where now we allow $\statePSPACE$ to contain not just sequences of pure states but also mixed states that are computable in polynomial space. 

\begin{theorem}[$\statePSPACE$ is closed under purification]
\label{thm:purification-informal}
Let $(\rho_n)_{n \in \N} \in \statePSPACE_{\delta(n)}$ denote a sequence of mixed states for some error function $\delta(n)$. Then there exists a sequence $(\ket{\psi}_n)_{n \in \N} \in \statePSPACE_{2 \sqrt{\delta(n)}}$ of pure states such that each $\ket{\psi_n}$ is a purification of $\rho_n$.
\end{theorem}

This theorem justifies broadening the definition of $\statePSPACE$ and $\stateQIP$ to include mixed states. The result $\statePSPACE = \stateQIP$ still holds: to interactively synthesize a sequence $(\rho_n)_n \in \statePSPACE$, the verifier can instead interactively synthesize a purification $(\ket{\psi_n})_n$ and then trace out a part of the output state to get $\rho_n$.\footnote{Technically, \cref{thm:purification-informal} is not quite sufficient for this argument: we also need that the verifier can efficiently compute a description of the classical Turing machine that  outputs the circuits for synthesising the purifications $(\ket{\psi_n})_n$ from the description of the Turing machine that outputs the circuits for synthesising $(\rho_n)_n$. We show that this is the case in the full version of the theorem (\cref{thm:purification}).} Conversely, the main result of this paper, that $\stateQIP \subseteq \statePSPACE$, holds irrespective of whether the desired state to be synthesized is pure or not. 

We prove \Cref{thm:purification-informal} in \Cref{sec:purification}.
The proof leverages the same algorithmic techniques that we use for our SDP solver, namely performing (space-efficient) transformations on block encodings in order to go from a block encoding of a mixed state to a block encoding of its purification. 

\paragraph{Complexity of optimal provers.} 
So far, we have focused on which quantum states a polynomial-time verifier can prepare by interacting with an all-powerful quantum prover.
Switching our focus to the prover, it is natural to ask what computational resources are actually needed to \emph{implement} the actions of an optimal prover in a quantum interactive protocol.
Jain, Ji, Upadhyay, and Watrous~\cite{jain2011qip} showed that estimating the acceptance probability of optimal provers for a family of interactive protocols is complete for $\PSPACE$. However, this does not immediately tell us whether the optimal provers' unitary operations can be uniformly computed by a family of polynomial-space quantum circuits.\footnote{Here, it is important that the question is about uniform computation of the provers' actions; otherwise every unitary on $n$ qubits can be implemented via an $n$-qubit circuit of size $2^{O(n)}$, but there is no guarantee \emph{a priori} that those circuits for a family of unitaries can be specified in a space-uniform manner.} 

The main difficulty is in translating a statement about the complexity of a \emph{decision} problem (i.e.~``is there a prover that makes this verifier accept with high probability?'') to a statement about the complexity of an associated \emph{unitary synthesis} problem (i.e.~``implement the unitary operations of a prover that is accepted with high probability''). In general, we do not have a very clear understanding of how these complexities relate to each other -- in fact, this is the essence of the Unitary Synthesis Problem posed by Aaronson and Kuperberg~\cite{aaronson2007quantum,aaronson2016complexity} and explored in the interactive setting  by Rosenthal and Yuen~\cite{rosenthal2022interactive}. 

We show that indeed the optimal provers can be implemented in uniform quantum polynomial space.

\begin{theorem}[Optimal provers in $\unitaryPSPACE$, informal]
\label{thm:prover-informal}
Let $(V_n)_{n \in \N}$ denote a family of polynomial-time quantum verifiers. Let $P_n$ denote a prover that is accepted with (near) optimal probability by verifier $V_n$, and let $U_{n,j}$ denote a unitary describing $P_n$'s action in the $j$'th round of the interaction. Then for any sequence $j(n)$ of round choices, the family of unitaries $(U_{n,j(n)})_{n}$ is in $\unitaryPSPACE_{1/\poly(n)}$.
\end{theorem}
Here, $\unitaryPSPACE_{\delta(n)}$ is a \emph{unitary complexity class}, consisting of sequences of unitaries $(U_n)_{n \in \N}$ such that there is a space-uniform family $(C_n)_{n \in \N}$ of polynomial-space quantum circuits (which may involve intermediate measurements and other non-unitary operations) where $C_n(\ket{\psi})$ is $\delta(n)$-close to $U_n \ket{\psi}$ for all input states $\ket{\psi}$. The notion of $\unitaryPSPACE$ and other unitary complexity classes were introduced by Rosenthal and Yuen~\cite{rosenthal2022interactive} to study the complexity of implementing unitary transformations.
The formal version of this theorem is presented as \Cref{thm:prover-complexity}. We prove this as a corollary of a more general result about implementing \emph{Uhlmann transformations}, which we describe next.

\paragraph{Uhlmann Transformation Problem.} The well-known Uhlmann's theorem~\cite{uhlmann1976transition} states that for two bipartite pure states $\ket{\psi},\ket{\varphi}$ on registers $\reg{A}$ and $\reg{B}$, there exists a unitary operator $U$ (that we call an \emph{Uhlmann transformation for the pair $\ket{\psi},\ket{\varphi}$}) acting only on register $\reg{B}$ such that
\[
    \bra{\psi} I_{\reg{A}} \otimes U \ket{\varphi} = \fidelity(\rho,\sigma) \,,
\]
where $\fidelity(\rho,\sigma) = \norm{\sqrt{\rho}\sqrt{\sigma}}_1$ denotes the (square root) fidelity between $\rho = \Tr_{\reg{B}}(\ketbra{\psi}{\psi})$ and $\sigma = \Tr_{\reg{B}}(\ketbra{\varphi}{\varphi})$. 
In other words, if two pure states have reduced subsystems on which the states are close, then to map one pure state close to the other it suffices to apply a unitary on the complement of the subsystem. 

We consider an \emph{algorithmic} version of Uhlmann's theorem, which we call the \emph{Uhlmann Transformation Problem}: given circuits (or perhaps succinct descriptions of them) that output $\ket{\psi}$ and $\ket{\varphi}$, implement an Uhlmann transformation for the pair $\ket{\psi},\ket{\varphi}$. 

We show that this problem is solvable in quantum polynomial space if the states $\ket{\psi}$ and $\ket{\varphi}$ are in $\statePSPACE$ (whereas Uhlmann's theorem as an information-theoretic statement of course holds for arbitrary states). Concretely, we prove the following (see \Cref{thm:uhlmann} for the formal statement):

\begin{theorem}[Algorithmic Uhlmann's Theorem, informal]
\label{thm:uhlmann-informal}
Let $(\ket{\psi_n})_n, (\ket{\varphi_n})_n$ be pure state families in $\statePSPACE$ where for each $n$ the states $\ket{\psi_n},\ket{\varphi_n}$ have the same number of qubits and the qubits can be divided into two registers $\reg{A}_n \reg{B}_n$. Then there exists a sequence of unitaries $\{ K_n \}_n \in \unitaryPSPACE_{1/\poly(n)}$ such that $K_n$ acts on register $\reg{B}_n$ and satisfies
\[
    \Big \| (I_{\reg{A}_n} \otimes K_n) \ket{\varphi_n} - \ket{\psi_n} \Big \|^2 \leq 2(1 - \fidelity(\rho_n,\sigma_n)) + \frac{1}{\poly(n)}\,,
\]
where $\rho_n,\sigma_n$ are the reduced density matrices of $\ket{\psi_n},\ket{\varphi_n}$ respectively on register $\reg{A}_n$. 
\end{theorem}

The relation between implementing an optimal prover strategy for a quantum interactive protocol and the Uhlmann Transformation Problem is as follows: consider a quantum interactive protocol between a verifier and prover, fix a round $j$, and suppose $\ket{\psi}_{\reg{W} \reg{M} \reg{Q}}$ denotes the global pure state of the protocol right after the $j$'th message has been sent to the prover and $\ket{\varphi}_{\reg{W} \reg{M} \reg{Q}}$ denotes the global protocol state after the prover has responded with the $(j+1)$st message. Here $\reg{W}, \reg{M}, \reg{Q}$ denote the verifier's private workspace, the message register that is passed between verifier and prover, and the prover's private workspace, respectively. 
Since the state of the verifier's private workspace register $\reg{W}$ has not changed between it sending out and receiving the prover's message, we have that $\Tr_{\reg{MQ}}(\psi) = \Tr_{\reg{MQ}}(\varphi)$. 
Therefore, by Uhlmann's theorem there exists a unitary $U$ acting only on registers $\reg{M} \reg{Q}$ such that $(I_{\reg{W}} \otimes U) \ket{\psi} = \ket{\varphi}$. 
Thus, to implement the prover strategy it suffices to implement the Uhlmann transformations corresponding to purifications of consecutive ``snapshots'' of the reduced state on the verifier and message registers $\reg{W}\reg{M}$. 
A consequence of our proof of $\stateQIP \subseteq \statePSPACE$ is that these purifications of the intermediate states are in $\statePSPACE$. 
This, combined with our Algorithmic Uhlmann's Theorem (\Cref{thm:uhlmann-informal}), implies that there exists a successful honest prover strategy that can be computed in $\unitaryPSPACE$ (\Cref{thm:prover-informal}).

Our proof of \Cref{thm:uhlmann-informal} also uses the same quantum algorithmic techniques as used in our SDP solver: given circuits for the states $\ket{\psi_n}, \ket{\varphi_n}$, we constructively build a block encoding of the Uhlmann transformation for the pair $\ket{\psi_n},\ket{\varphi_n}$. This requires a number of transformations, including the oblivious amplitude amplification procedure of~\cite{berry2014exponential}; we provide an analysis of it in the approximate setting, which to our knowledge is novel.

Finally, we mention a broader motivation for considering the Uhlmann Transformation Problem. Although abstractly defined, it turns out to be a common computational task occurring in a variety of unitary synthesis problems, ranging from decoding black hole radiation~\cite{hayden2007black,harlow2013quantum,aaronson2016complexity} to quantum state merging~\cite{horodecki2007quantum},  entanglement distillation~\cite{abeyesinghe2009mother}, and attacks on quantum cryptography~\cite{lo1997quantum}. The recurrence of the Uhlmann Transformation Problem in these seemingly unrelated settings suggests that it may play a fundamental role in a complexity theory of unitary synthesis tasks.

\subsection{Technical overview of the SDP solver}
\label{sec:overview}

As mentioned above, the technical tool underlying our main result $\stateqip \subseteq \statePSPACE$ is an algorithm for solving exponentially large \pspace-computable SDPs.
Here we provide a brief overview of the algorithm and its implementation with space-bounded quantum circuits. Instead of using the primal-dual method of Arora and Kale~\cite{arora2016combinatorial}, which was used in the original proof of $\QIP = \PSPACE$ as well as works on quantum SDP solvers~\cite{brandao2017quantum,van2020quantum}, we instead adapt the zero-sum game approach to solving SDPs; this was presented in the classical setting in~\cite{kale2007efficient} and used in the quantum setting in~\cite{wu2010equilibrium,gutoski2012parallel,brandao2017quantum2}. 

At a high level, the algorithm works as follows. Let $(\Phi,B)$ be a small-width, $\PSPACE$-computable SDP instance with dimension $D = 2^{\poly(n)}$ (meaning that $\Phi$ maps $D \times D$ Hermitian matrices to $D \times D$ Hermitian matrices). Let $\eps = 1/\poly(n)$. 
Then, for $T = \frac{\ln D}{\eps^2}$ iterations, the algorithm generates a sequence of $D$-dimensional density matrices $\rho_1,\rho_2,\ldots,\rho_T$ as follows: 
\begin{enumerate}[label=\arabic*.]
    \item Set $\rho_1 = I/D$, the maximally mixed state.
    \item For $t = 1,\ldots,T-1$:
    \begin{enumerate}
        \item Compute a Hermitian matrix $H_t$ such that $\langle H_t, \Phi(\rho_t) - B \rangle = \| \Phi(\rho_t) - B \|_1$.
        \item Compute the density matrix $\rho_{t+1} = \frac{\exp \Big( - \eps \Phi^*(H_1 + \cdots + H_t) \Big)}{\Tr \Big( \exp \Big( - \eps \Phi^*(H_1 + \cdots + H_t) \Big) \Big)}$.
    \end{enumerate}
    \item Output $\rho = \frac{1}{T} \sum_{t = 1}^T \rho_t$.
\end{enumerate}
As before, the map $\Phi^*$ is the adjoint of $\Phi$ (which is also a $\PSPACE$-computable superoperator). In each iteration, the matrix $H_t$ can be thought of as identifying the directions in which the constraint $\Phi(\rho_t) = B$ is violated by the current ``hypothesis'' $\rho_t$. 
We can show that after $T$ iterations, the averaged hypothesis state $\rho$ will be $O(\eps)$-close to minimizing the trace norm $\| \Phi(X) - B \|_1$. If the SDP is indeed feasible, then $\rho$ is an $O(\eps)$-feasible solution. 
We also show that this algorithm is robust in the sense that if we only use an approximation to $H_t$ and compute $\rho_{t+1}$ up to some error, the output state $\rho$ is still approximately feasible (for a suitable choice of parameters).

We give a quantum implementation of this algorithm that uses $\poly(n)$ qubits of memory (but takes up to $\poly(D)$ time steps). Instead of maintaining the actual density matrices $\{ \rho_t \}_t$, the algorithm maintains in each iteration a block encoding of $\rho_t$, which is a unitary $U$ such that
\[
    \alpha (I \otimes \bra{0^a}) U (I \otimes \ket{0^a}) \approx \rho_t \,,
\]
where $\alpha >0$ is a scaling factor that we call the \emph{post-selection factor}, and $a$ is the number of ancilla qubits. Importantly, the unitary $U$ will be computable by a polynomial-space quantum circuit. 

We now want to apply operations to this block encoding to turn it into a block encoding of $\rho_{t+1}$.
For this, we first compute a block encoding of $\Phi(\rho_t)$ and $B$; this uses the $\PSPACE$-computability of the SDP instance.
We can combine this into a block encoding of $\Phi(\rho_t) - B$ by the linear combinations of unitaries (LCU) technique~\cite{berry2015simulating,gilyen2019qsvt}.

To get the Hermitian matrix $H_t$, ideally one would like to compute the \emph{sign function} of $\Phi(\rho_t) - B$ (i.e.~the matrix function that maps every positive eigenvalue to $1$ and negative eigenvalue to $-1$). 
We instead compute a \emph{polynomial approximation} of the sign function and apply it to $\Phi(\rho_t) - B$ to get (a block encoding of) an approximation of $H_t$. 
Because the eigenvalues of $\Phi(\rho_t) - B$ can be exponentially small, we require an exponentially precise approximation to the sign function, necessitating an exponential-degree polynomial.
Unfortunately, existing approximations used in the block encoding literature~\cite{low2017hamiltonian,gilyen2019quantum} are only designed for polynomial degree and it is not even clear whether their coefficients can be computed in \pspace if the degree is exponential.
This requires us to use a different approximation of the sign function, which is simply its orthogonal projection onto the Chebyshev polynomials.
The coefficients of this expansion can be computed explicitly (up to exponential degree) and we can analyse its error by relating it to the polynomial approximation from~\cite{low2017hamiltonian}.
We can apply this polynomial approximation (with exponential degree) to the block encoding of $\Phi(\rho_t) - B$ using polynomially many qubits and obtain a block encoding of an exponentially good approximation to $H_t$.
Similarly, instead of computing the exponential function exactly, we can use a polynomial approximation to construct the block encoding of (an approximation of) the state $\rho_{t+1}$ (similarly to e.g.,~\cite{gilyen2019qsvt}); this step requires that $\Phi^*$ is contracting, as the polynomial approximation of the exponential function can only be applied to matrices of bounded norm.

In this manner, we can transform a block encoding of $\rho_t$ into a block encoding of $\rho_{t+1}$.
However, because the transformation applies exponential-degree polynomials, the post-selection factor $\alpha$ grows exponentially.
Therefore, we need to use a fixed-point amplitude amplification procedure~\cite{grover2005fixed} (which again only requires polynomial space, but exponential time) to restore the parameters of the block encoding of $\rho_{t+1}$.
As a result, we can turn a ``good'' block encoding of $\rho_{t}$ into another ``good'' block encoding of $\rho_{t+1}$.
With a somewhat tedious error analysis and using the robustness of our SDP algorithm, we can show that the final output of this procedure is a block encoding of an approximately feasible state $\rho$.

Unrolling the MMWU loop, we see that this block encoding is a recursive composition of block encodings, where the recursion depth is polynomial. Since each recursion level only adds a polynomial \emph{additive} number of qubits to the required space, the overall space usage of the block encoding is polynomial.

\paragraph{From solving SDPs to $\stateQIP \subseteq \statePSPACE$.}
Having developed our SDP solver, we can use it to show $\stateQIP \subseteq \statePSPACE$ (allowing for inverse polynomial error as noted above).
For this, we express the $\stateQIP$-protocol as a feasibility SDP as described earlier.
More precisely, this feasibility problem has the property that any feasible solution corresponds to the intermediate states of running the $\stateqip$-protocol with a successful prover; the last of these intermediate states is the output of the $\stateqip$-protocol, which is what we want to synthesize.
We can use our SDP solver to compute a block-encoding of an approximately feasible solution to this SDP.
Then, we can extract this approximately feasible solution from the block encoding, so we obtain a $\statePSPACE$-preparation procedure for a state whose density matrix is an approximately feasible solution.

Unfortunately, unlike for $\QIP = \PSPACE$, an approximately feasible solution (i.e.~a solution that only violates the SDP by a little) is not sufficient.
Additionally, we need to show that this approximately feasible solution is close in trace distance to an exactly feasible solution, i.e.~that the approximate solution can be ``rounded'' to an exact solution.
This is required because the definition of $\statePSPACE$ requires generating the desired state up to some error in trace distance.
We show in \cref{lem:approx_feasible_rounding} that this rounding property does indeed hold for any SDP derived from a $\stateQIP$ protocol.
As a result, we obtain a $\statePSPACE$ algorithm for preparing the output state of a given $\stateqip$ protocol (up to arbitrary inverse polynomial error), completing the proof.

\subsection{Open problems}

We end the introduction by listing some open problems. 
\begin{enumerate}
    \item Can the completeness-soundness gap of $\stateQIP$ protocols be amplified without increasing the number of rounds? 
    \item Can the number of rounds in any $\stateQIP$ protocol be reduced to 3 (or any other constant)? This would match the corresponding result for $\QIP$ protocols~\cite{vidick2016quantum}.
    \item How do the exponential-precision versions of $\stateQIP$ and $\statePSPACE$ relate to each other?

    \item Delavenne, et al.~\cite{gall_stateqma} introduced the model of \emph{Merlin-Arthur proof systems} for state synthesis (see also~\cite{gall2022distributed}), in which there is a single message from the prover to the verifier. They showed that $\mathsf{statePreciseQMA}$, in which the completeness-soundness gap can be inverse exponential, is contained in $\statePSPACE$. Does the converse hold? This would be an interesting analogue of the $\mathsf{PreciseQMA} = \mathsf{PSPACE}$ result of Fefferman and Lin~\cite{fefferman2016quantum}. 
    
    \item We proved that optimal prover strategies can be implemented in $\unitaryPSPACE$. If we make a complexity assumption, such as $\mathsf{P} = \PSPACE$, can optimal prover strategies be implemented in $\class{unitaryBQP}$ (i.e.~implemented via polynomial-sized quantum circuits)?\footnote{We thank William Kretschmer for suggesting this question to us.}
    \item Rosenthal and Yuen~\cite{rosenthal2022interactive} also defined unitary complexity classes $\class{unitaryQIP}$ and $\unitaryPSPACE$. Does the analogous equality $\class{unitaryQIP} = \unitaryPSPACE$ hold? Neither $\unitaryPSPACE \subseteq \class{unitaryQIP}$ nor $\class{unitaryQIP} \subseteq \class{unitaryPSPACE}$ is yet known.
\end{enumerate}

\paragraph{Organisation.} The rest of the paper is organized as follows. \Cref{sec:prelims} establishes the notation and conventions used for quantum information theory, quantum circuits, and quantum states. In \Cref{sec:block-encodings} we develop primitives for transforming block encodings with exponential precision in polynomial (quantum) space.
In \Cref{sec:exp_sdp} we present our general quantum algorithm for solving $\PSPACE$-computable SDPs and show that by using a block-encoding based implementation, it can be solved with a polynomial number of qubits. In \Cref{sec:state_classes} we apply our SDP solver to $\stateQIP$ protocols and prove our main result, \Cref{thm:main-informal}. In \Cref{sec:purification} we prove that $\statePSPACE$ is closed under purification (\Cref{thm:purification-informal}). In \Cref{sec:strategy}, we show that optimal prover strategies can be implemented in quantum polynomial space (\Cref{thm:prover-informal}) by studying the more general Uhlmann Transformation Problem (\cref{thm:uhlmann-informal}).

\paragraph{Acknowledgments.} We thank Omar Fawzi, Andr\'{a}s Gily\'{e}n, William Krestchmer, Joe Renes, Gregory Rosenthal, and David Sutter for helpful discussions. We thank anonymous reviewers for their helpful feedback. 
 This work was done in part while the authors were visiting the Simons Institute for the Theory of Computing. TM acknowledges support from the ETH Z\"{u}rich Quantum Center. 
 HY is supported by AFOSR award FA9550-21-1-0040, NSF CAREER award CCF-2144219, and the Sloan Foundation.

\section{Preliminaries} \label{sec:prelims}

\paragraph{Quantum information theory.}
A \emph{register} $\reg{R}$ is a named finite-dimensional complex Hilbert space. If $\reg{A}, \reg{B}, \reg{C}$ are registers, for example, then the concatenation $\reg{A} \reg{B} \reg{C}$ denotes the tensor product of the associated Hilbert spaces. We abbreviate the tensor product state $\ket{0}^{\ot n}$ as $\ket{0^n}$. For a linear transformation $L$ and register $\reg R$, we write $L_{\reg R}$ to indicate that $L$ acts on $\reg R$, and similarly we write $\rho_{\reg R}$ to indicate that a state $\rho$ is in the register $\reg R$. We write $\Tr(\cdot)$ to denote trace, and $\Tr_{\reg R}(\cdot)$ to denote the partial trace over a register $\reg R$.
We denote the set of linear transformations on $\reg R$ by $\linear(\reg R)$.
For a pure state $\ket\varphi$, we write $\varphi$ to denote the density matrix $\ketbra{\varphi}{\varphi}$. We denote the identity transformation by $I$.
For an operator $X \in \linear(R)$, we define $\| X \|_\infty$ to be its operator norm, and $\| X\|_1 = \Tr(|X|)$ to denote its trace norm. We write $\td(\rho,\sigma) = \frac{1}{2} \| \rho - \sigma \|_1$ to denote the trace distance between two density matrices $\rho,\sigma$, and $\fidelity(\rho,\sigma) = \| \sqrt{\rho} \sqrt{\sigma} \|_1$ for the (square root) fidelity between $\rho,\sigma$. 

\paragraph{Families of quantum circuits and states.}
For convenience we assume that all quantum circuits use gates from the universal gate set $\{ H, \mathit{CNOT}, T \}$~\cite[Chapter 4]{nielsen2000quantum} (although our results hold for any universal gate set consisting of gates with algebraic entries). A \emph{unitary quantum circuit} is one that consists only of gates from this gate set. A \emph{general quantum circuit} is a quantum circuit that can additionally have non-unitary gates that (a) introduce new qubits initialized in the zero state, (b) trace them out, or (c) measure them in the standard basis. We say that a general quantum circuit uses space $s$ if the total number of qubits involved at any time step of the computation is at most $s$. The description of a general quantum circuit is a sequence of gates (unitary or non-unitary) along with a specification of which qubits they act on.
A general quantum circuit $C$ implements a quantum channel $\Phi_C: \linear(\reg R) \to \linear(\reg R')$ from some input register $\reg R$ to some output register $\reg R'$.

\begin{definition}[Polynomial size and space circuit families]
We say that $(C_n)_{n \in \N}$ is a family of \emph{polynomial-size general quantum circuits} if there exists a polynomial $p$ such that $C_n$ has size (i.e.~number of gates) at most $p(n)$. 
We say that $(C_n)_{n \in \N}$ is a family of \emph{polynomial-space general quantum circuits} if there exists a polynomial $p$ such that $C_n$ uses at most $p(n)$ space.
\end{definition}

\begin{definition}[Uniform circuit families]
A family of general quantum circuits $(C_n)_{n \in \N}$ is called \emph{time-uniform} (or simply \emph{uniform}) if $(C_n)_{n \in \N}$ is polynomial-size and there exists a classical polynomial-time Turing machine that on input $1^n$ outputs the description of $C_n$. Similarly, a family of general quantum circuits $(C_n)_{n \in \N}$ is called \emph{space-uniform} if $(C_n)_{n \in \N}$ is polynomial-space and there exists a classical polynomial-space Turing machine that on input $(1^n,i)$ outputs the $i$'th gate of $C_n$.
\end{definition}

\begin{definition}[\pspace-computability] \label{def:pspace_comp}
Let $\eps: \N \to [0,1]$ denote a function. Let $(\ket{\psi_n})_{n \in \N}$, $(B_n)_{n \in \N}$ , and $(\Phi_n)_{n \in \N}$ be a family of vectors, square matrices, and superoperators respectively whose dimensions are bounded by $2^{p(n)}$ for some polynomial $p(n)$. We let $d_n,d_n'$ be such that
$\ket{\psi_n} \in \C^{d_n}$, $B_n \in \C^{d_n \times d_n}$, and $\Phi_n: \linear(\C^{d_n}) \to \linear(\C^{d_n'})$. 
\begin{itemize}
\item We say that $(\ket{\psi_n})_n$ is \emph{$\eps$-$\PSPACE$-computable} if there exists a polynomial-space Turing machine $A$ that on input $(1^n,i)$, with $i \in [d_n]$ outputs a complex number $\alpha_i$ such that $\| (\alpha_i)_{i \in [d_n]} - \ket{\psi_n} \|_2 \leq \eps(n)$. 

\item We say that $(B_n)_n$ is \emph{$\eps$-$\PSPACE$-computable} if there exists a polynomial-space Turing machine $A$ that on input $(1^n,i,j)$, with $i,j \in [d_n]$ outputs a complex number $\alpha_{ij}$ such that $\| (\alpha_{ij})_{(i,j) \in [d_n] \times [d_n]} - B_n \|_1 \leq \eps(n)$. 

\item We say that $(\Phi_n)_n$ is \emph{$\eps$-$\PSPACE$-computable} if there exists a polynomial-space Turing machine $A$ that on input $(1^n,i,j)$ for $i,j \in [d_n']$, given oracle access to the entries of a matrix $X \in \C^{d_n \times d_n}$, outputs a value $\alpha_{ij}$ such that $\| (\alpha_{ij})_{(i,j) \in [d_n]\times [d_n]} - \Phi_n(X) \|_1 \leq \eps(n) \cdot \|X\|_1$.
\end{itemize}
If for every polynomial $q(n)$, an object is $2^{-q(n)}$-\pspace-computable, we drop the explicit $\eps$-dependence and simply call the object \pspace-computable.
\end{definition}

We will usually leave the size parameter $n$ implicit and e.g.~call a matrix \pspace-computable without explicitly specifying the family to which it belongs.

\subsection{Polynomial-space states and unitaries}
\label{sec:pspace}

As mentioned in the introduction, state complexity classes are sequences of quantum states that require certain resources (e.g.~a polynomial number of qubits) to be synthesized.
The first state complexity class we need to introduce is $\statePSPACE$.
We use the following definition of $\statePSPACE$, which generalizes the definition presented in~\cite{rosenthal2022interactive} to include sequences of mixed states.

\begin{definition}[$\statePSPACE$]
	Let $\delta: \N \to [0,1]$ be a function. Then $\statePSPACE_\delta$ is the class of all sequences of density matrices $(\rho_n)_{n \in \N}$ such that each $\rho_n$ is a state on $n$ qubits, and there exists a space-uniform family of general quantum circuits $(C_n)_{n \in \N}$ such that for all sufficiently large $n \in \N$, the circuit $C_n$ takes no inputs and $C_n$ outputs a density matrix $\sigma_n$ such that 
	\[
	\td(\sigma_n, \rho_n) \leq \delta(n)~.
	\]
	We define $\class{pureStatePSPACE}_{\delta}$ to be the subset of $\statePSPACE_\delta$ consisting of families of \emph{pure} states. We define the class $\statePSPACE$ to be
	\[
	    \statePSPACE = \bigcap_{q} \statePSPACE_{1/q(n)}
	\]
	where the intersection is over all polynomials $q:\N \to \R$, and similarly define $\class{pureStatePSPACE}$.
\end{definition}

The following lemma shows that the class $\statePSPACE_\delta$ is robust under perturbation.

\begin{lemma}
\label{lem:state-pspace-robust}
Let $(\rho_n)_{n \in \N} \in \statePSPACE_{\delta(n)}$ for some function $\delta(n)$. Suppose $(\tilde{\rho}_n)_{n \in \N}$ is a state sequence satisfying $\td(\psi_n,\tilde{\rho}_n) \leq \eps(n)$ for another function $\eps(n)$. Then $(\tilde{\rho}_n)_{n \in \N} \in \statePSPACE_{\delta(n) + \eps(n)}$. 
\end{lemma}
\begin{proof}
Let $A$ denote a $\statePSPACE_{\delta(n)}$ algorithm that synthesizes the sequence $(\rho_n)_{n \in \N}$ up to $\delta(n)$ error. Then by the triangle inequality, it also synthesizes the sequence $(\tilde{\rho}_n)_{n \in \N}$ up to $\delta(n) + \eps(n)$ error.
\end{proof}

In addition to state complexity classes, we also need to consider unitary complexity classes, which are sequences of unitaries that require certain resources (e.g.~a polynomial number of qubits acted upon by a space-uniform circuit) to implement (i.e.~to apply the unitary to any given input state).
\begin{definition}[$\unitaryPSPACE$]
	Let $\delta: \N \to [0,1]$ be a function. Then $\unitaryPSPACE_\delta$ is the class of all sequences $(U_n)_{n \in \N}$ such that each $U_n$ is a unitary acting on $n$ qubits, and there exists a space-uniform family of general quantum circuits $(C_n)_{n \in \N}$ such that for all sufficiently large $n \in \N$, for all $n$-qubit states $\ket{\psi}$ and
	\[
	    \td( C_n(\psi), U\psi U^\dagger) \leq \delta(n)~.
	\]
	We define the class $\unitaryPSPACE$ to be
 	\[
 	    \unitaryPSPACE = \bigcap_{q} \unitaryPSPACE_{\exp(-q(n))}
 	\]
 	where the intersection is over all polynomials $q:\N \to \R$.
\end{definition}

We will also need a version of $\unitaryPSPACE$ that does not allow mid-circuit measurements, which we call $\pureUnitaryPSPACE$.

\begin{definition}[$\pureUnitaryPSPACE$]
	Let $\delta: \N \to [0,1]$ be a function. Then $\pureUnitaryPSPACE_\delta$ is the class of all sequences $(U_n)_{n \in \N}$ such that each $U_n$ is a unitary acting on $n$ qubits, and there exists a space-uniform family of \emph{unitary} quantum circuits (i.e.~there are no measurements or tracing out) $(C_n)_{n \in \N}$ such that for all sufficiently large $n \in \N$ and  for all $n$-qubit states $\ket{\psi}$, 
	\[
	    \Big \| C_n \ket{\psi} \ket{0 \cdots 0} - (U \ket{\psi}) \ket{0\cdots 0} \Big \|_2 \leq \delta(n)~.
	\]
	We define the class $\pureUnitaryPSPACE$ to be
 	\[
 	    \pureUnitaryPSPACE = \bigcap_{q} \pureUnitaryPSPACE_{\exp(-q(n))}
 	\]
 	where the intersection is over all polynomials $q:\N \to \R$.
\end{definition}
Since the definition of $\pureUnitaryPSPACE$ requires a unitary circuit that returns any ancilla qubits to their original state, such unitaries can also be run coherently to simulate the controlled-$U$ operation. 
Also note that while our definition of $\statePSPACE$ allows inverse polynomial error, our definitions of $\unitaryPSPACE$ and \pup require inverse exponential error.
The reason for this will become clear later, but we briefly describe it here: we will show that we can approximate the output of \cref{algo:mmwu} to within inverse exponential error with polynomial-space unitaries, i.e.~morally speaking, \cref{algo:mmwu} is a \pup-algorithm for this exponentially precise definition of \pup.
However, \cref{algo:mmwu} itself (even for an exact implementation) can only produce a feasible density matrix to an SDP up to inverse polynomial error, so if we use this algorithm for a $\statePSPACE$ procedure, we need to allow inverse polynomial error in the $\statePSPACE$-state preparation.

\subsection{Exponentially precise \pspace-computable polynomial approximations}
\label{sec:poly_approx}
We will make extensive use of block encodings, which we will introduce in \cref{sec:block-encodings}, and will frequently want to apply functions to such block encodings.
However, we cannot apply general functions to block encodings.
Instead, we will need to approximate the function we want to apply as a linear combination of Chebyshev polynomials.
Many prior works on the block encoding framework (e.g.~\cite{gilyen2019qsvt,gilyen2019quantum,low2017hamiltonian}) construct polynomial approximations to functions of interest.
However, these approximations are designed to be used in a regime where the degree of the approximation is polynomial, which is required if one is restricted to polynomial time.
This will not be sufficient for our purposes: we will require approximations with exponential degree.
Such approximations cannot be applied in polynomial time, but as we will show in \cref{lem:be-polys}, they can be applied in quantum polynomial space.
However, simply using existing approximations and taking their degree to be exponential does not work: it is often not clear whether the coefficients of the approximation are computable in \pspace if the degree is exponential, e.g.~because these coefficients are expressed as products of doubly-exponentially small and large quantities.
In this section we construct exponentially good approximations to the sign function and the square root function with \pspace-computable coefficients using the orthogonal projection onto Chebyshev polynomials, which can be computed in \pspace by numerical integration.
We begin by recalling the definition of Chebyshev polynomials.
\begin{definition}[Chebyshev polynomials] \label{def:cheby}
The Chebyshev polynomials (of the first kind) $T_k(x)$ are defined via the following recursion relation: $T_0(x) = 1$, $T_1(x) = x$, and $T_{k+1}(x) = 2k T_k(x) - T_{k-1}(x)$.
For $x \in [-1,1]$, an equivalent definition is $T_k(\cos \theta) = \cos(k \theta)$.
\end{definition}

We denote $\langle f, g\rangle \deq \frac{2}{\pi} \int_{-1}^1 f(x) g(x) \frac{dx}{\sqrt{1-x^2}}$ for functions $f$ and $g$ for which this integral exists.
It is a standard property that $\langle \cdot , \cdot \rangle$ is an inner product on the space of polynomials on $[-1,1]$ of some fixed degree.
Furthermore, the Chebyshev polynomials are an orthogonal basis on this space.
More specifically, we can express any degree-$d$ polynomial $P_d$ as the following linear combination of Chebyshev polynomials:
\begin{align}
P_d = \frac{\langle T_0, P_d \rangle}{2} + \sum_{k = 1}^d \langle T_k, P_d \rangle T_k \,. \label{eqn:gen_cheby_expansion}
\end{align}
(The 1/2 factor for the $k=0$ term is necessary because $\langle T_0, T_0\rangle = 2$.)
We will make use of the following result from~\cite{powell1967maximum}.
\begin{lemma} \label{lem:unif_vs_proj}
Suppose that a function $f: [-1,1]\to \R$ has an $\eps$-good \emph{uniform} degree-$d$ polynomial approximation $P^*_d(x)$, i.e.~$\max_{x \in [-1,1]}|f(x) - P^*_d(x)| \leq \eps$.
Then the orthogonal projection of $f$ onto degree-$d$ polynomials given by $P_d = \frac{\langle T_0, f \rangle}{2} + \sum_{k = 1}^d \langle T_k, f \rangle T_k$ satisfies $\max_{x \in [-1,1]}|f(x) - P_d(x)| \leq O(\eps \log d)$.
\end{lemma}

We can use this and the approximations constructed in~\cite{gilyen2019quantum,low2017hamiltonian} to construct \pspace-computable approximations to the sign and square root function.
We treat each function in turn.

\begin{lemma}[Exponentially good approximation to the sign function] \label{lem:sign-approx}
For any $\kappa \geq 2^{-\poly(n)}$, there exists a $d = O \left( \frac{\log1/\kappa}{\kappa} \right) = O(2^{\poly(n)})$ and \pspace-computable coefficients $c_0, \dots, c_d$ such that the polynomial 
\begin{align*}
P_{d}^{\sgn} = \sum_{i = 0}^d c_i T_i
\end{align*}
is odd and satisfies $|\sgn(x) - P_{d}^{\sgn}(x)| \leq \kappa$ for all $x \in [-1, 1] \setminus [-\kappa, \kappa]$, and $|P_{d}^{\sgn}(x)| \leq 1 + \kappa$ for all $x \in [-1, 1]$.
Furthermore, the coefficient vector $c = (c_1, \dots, c_d)$ has norm bounded by $\norm{c}_1 \leq O(\log d)$.
\end{lemma}
\begin{proof}
We write $f \approx_{\eps, \kappa} g$ if $|f(x) - g(x)| \leq \eps$ for all $x \in [-1,1] \setminus [-\kappa, \kappa]$, and $f \approx_\eps g$ if this holds for $\kappa = 0$.
\cite[Lemma 10]{low2017hamiltonian} shows that $\sgn \approx_{O(\kappa), \kappa} g_k$ for $g_k(x) \deq \erf(kx)$ the rescaled error function and $k = O(\log(1/\kappa)^{1/2}/\kappa)$.
Furthermore, \cite[Corollary 4]{low2017hamiltonian} shows that for $d = O(\sqrt{(k^2 + \log(1/\kappa^2))\log(1/\kappa^2)})$ there exists a polynomial $P^*_d$ such that $g_k \approx_{O(\kappa^2)} P^*_d$.

Unfortunately, it is not clear whether the polynomial $P^*_d$ has \pspace-computable coefficients with respect to the basis of Chebychev polynomials when the degree is allowed to be exponential.
However, we can use the existence of $P^*_d$ combined with \cref{lem:unif_vs_proj} to show that the orthogonal projection of $g_k$ onto the Chebychev polynomials is also a good polynomial approximation, and the coefficients of this projection will be \pspace-computable.
Concretely, define $c_0 = \frac{\langle T_0, g_k \rangle}{2}$ and $c_i = \langle T_i, g_k \rangle$ for $i = 1, \dots, d$ and consider the polynomial $P_{d}^{\sgn} = \sum_{i = 0}^d c_i T_i$ as in the lemma statement.
Then, by \cref{lem:unif_vs_proj}, $g_k \approx_{O(\kappa^2 \log d)} P_{d}^{\sgn}$.
For our choice of $k$ and $d$, $O(\kappa^2 \log d) = O(\kappa)$, so $g_k \approx_{O(\kappa)} P_{d}^{\sgn}$.
Combining this with $\sgn \approx_{O(\kappa), \kappa} g_k$, we get that $\sgn \approx_{O(\kappa),\kappa} P_{d}^{\sgn}$ for $d = O(\log(1/\kappa)/\kappa) = O(2^{\poly(n)})$.
Choosing the implicit constant in $d = O(\log(1/\kappa)/\kappa)$ large enough, we can ensure that $\sgn \approx_{\kappa,\kappa} P_{d}^{\sgn}$.

Additionally, since $g_k$ is an odd function and $T_i$ is even if $i$ is even, $c_i = 0$ for even $i$.
As a result, $P_d^{\sgn}$ is a linear combination of the odd Chebyshev polynomials $T_i$ for odd $i$, so $P_d^{\sgn}$ is itself also an odd function as claimed.
Furthermore, with the implicit constant in $d = O(\log(1/\kappa)/\kappa)$ chosen large enough, $g_k \approx_{\kappa} P_{d}^{\sgn}$; since $|g_k(x)| \leq 1$ for $x \in [-1, 1]$, this means that $|P_{d}^{\sgn}(x)| \leq 1 + O(\kappa)$ for $x \in [-1, 1]$ as claimed.

It remains to argue that the coefficients $c_i$ are \pspace-computable.
For this, we observe that since $\frac{d}{dx}g_k(x) \leq k$ and $\frac{d}{dx} T_i(x) \leq O(i^2)$ for all $x \in \bits$, the integrand in 
\begin{align*}
\langle T_i, g_k \rangle \deq \frac{2}{\pi} \int_{-1}^1 T_i(x) g_k(x) \frac{dx}{\sqrt{1-x^2}}
\end{align*}
has derivative at most $O(2^{\poly(n)})$ for any $i \leq d = O(2^{\poly(n)})$.
Therefore, we can perform numerical integration with exponentially many integration nodes to estimate $\langle T_i, g_k \rangle$ to within accuracy $2^{-\poly(n)}$.
Since the integrand can also be evaluated to arbitrary accuracy in \pspace, this means that we can compute $c_i$ in \pspace up to accuracy $2^{-\poly(n)}$.

Finally we need to bound the norm of the coefficient vector $c = (c_1, \dots, c_d)$.
For this, we define coefficients $\tilde c_i \deq \langle T_i, \sgn \rangle$.
Direct integration shows that $\tilde c_i = (-1)^{(i-1)/2}\frac{4}{\pi \cdot i}$, so $\norm{\tilde c}_1 = \sum_{i = 1}^d |\tilde c_i| = O(\log d)$ by the formula for partial sums of the harmonic series.
We can now relate $\norm{c}$ and $\norm{\tilde c}$ by noting that since $g_k$ and $\sgn$ differ by at most 1 on the interval $[-\kappa, \kappa]$ and by at most $O(\kappa)$ on the rest of the interval, $|c_i - \tilde c_i| = O(\kappa)$.
Therefore, $|\norm{c}_1 - \norm{\tilde c}_1| \leq O(d \kappa) = O(\log (1/\kappa)) = O(\log(d))$.
As a result, $\norm{c}_1 = O(\log d)$ as claimed.
\end{proof}

\begin{lemma}[Exponentially good approximation to the square root function] \label{lem:sqrt-approx}
For any $\kappa \geq 2^{-\poly(n)}$, there exists a $d = O \left( \frac{\log1/\kappa}{\kappa^2} \right) = O(2^{\poly(n)})$ and \pspace-computable coefficients $c_0, \dots, c_d$ such that the polynomial 
\begin{align*}
P_{d}^{\sqrt{~}} = \sum_{i = 0}^d c_i T_i
\end{align*}
satisfies $\left|\sqrt{\frac{x+1}{2}} - P_{d}^{\sqrt{~}}(x)\right| \leq \kappa$ for all $x \in [-1, 1]$.
\end{lemma}
\begin{proof}
Define the function $g(x) = \sqrt{(1-c)\frac{x + 1}{2} + c}$ for $c = \kappa^2/8$.
Then for any $x \in [-1,1]$, $|\sqrt{(x+1)/2} - g(x)| \leq \kappa/2$.
To see that this is the case, denote $y = (x+1)/2$.
Then, using that $g(x) \geq \sqrt{(x+1)/2}$ on the interval $x \in [-1,1]$,
\begin{align*}
|g(x) - \sqrt{(x+1)/2}| = \sqrt{(1-c)y + c} - \sqrt{y} \,. 
\end{align*}
If $y \leq c$, this is trivially upper-bounded by $\sqrt{2 c} = \kappa/2$.
On the other hand, if $y > c$, we find that
\begin{align*}
\sqrt{(1-c)y + c} - \sqrt{y} &\leq \sqrt{(1-c)y} \left( \sqrt{1 + \frac{c}{(1-c)y}} - \sqrt{\frac{1  }{1-c}} \right) \\
&\leq \sqrt{(1-c)y} \frac{c}{2(1-c) y} \\
&\leq \frac{c}{2 \sqrt{(1 - c) y}} \leq \sqrt{c} \leq \kappa/2 \,.
\end{align*}
In the last line, we used $1 - c \geq 1/4$ and $y \geq c$.

The proof now concludes in the same way as for \cref{lem:sign-approx}: we define $c_0 = \frac{\langle T_0, g \rangle}{2}$ and $c_i = \langle T_i, g \rangle$ for $i = 1, \dots, d$ and consider the polynomial $P_{d}^{\sqrt{~}} = \sum_{i = 0}^d c_i T_i$ as in the lemma statement.
By \cite[Corollary 3.4.14]{gilyen2019quantum}, there exists a degree-$d$ polynomial $P^*_d$ satisfying $|g(x) - P^*_d(x)| \leq \kappa^2$ for all $x \in [-1,1]$ for some $d = O(\log(1/\kappa)/\kappa^2)$.
Then, taking the implicit constants in the degree large enough and applying \cref{lem:unif_vs_proj}, we find that $|\sqrt{(x+1)/2} - P_d^{\sqrt{~}}| \leq \kappa$ for all $x \in [-1,1]$.
The coefficients $c_i$ are \pspace-computable by the same argument as in \cref{lem:sign-approx}.
\end{proof}

\section{Block encodings}
\label{sec:block-encodings}

\begin{definition} \label{def:be}
Let $A \in \linear(\C^D)$. Then a unitary $U$ acting on $\C^{D} \otimes \cal{H}$ for some Hilbert space $\cal{H}$ of dimension $\dim(\cH) = 2^a$ is an \emph{$(\alpha, \eps, a)$-block encoding} of $A$ for some $\alpha \geq 1, \eps > 0$ if:
\[
    \norm{A - \alpha (I \otimes \bra{0^a}) \, U \, (I \otimes \ket{0^a})}_\infty \leq \eps~.
\]
\end{definition}
\medskip 

We call $\alpha$ the \emph{post-selection factor} and $\eps$ the \emph{error} of the block encoding.
The parameter $a$ denotes the number of ancilla qubits in the block encoding.

In the following, we present a number of lemmas that have the following form: given $\pureUnitaryPSPACE$-computable block encoding $U$ with property $X$, there exist a $\pureUnitaryPSPACE$-computable block encoding $V$ with property $Y$. Technically speaking these lemmas should refer to sequences of block encodings $(U_n)_{n \in \N}$ and $(V_n)_{n \in \N}$. However for clarity we omit the sequence notation, and implicitly assume that the block encodings being discussed are part of a uniformly-specified family of unitaries. Throughout the paper, all parameters (such as the dimension $D$, the normalization $\alpha$, the error $\eps$) except for the ancilla size are functions of a parameter $n$ that grows to infinity.
The dimension $D$ is always $2^{p(n)}$ for some polynomial $p$, i.e.~up to polynomial factors we can think of $n$ as the number of qubits on which the matrix $A$ acts. 

In the following lemma statements, we write ``$\poly(n)$'' as shorthand for some polynomial $p(n)$ which may depend on other polynomials specified earlier. For example, if we write ``Let $A \in \linear(\C^D)$ for some $D = 2^{\poly(n)}$. Then there exists an $\eps \leq 2^{-\poly(n)}$ such that...'' then we mean ``Let $p(n)$ be a polynomial and let $A \in \linear(\C^D)$ for some $D = 2^{p(n)}$. Then there exists a polynomial $q(n)$, depending on $p(n)$, and an $\eps \leq 2^{-q(n)}$ such that...''.

Importantly, we need to keep the ancilla size as a separate parameter in order to argue that the ancilla size increases by at most an \emph{additive} polynomial in $n$ in each transformation to the block encoding (instead of e.g.~being squared). For this, in each of the lemmas below, we will argue that the transformation maps a block encoding with  $a$ ancilla qubits to a block encoding with $a + \poly(n)$ ancilla qubits, where the $\poly(n)$-term is understood to be independent of $a$. If we were to set $a = \poly(n)$ and simply argue that each transformation maps from $\poly(n)$ to $\poly(n)$ number of ancillas, then the distinction between an additive and multiplicative polynomial increase in the ancilla number would be lost. This distinction will become important when we implement the MMWU algorithm with block encodings in \cref{sec:solving}.

\subsection{Preparing block encodings}
We recall the following results from \cite{gilyen2019quantum} that enable us to prepare block encodings of \pspace-computable matrices and reduced density matrices of quantum states.
\begin{lemma}[Block encodings for entry-computable matrices]
\label{lem:be-entry}
Let $A \in \linear(\C^D)$ be an $\eps$-\pspace-computable matrix whose entries have magnitude at most $1$ (with $D = 2^{\poly(n)}$ as mentioned above).
Then there exists a $(D, \eps, 1 + \log D)$-block encoding of $A$ that is \pup-computable.
\end{lemma}
\begin{proof}
Define the following unitary $U$ which acts on $2\log D + 1$ qubits: for $j \in [D]$, the unitary $U$ maps the basis vector $\ket{j} \otimes \ket{0^{\log D}} \otimes \ket{0}$ to the state
\[
    \ket{\theta_j} \deq \frac{1}{\sqrt{D}} \sum_{i=1}^D \ket{i} \otimes H^{\otimes \log D} \ket{j} \otimes \Big( \tilde{A}_{ij} \ket{0} + \sqrt{1 - |\tilde{A}_{ij}|^2} \ket{1}\Big)
\]
where $H^{\otimes \log D}$ denotes $\log D$ Hadamards applied to the binary representation of $j$, and $\tilde{A}_{ij}$ denotes the $\PSPACE$-computable entries of $\tilde{A}$ satisfying $\| \tilde{A} - A \|_1 \leq \eps$. 
Observe that the $\{ \ket{\theta_j} \}_j$ states are orthonormal and that $U$ is \pup-computable. We then see that
\begin{align*}
    (I \otimes \bra{0^{\log D + 1}}) U (I \otimes \ket{0^{\log D + 1}}) &= \frac{1}{D} \sum_{i,j} A_{ij} \ketbra{i}{j} = \frac{\tilde{A}}{D}~.
\end{align*}
Thus $\| D (I \otimes \bra{0^{\log D + 1}}) U (I \otimes \ket{0^{\log D + 1}}) - A \|_\infty \leq \| \tilde{A} - A \|_1 \leq \eps$ as desired.
\end{proof}

\begin{lemma}
\label{lem:be-partial-trace-0}
Let $U \in \linear(\C^D \otimes \C^D)$ be a $\pureUnitaryPSPACE$-computable unitary such that $\ket{\psi} \deq U \ket{0^{2\log D}}$ is a bipartite state on registers $\reg{X} \reg{Y}$.
Then there exists a $\pureUnitaryPSPACE$-computable $(1,0,2\log D)$-block encoding of the reduced density matrix $\Tr_{\reg{Y}}(\ketbra{\psi}{\psi})$.
\end{lemma}
\begin{proof}
This follows from~\cite[Lemma 3.3.2]{gilyen2019quantum}.
\end{proof}

Recalling \cref{def:be}, we immediately obtain the following robust version of the above lemma.
\begin{corollary} \label{lem:be-reduced-density}
Let $U$ be a $\pureUnitaryPSPACE$-computable unitary such that $\ket{\psi} \deq U \ket{0^n}$ is an $n$-qubit bipartite state on registers $\reg{X} \reg{Y}$.
Consider any state $\rho$ on register $\reg X$ such that $\norm{\rho - \ptr{\reg Y}{\proj{\psi}}}_\infty \leq \eps$.
Then $\rho$ has a $(1, \eps, n)$-block encoding in \pup.
\end{corollary}

\subsection{Basic operations on block encodings}
The next two lemmas describe how to obtain block encodings of products and linear combinations of block encodings.

\begin{lemma}[Products of block encodings]
\label{lem:be-product}
Let $U_A$ be an $(\alpha, \eps_A, a)$-block encoding of a matrix $A \in \linear(\C^D)$ and $U_B$ be an $(\beta, \eps_B, b)$-block encoding of a matrix $B \in \linear(\C^D)$.
Suppose $U_A$ and $U_B$ are both in \pup.
Then there exists a $(\alpha \beta, \beta \eps_A + \alpha \eps_B, a+b)$-block encoding of $A \cdot B$ in \pup.
\end{lemma}
\begin{proof}
The statement follows directly from \cite[Lemma 3.3.10]{gilyen2019quantum} by observing that if $U_A, U_B$ are in \pup, then $(I_A \ot U_A)(I_B \ot U_B)$ is also in \pup, where $I_A, I_B$ are identities acting on (potentially different) polynomial (in $\log D$) numbers of qubits.
\end{proof}

\begin{lemma}[Linear combinations of block encodings]
\label{lem:be-lin-comb}
Let $m = m(n) \leq 2^{\poly(n)}$, $y,\alpha \in \C^m$ be \pspace-computable with $\norm{y}_1 \leq 2^{\poly(n)}$ and $\alpha_j \leq 2^{\poly(n)}$ for all $j$. Let $\eps \geq 2^{-\poly(n)}$. 
Let $A_1,\ldots,A_m \in \linear(\C^{D})$ be matrices with $(\alpha_j, \eps, a)$-block encodings $U_j$ in \pup. Then, for all polynomials $q(n)$ there exists a $(\norm{\tilde y}_1, \norm{\tilde y}_1 (2^{-q(n)} + \eps'), a + \poly(n))$-block encoding $U$ of the linear combination $\sum y_i A_i$ in \pup, where $\tilde y = (\alpha_j y_j)_{j = 1,\ldots,m}$ and $\eps' = \eps/\min \alpha_j$. 
\end{lemma}
\begin{proof}
We first observe that for all $b > 0$, if $U$ is an $(\alpha, \eps, a)$-block encoding of $A$, then it is also an $(\alpha/b, \eps', a)$-block encoding of $A/b$ for all $\eps' \geq \eps/b$.
Therefore, $U_j$ is a $(1, \eps', a)$-block encoding of $A_j/\alpha_j$ for $\eps' = \eps/\min \alpha_j$. 

Now define a vector $\tilde y$ with entries $\tilde y_j = y_j \alpha_j$, and $y' = \frac{\tilde y}{\norm{\tilde y}_1}$.
Since $y,\alpha$ are \pspace-computable, $y'$ is also \pspace-computable.
Then, we can adapt the technique from~\cite{grover2002creating} to show that $\sum_i \sqrt{y'_i} \ket{i}$ is in $\statePSPACE$ and its preparation uses $O(\log(m)) = \poly(n)$ qubits. (Note that $y$ may contain negative elements, so $y'$ is not necessarily a probability distribution, but it is easy to see that the technique from~\cite{grover2002creating} still works.) 

Applying~\cite[Lemma 3.3.9]{gilyen2019quantum}, we see that there exists a $(1, 2^{-q(n)} + \eps', a + \poly(n))$-block encoding $U$ of $\sum y'_j \frac{A_j}{\alpha_j}$ for all polynomials $q(n)$.
From the proof of~\cite[Lemma 3.3.9]{gilyen2019quantum} it is also easy to see that since the block encodings of $A_j$ are in \pup, so is the block encoding of $\sum y'_j \frac{A_j}{\alpha_j}$.

We can conclude the proof by observing that $\sum y_j A_j = \norm{\tilde y}_1 \sum y'_j \frac{A_j}{\alpha_j}$, so $U$ is a $(\norm{\tilde y}_1, \norm{\tilde y}_1 (2^{-q(n)} + \eps'), a + \poly(n))$-block encoding of $\sum y_j A_j$.
\end{proof}

\begin{lemma}[Purifying block encodings]
\label{lem:be-purification}
Let $U$ denote a $\pureUnitaryPSPACE$-computable $(\alpha,\eps,a)$-block encoding of a matrix $A \in \linear(\C^D)$. Then there exists a $\pureUnitaryPSPACE$-computable $(\alpha, \eps, a + 2\log D)$-block encoding $V$ of the matrix $(A \otimes I) \ketbra{\Phi}{0^{2\log D}}$ where $\ket{\Phi} = \frac{1}{\sqrt{D}} \sum_{j=1}^D \ket{j}\ket{j}$ is the maximally entangled state of dimension $D$.  
\end{lemma}
\begin{proof}
    Consider the following circuit $W$ on $4 \log D$ qubits, divided into registers $\reg{X}, \reg{Y}$ of $2 \log D$ qubits each. 
    \begin{enumerate}
        \item Apply a unitary on $\reg{Y}$ that maps the all zeroes state to the maximally entangled state $\ket{\Phi}$.
        \item Swap the registers $\reg{X}$ and $\reg{Y}$.
    \end{enumerate}
    Observe that, treating $\reg{Y}$ as the ancilla register, $W$ is a $(1, 0, 2\log D)$-block encoding of the matrix $\ketbra{\Phi}{0 \cdots 0}_{\reg{X}}$ and is clearly $\pureUnitaryPSPACE$-computable. Thus, using \Cref{lem:be-product} we can combine the block encoding $U$ of $A$ with the block encoding $W$ to get a $\pureUnitaryPSPACE$-computable $(\alpha,\eps,a+2\log D)$-block encoding $V$ of $A \ketbra{\Phi}{0^{2\log D}}$. 
\end{proof}

\begin{lemma}[Partial trace of block encodings]
\label{lem:be-partial-trace}
Let $A \in \linear(\reg{X} \otimes \reg{Y})$ be a matrix (with $D_X = \dim(\reg{X}), D_Y = \dim(\reg{Y}) \leq 2^{\poly(n)}$) with an $(\alpha, \eps, a)$-block encoding $U$ in \pup.
Then there exists an $(D_Y \alpha, 2 D_Y \eps, a + \poly(n))$-block encoding U of $\ptr{Y}{A}$ (where the partial trace is over the second tensor factor with dimension $D_Y$) in \pup.
\end{lemma}
\begin{proof}
Let $\{ \ket{i} \}_{i = 1, \dots, \log_2 D_Y}$ denote a basis for register $\reg Y$ (where we assume without loss of generality that $D_Y$ is a power of 2) and let $J_i$ be a unitary acting on the $\reg Y$ register that maps $\ket{0}$ to $\ket{i}$. 
Observe that $\ptr{\reg Y}{A} = \sum_{i} \bra{0}_{\reg Y} J_i A J_i^\dagger \ket{0}_{\reg Y}$.
Since $A$ has an $(\alpha, \eps, a)$-block encoding $U$ in \pup and $J_i$ is efficiently computable, it is easy to see that $J_i U J_i$ is an $(\alpha, \eps, a + \log D_Y)$-block encoding of $\bra{0}_{\reg Y} J_i A J_i^\dagger \ket{0}_{\reg Y}$, where now $\reg Y$ is included in the ancilla register of the block encoding.
Furthermore, $J_i U J_i$ is clearly in \pup.
Therefore, since $D^Y \leq 2^{\poly(n)}$, we can apply \cref{lem:be-lin-comb} to obtain a $(D_Y \alpha, 2 D_Y \eps, a+\poly(n))$-block encoding of $\ptr{\reg Y}{A}$ in \pup.
\end{proof}

\begin{lemma}[Applying superoperators to block encodings]
\label{lem:be-superoperator}
Let $\Phi: \linear(\C^D) \to \linear(\C^D)$ be a \pspace-computable superoperator and $A \in \linear(\C^D)$ a matrix with an $(\alpha, \eps, a)$-block encoding in \pup.
Then, for all polynomials $q(n)$ there exists a $(\alpha D^3, 2 \alpha \eps D + 2 D^3 2^{-q(n)}, a + \poly(n))$-block encoding of $\Phi(A)$ in \pup.
\end{lemma}
\begin{proof}
Define the Choi matrix 
\begin{align*}
J = \sum_{i,j} \ketbra{i}{j}_\reg{X} \otimes \Phi(\ketbra{j}{i})_\reg{Y} \in \linear(\C^{D^2}) \,.
\end{align*}
Note that 
\[
    \Phi(A) = \Tr_\reg{X}( J (A \otimes I_{\reg Y})) \,.
\]
Since $\Phi$ is \pspace-computable, so is $J$.
Therefore, using \cref{lem:be-entry}, for all polynomials $q(n)$ there exists a $(D^2, 2^{-q(n)},\poly(n))$-block encoding of $J$ in \pup.
By assumption, $A$ has an $(\alpha, \eps, a)$-block encoding $U$ in \pup, whence it is easy to see that $I_\reg{Y} \ot U$ is an $(\alpha, \eps, a)$-block encoding of $I_{\reg Y} \ot A$ (and, after swapping the registers $\reg X$ and $\reg Y$, also of $A \ot I_{\reg Y}$).
Therefore, we can apply \cref{lem:be-product} to find that there exists a $(\alpha D^2, \alpha \eps + D^2 2^{-q(n)}, a + \poly(n))$-block encoding of $J (A \otimes I_{\reg Y})$ in \pup.
Then it follows from \cref{lem:be-partial-trace} and $D \leq 2^\poly(n)$ that there exists a $(\alpha D^3, 2 \alpha \eps D + 2 D^3 2^{-q(n)}, a + \poly(n))$-block encoding of $\Phi(A)$ in \pup.
\end{proof}

\subsection{Renormalising block encodings by fixed-point amplitude amplification}

The block encoding-based implementation of our SDP solver presented in \cref{sec:exp_sdp} will run into the problem that the post-selection factor $\alpha$ grows exponentially as we apply transformations to our block-encodings.
To remedy the situation, in this section we show that we can restore the post-selection factor to $\alpha = 1$ using a fixed-point amplitude amplification scheme.
This is formalised in the following lemma.

\begin{lemma}
Let $\gamma(n) \leq 2^{\poly(n)}$ and let $U$ be a $\pureUnitaryPSPACE$-computable $(\alpha,\eps,a)$-block encoding of the matrix $\ketbra{\psi}{0^{\log D}}$ for some $\log D$-qubit state $\ket{\psi}$ and for some post-selection factor $\alpha \leq \gamma(n)$. Then for all polynomials $q(n)$ there exists a  $(1,2\eps + 2^{-q(n)}, a+2\log D)$-block encoding $V$ of the matrix $\ketbra{\psi}{0^{\log D}}$~. Furthermore, $V$ is $\pureUnitaryPSPACE$-computable.
\label{lem:be-fixedpoint}
\end{lemma}
\begin{proof}
For simplicity we assume $D = 2^n$, although the proof also works for $D = 2^{\poly(n)}$. Note that $U$, by definition, satisfies the following:
\[
    \Big \| \ketbra{\psi}{0^n} - \alpha (I \otimes \bra{0^a}) U (I \otimes \ket{0^a}) \Big \|_\infty \leq \eps~.
\]
This implies that
\[
    \Big \| \ket{\psi} - \alpha (I \otimes \bra{0^a}) U \ket{0^{n+a}} \Big \|_2 \leq \eps
\]
where now the norm is the Euclidean norm. Let $
    \ket{\psi'} = c \alpha (I \otimes \bra{0^a}) U \ket{0^{n+a}}$
denote the $n$-qubit state where $c$ is such that $\ket{\psi'}$ has unit norm. It must be that $c$ satisfies $1 - \eps \leq c \leq 1 + \eps$
so that overall $\ket{\psi'}$ is at least $2\eps$-close to $\ket{\psi}$. We can write
\[
    \Pi U \ket{0^{n+a}} = (c\alpha)^{-1} \ket{\psi'} \otimes \ket{0^a}
\]
where $\Pi = I \otimes \proj{0^a}$ is the projector onto the ancilla register being zero.

We now appeal to the fixed-point amplitude amplification algorithm of~\cite{grover2005fixed}:\footnote{One could also try to use~\cite{yoder2014fixed}, which has a quadratic speedup compared to the algorithm of~\cite{grover2005fixed}, but we use the latter because of its simpler description.} 

\begin{theorem}[Fixed point amplitude amplification~\cite{grover2005fixed}]
\label{thm:fixedpoint}
Let $0 \leq \eta \leq 1$ and let $U$ denote a unitary such that
for a vector $\ket{s}$ and a projector $\Pi$ we have $\| \Pi U \ket{s} \|^2 = 1 - \eta$. For $m = 1,2,\ldots$ let $U^{(m)}$ denote the recursively defined unitary
\[
    U^{(m)} = U^{(m-1)} R_s U^{(m-1),\dagger} R_t U^{(m-1)}
\]
where
\begin{gather*}
    R_s = I - e^{-i\pi/3} \ketbra{s}{s}~, \\
    R_t = I - e^{-i\pi/3} \Pi~,
\end{gather*}
and $U^{(0)} = U$. Then we have that $\| \Pi U^{(m)} \ket{s} \|^2 = 1 - \eta^{3^m}$.
\end{theorem}

Applying \Cref{thm:fixedpoint} with $1 - \eta = (c\alpha)^{-1}$, $\ket{s} = \ket{0^{n+a}}$ and $\Pi = I \otimes \ketbra{0^a}{0^a}$, we get that $\hat{V} = U^{(m)}$ for $3^m = 2 q(n) \gamma(n)$ 
satisfies
 \[
    \hat{V} \ket{0^{n+a}} = \sqrt{1 - \eta^{3^m}} \ket{\psi'} \otimes \ket{0^a} + \ket{\theta'}
\]
where $\ket{\theta'}$ is orthogonal to any state where the ancilla qubits are zero, and 
\[
    1 - \eta^{3^m} = 1 - \Big(1 - \frac{1}{c\alpha} \Big)^{2q(n) \gamma(n)} \geq 1 - \Big(1 - \frac{1}{c\alpha} \Big)^{c\alpha q(n)} \geq 1 - 2^{-q(n)}~. 
\]
Note that the number of calls to $U$ and $U^\dagger$ in the recursive construction of $\hat{V} = U^{(m)}$ is $3^m = O(q(n) \gamma(n))$, and also it does not involve any extra ancilla qubits. It is easy to see that the unitaries $R_s$ and $R_t$ can be computed efficiently.

Then, we get that 
\[
    (I\otimes \bra{0^a}) \hat{V} \ket{0^{n+a}} = \sqrt{1 - \eta^{3^m}} \ket{\psi'}~.
\]
Recall that $U$, and thus $\hat{V}$ act on $n + a$ qubits. Now define a unitary $V$ that acts on $2n+a$ qubits as follows: first, in the last $n+a$ qubits, it computes $\hat{V}$. Then, it swaps the first $n$ with the second $n$ qubits. It is straightforward to verify that $V$ is $\pureUnitaryPSPACE$-computable, and that
\[
     (I \otimes \bra{0^{n+a}}) V (I \otimes \ket{0^{n+a}}) = \sqrt{1 - \eta^{3^m}} \ketbra{\psi'}{0^n}~.
\]
Thus we have
\begin{align*}
    \Big \| \ketbra{\psi}{0^n} - (I \otimes \bra{0^{n+a}})V(I \otimes \ket{0^{n+a}}) \Big \|_\infty &= \Big \| \ketbra{\psi}{0^n} - \sqrt{1 - \eta^{3^m}} \ketbra{\psi'}{0^n} \Big \|_\infty \\
    &\leq \Big \| \ket{\psi} - \ket{\psi'} \Big \|_2 + (1 - \sqrt{1 - \eta^{3^m}}) \\
    &\leq 2\eps + 2^{-q(n)}~.
\end{align*}
This directly implies that $V$ is a $(1,2\eps + 2^{-q(n)},a+2n)$-block encoding of $\ketbra{\psi}{0 \cdots 0}$, as desired.
\end{proof}

\subsection{Applying polynomials to block encodings}

Given a block encoding of some matrix $A$, we will want to transform this into a block encoding for a matrix $p(A)$ for some polynomial $p$.
We begin with the case where $p$ is a Chebyshev polynomial $T_k$ (see \cref{lem:chebyshev-block}), in which case this transformation is relatively straightforward; the only difference to previous work~\cite{gilyen2019quantum,gilyen2019qsvt} is that we require the degree of the Chebyshev polynomial to be exponentially large.
Since the Chebyshev polynomials for a basis of fixed-degree polynomials, we can extend this result to arbitrary polynomials $p(\cdot)$ by taking linear combinations of Chebyshev polynomials (\cref{lem:be-polys}).
In particular, this allows us to approximately apply the sign function and the square root function to block encodings using the polynomials approximations developed in \cref{sec:poly_approx} (see \cref{lem:be-sign-poly} and \cref{lem:be-sqrt}).
 
\begin{lemma}[Chebyshev polynomials of block encodings]
Let $k = 2^{\poly(n)}$ and let $A \in \linear(\C^D)$ denote a Hermitian matrix with an $(\alpha,\epsilon,a)$-block encoding $U$ in \pup. Then there exists a $(1,4k\sqrt{\eps/\alpha}, a+1)$-block encoding $V$ of $T_k(A/\alpha)$ computable in \pup.
\label{lem:chebyshev-block}
\end{lemma}
\begin{proof}
Let $\tilde{A} \in \linear(\C^D)$ denote the matrix $(I \otimes \bra{0^a}) U (I \otimes \ket{0^a})$. It follows from~\cite[Lemma 2.2.7]{gilyen2019quantum} that there is a $(1,0, a+1)$-block encoding $V$ of $T_k(\tilde{A})$. Furthermore, this block encoding makes $k$ queries to $U$ and is computable in \pup (provided that $U$ is computable in \pup). 

The robustness of the Quantum Singular Value Transform~\cite[Lemma 2.4.3]{gilyen2019quantum} implies that $ \| T_k(A/\alpha) - T_k(\tilde{A}) \|_\infty \leq 4 k \sqrt{ \| A/\alpha - \tilde{A} \|_\infty } \leq 4k \sqrt{\eps/\alpha}$.
\end{proof}

\begin{lemma}[Arbitrary polynomials of block encodings] \label{lem:be-polys}
Let $k = 2^{\poly(n)}$ and let $P(x) = \sum_{j=1}^k c_j T_j(x)$ denote a linear combination of Chebyshev polynomials with a \pspace-computable coefficient vector $c = (c_j)_{j = 1,\ldots,k}$. Let $A \in \linear(\C^D)$ be a Hermitian matrix with an $(\alpha,\epsilon,a)$-block encoding $U$ computable in \pup. Then for all polynomials $q(n)$ there exists a $(\| c \|_1 , \| c \|_1 (2^{-q(n)} + 4 k \sqrt{\eps/\alpha}), a+\poly(n))$-block encoding $V$ of $P(A/\alpha)$ that is computable in \pup.
\end{lemma}
\begin{proof}
This follows directly from \Cref{lem:chebyshev-block} and \Cref{lem:be-lin-comb}. 
\end{proof}

In particular, we can apply \cref{lem:be-polys} to the polynomial approximation of the sign function that we constructed in \cref{lem:sign-approx}.

\begin{lemma}[Sign polynomial applied to block encodings] \label{lem:be-sign-poly}
Let $A \in \linear(\C^D)$ be a Hermitian matrix with an $(\alpha,\epsilon, a)$-block encoding computable in \pup with $\alpha \leq 2^{\poly(n)}$.
Then, for all $d \leq 2^{\poly(n)}$ and all polynomials $q(n)$, there exists a $(\alpha', \alpha'(2^{-q(n)} + 4 d \sqrt{\eps/\alpha}), a+\poly(n))$-block encoding of $P_d^{\sgn}(A/\alpha)$ in \pup, where $P_d^{\sgn}$ is the polynomial approximation of the sign function from \cref{lem:sign-approx} and $\alpha' \leq O(\log d) = \poly(n)$ is \pspace-computable (given $d$).
\end{lemma}
\begin{proof}
By \cref{lem:sign-approx}, we can write $P_d^{\sgn}(x) = \sum_{i = 0}^d c_{i} T_{i}(x)$ as a linear combination of Chebchev polynomials with \pspace-computable coefficients $c_{i}$.
Furthermore, \cref{lem:sign-approx} shows that the coefficient vector $c = (c_1,\dots, c_d)$ satisfies $\norm{c}_1 = O(\log d) = \poly(n)$.
Then, we get from \cref{lem:be-polys} that for all polynomials $q(n)$, $P_d^{\sgn}(A/\alpha)$ has a $(\alpha', \alpha'2^{-q(n)} + 4 d \alpha' \sqrt{\eps/\alpha}, a+\poly(n))$-block encoding in \pup with \pspace-computable $\alpha = \norm{c}_1 = O(\log d)$.
\end{proof}

By the same argument as for \cref{lem:be-sign-poly}, we can also apply the square root polynomial from \cref{lem:sqrt-approx} to a block encoding.
\begin{lemma}[Square root polynomial applied to block encodings] \label{lem:be-sqrt}
Let $A \in \linear(\C^D)$ be a Hermitian matrix with an $(\alpha,\epsilon,a)$-block encoding computable in \pup with $\alpha \leq 2^{\poly(n)}$.
Then, for all $d \leq 2^{\poly(n)}$ and all polynomials $q(n)$, there exists a $(\alpha', \alpha'2^{-q(n)} + 4 d \alpha' \sqrt{\eps/\alpha}, a+\poly(n))$-block encoding of $P_d^{\sqrt{~}}(A/\alpha)$ in \pup, where $P_d^{\sqrt{~}}$ is the polynomial approximation of the square root function from \cref{lem:sqrt-approx} and $\alpha' = O(d)$ is \pspace-computable (given $d$).
\end{lemma}
\begin{proof}
The proof is identical to that of \cref{lem:be-sign-poly}, with the only exception that here we use the trivial bound $c_i \leq O(1)$ for the coefficients of the polynomial $P_d^{\sqrt{~}}$ from \cref{lem:sqrt-approx}, leading to a \pspace-computable $\alpha' \leq O(d)$ (rather than $\alpha' \leq O(\log d)$ as in \cref{lem:be-sign-poly}; one can prove logarithmic scaling with a little extra work in this case, too, but we will not need it).
\end{proof}

From this, we can obtain the following corollary for square roots of block encoded \emph{positive semidefinite} matrices.
\begin{corollary} \label{lem:sqrt_be_simplified}
Let $A \in \linear(\C^D)$ be a positive semidefinite matrix with an $(\alpha,\epsilon,a)$-block encoding computable in \pup with \pspace-computable $\alpha \leq 2^{\poly(n)}$ and any $\eps \geq 2^{-\poly(n)}$.
Then, there exists an $(\alpha', c \sqrt{\eps}, a+\poly(n))$-block encoding of $\sqrt{A}$ in \pup for some \pspace-computable $\alpha', c \leq 2^{\poly(n)}$.
\end{corollary}
\begin{proof}
By assumption we have an $(\alpha, \eps,a)$-block encoding of $A$, and we trivially have a $(1,0,0)$-block encoding of $\Id$.
Therefore, using \cref{lem:be-lin-comb} (with a sufficiently large $q(n)$) there exists a $(2\alpha, 4 \alpha \eps, a+\poly(n))$-block encoding of $B \deq A - \alpha \Id$ in \pup.
Since $A$ has an $(\alpha, \eps,a)$-block encoding, $\norm{A}_\infty \leq \alpha$, so $B/\alpha$ has eigenvalues in the interval $[-1,1]$.
Now choose $d = O(\log(\eps)/\eps)$ large enough such that the polynomial $P^{\sqrt{~}}_d$ from \cref{lem:sqrt-approx} satisfies $\left|\sqrt{\frac{x+1}{2}} - P_{d}^{\sqrt{~}}(x)\right| \leq \sqrt\eps$ for all $x \in [-1, 1]$.
With this choice, 
\begin{align}
\norm{\sqrt{\frac{A}{2\alpha}} - P_{d}^{\sqrt{~}}(B/\alpha)}_\infty = \norm{\sqrt{\frac{B/\alpha+\Id}{2}} - P_{d}^{\sqrt{~}}(B/\alpha)}_\infty \leq \eps \,. \label{eqn:poly_sqrt_close}
\end{align}
Using \cref{lem:be-sqrt}, we have an $(\tilde\alpha, \tilde\alpha2^{-q(n)} + 8 d \tilde\alpha \sqrt{\eps}, a+\poly(n))$-block encoding of $P_d^{\sqrt{~}}(B/\alpha)$ in \pup for \pspace-computable $\tilde\alpha \leq 2^\poly(n)$ and all polynomials $q(n)$.
With a sufficiently large choice of $q(n)$ and using \cref{eqn:poly_sqrt_close}, this is also an $(\tilde\alpha, (10 d \tilde\alpha + 1) \sqrt{\eps}, a+\poly(n))$-block encoding of $\sqrt{A}/\sqrt{2\alpha}$.
By \cref{def:be}, this is equivalently a $(\sqrt{2\alpha}\tilde\alpha, (10 d \tilde\alpha + 1) \sqrt{2\alpha} \sqrt{\eps}, a+\poly(n))$-block encoding of $\sqrt{A}$.
Defining $\alpha' = \sqrt{2\alpha}\tilde\alpha$ and $c = (10 d \tilde\alpha + 1) \sqrt{2\alpha}$ completes the proof.
\end{proof}

\subsection{Preparing the Gibbs state of a block-encoded Hamiltonian}
For our implementation of the SDP solver in \cref{sec:exp_sdp}, we will need to prepare (block encodings of) Gibbs states, i.e.~states of the form $\exp(-\beta H)/\tr{\exp(-\beta H)}$ for some Hamiltonian $H$ that is provided in the form of a block encoding.
The following lemma shows how to approximate the exponential of a block encoding by means of a Taylor expansion.

\begin{lemma}[Exponential function applied to block encoding]
\label{lem:be-exp}
Let $A \in \linear(\C^D)$ be Hermitian with $\norm{A}_\infty = \poly(n)$ and an $(\alpha, \eps, a)$-block encoding $U$ in \pup.
Let $\beta \leq \poly(n)$.
Then, for all polynomials  $k(n)$, there exists an $(\alpha', 4 k \alpha' \sqrt{\eps/\alpha} + 2 e^{2e \beta \norm{A}_\infty - k}, a+\poly(n))$-block encoding of $e^{\beta A}$ in \pup for some \pspace-computable $\alpha'(\alpha, \beta, k) \leq (\alpha \beta)^k$.
\end{lemma}
\begin{proof}
Let $m = \norm{A}_\infty = \poly(n)$.
On the interval $[-m/\alpha,m/\alpha]$, taking the Taylor expansion of degree $k$ of $x \mapsto \exp(\beta \alpha x)$ yields a polynomial $P_k$ that satisfies $|P_k(x) - \exp(\beta \alpha x)| \leq e^{2e \beta m - k}$.
We can express $P_k$ as a linear combination of Chebyshev polynomials of degree up to $k$: $P_k = \sum_{j = 0}^k c_j T_j$.
For $k \leq \poly(n)$, it is easy to verify that $c_j$ is non-negative and can be computed in \pspace.\footnote{In fact, $c_j$ can be computed in $\class{P}$. One explicit way of doing this is as follows: compute the (non-negative) coefficients of the truncated Taylor series of $\exp(\beta \alpha \cdot)$ and express each monomial as a sum of Chebyshev polynomials with non-negative coefficients using~\cite[Theorem 3.1]{SV14}; then combine terms. Since the degree is polynomial, there are at most a polynomial number of terms, and the coefficients for each term can be computed in $\class{P}$.}
Then, $\alpha' \deq \sum |c_j|$ is \pspace-computable, too, and we can bound $\alpha' \leq k (\alpha\beta)^k$.

By \cref{lem:be-polys}, for all polynomials $q(n)$ we get a \pup-computable $(\alpha', \alpha' 2^{-q(n)} + 4 k \alpha' \sqrt{\eps/\alpha}, a+\poly(n))$-block encoding of $P_k(A/\alpha)$, which, by the triangle inequality, is also a $(\alpha', \alpha' 2^{-q(n)} + 4 k \alpha' \sqrt{\eps/\alpha} + e^{2e \beta m - k}, a+\poly(n))$-block encoding of $\exp(\beta A)$.
Choosing a sufficiently large $q(n)$ such that $\alpha' 2^{-q(n)} \leq e^{2e \beta m - k}$, we obtain the desired result.
\end{proof}

\cref{lem:be-exp} shows how to apply the exponential function to a block encoding, but we do not necessarily know the normalisation factor $\tr{\exp(-\beta H)}$ that we need to use for the Gibbs state.
Instead of computing this factor explicitly, we can prepare a block-encoding of $\ketbra{\Gamma}{0\dots0}$, where $\ket{\Gamma}$ is a (normalised) purification of the Gibbs state.
To this, we can apply the amplitude amplification scheme from \cref{lem:be-fixedpoint} to set the post-selection factor to $\alpha = 1$.
The advantage of this approach is that to apply \cref{lem:be-fixedpoint}, we only need a crude bound on $\tr{\exp(-\beta H)}$ to determine the minimum number of iterations in the amplitude amplification scheme, but we do not need a precise estimate of $\tr{\exp(-\beta H)}$ as doing too many iterations in the amplitude amplification scheme does not do any harm.
The output of the amplitude amplification is a block encoding of a normalised purification of the Gibbs state with $\alpha = 1$; in \cref{lem:be-gibbs-mixed}, we use this block encoding to prepare the desired (mixed) Gibbs state with the correct normalisation $\tr{\exp(-\beta H)}$ and post-selection factor $\alpha = 1$. 
\begin{lemma}[Preparing a purified Gibbs state]
\label{lem:be-gibbs-pure}
Let $A \in \linear(\C^D)$ be Hermitian on $n = \log D$ qubits with $\norm{A}_\infty \leq p(n)$ for some polynomial $p(n)$ and an $(\alpha, \eps, a)$-block encoding $U$ in \pup.
Let $\beta \leq \poly(n)$. Then for all polynomials $q(n)$ there exists a $(1,c \sqrt{\eps} + 2^{-q(n)}, a+\poly(n))$-block encoding $V$ of the matrix $\ketbra{\Gamma}{0^{2n}}$, where $\ket{\Gamma}$ is a (normalised) $(2n)$-qubit purification of the Gibbs state $\frac{e^{\beta A}}{\Tr(e^{\beta A})}$.
Here, $c(\alpha, \beta, p) \leq 2^{\poly(n)}$ is a \pspace-computable function.
\end{lemma}
\begin{proof}
Choose a \pspace-computable $k(n) = \poly(n)$ such that $2 e^{e \beta p(n) - k} \leq 2^{-q(n)-2} e^{-\beta p(n)/2}$.
Then \Cref{lem:be-exp} yields an $\left(\alpha', \eps', a+\poly(n) \right)$-block encoding of $e^{\beta A/2}$ for a \pspace-computable $\alpha'(\alpha, \beta, k)$ and $\eps' = \frac{4 k \alpha'}{\alpha} \sqrt{\eps} + 2^{-q(n)-2} e^{-\beta p(n)/2}$.
By \Cref{lem:be-purification} there exists a $\left(\alpha', \eps', a+\poly(n) \right)$-block encoding $V$ of the matrix $(e^{\beta A/2} \otimes I) \ketbra{\Phi}{0 \cdots 0}$, where $\ket{\Phi}$ is the maximally entangled state on $\C^D \ot \C^D$. 

Define the normalised state 
\begin{align*}
\ket{\Gamma} = \frac{(e^{\beta A/2} \otimes I) \ket{\Phi}}{\norm{(e^{\beta A/2} \otimes I) \ket{\Phi}}_2} = \frac{(e^{\beta A/2} \otimes I) \ket{\Phi}}{\sqrt{\tr{e^{\beta A}}/D}} \,.
\end{align*}
$\ket{\Gamma}$ is a purification of the (normalised) Gibbs state $e^{\beta A}/\Tr(e^{\beta A})$.
We can rescale our block encoding by a factor $\nu \deq \sqrt{D/\tr{e^{\beta A}}}$ and get that $V$ is equivalently a $\left(\alpha' \nu, \eps' \nu, a+\poly(n) \right)$-block encoding of $\ketbra{\Gamma}{0 \cdots 0}$.
Since $\norm{A}_\infty \leq p(n)$, we have that $\Tr(e^{\beta A}) \geq D e^{-\beta p(n)}$.
Therefore, we can simplify the error terms and find that $V$ is an $\left(\alpha'\nu, \eps'' + 2^{-q(n)-2}, a + \poly(n) \right)$-block encoding of $\ketbra{\Gamma}{0\dots 0}$ for $\eps'' = \frac{4 k \alpha' e^{\beta p(n)/2}}{\alpha} \sqrt{\eps} \geq \eps' \nu$.

We can now apply \cref{lem:be-fixedpoint} with $\gamma = \alpha' e^{\beta p(n)/2} \geq \alpha'\sqrt{D/\tr{e^{\beta A}}}$ to obtain a $(1, 2 \eps'' + 2^{-q(n)}, a+\poly(n))$-block encoding of $\ketbra{\Gamma}{0 \cdots 0}$ in \pup.
To conclude, we define $c(\alpha, \beta, k, p) \deq \frac{8 k \alpha' e^{\beta p(n)/2}}{\alpha}$. Since $\alpha'(\alpha, \beta, k)$ is a \pspace-computable function, so is $c$, concluding the proof.
\end{proof}

\begin{corollary}[Preparing a mixed Gibbs state] \label{lem:be-gibbs-mixed}
Let $A \in \linear(\C^D)$ be Hermitian on $n = \log D$ qubits with $\norm{A}_\infty \leq p(n)$ for some polynomial $p(n)$ and an $(\alpha, \eps, a)$-block encoding $U$ in \pup.
Let $\beta \leq \poly(n)$. Then for all polynomials $q(n)$ there exists a $(1,c \sqrt{\eps} + 2^{-q(n)}, a+\poly(n))$-block encoding of the Gibbs state $\frac{e^{\beta A}}{\Tr(e^{\beta A})}$.
Here, $c(\alpha, \beta, p) \leq 2^{\poly(n)}$ is a \pspace-computable function.
\end{corollary}
\begin{proof}
By \cref{lem:be-gibbs-pure}, there exists a \pup-computable $(1, c\sqrt{\eps}, a+\poly(n))$-block encoding $V$ of $\ketbra{\Gamma}{0^{2n}}$ with a \pspace-computable $c(\alpha, \beta, k, p)$, where $\ket{\Gamma}$ is a purification of the Gibbs state $\frac{e^{\beta A}}{\Tr(e^{\beta A})}$
The corollary then follows from \cref{lem:be-reduced-density}.
\end{proof}

The following corollary shows that we can also ``extract'' a purification of the Gibbs state (as a quantum state, not inside a unitary block encoding) from \cref{lem:be-gibbs-pure}.
\begin{corollary}[Extracting a Gibbs state from a block encoding] \label{lem:gibbs-extract}
Let $A \in \linear(\C^D)$ be Hermitian on $n$ qubits with $\norm{A}_\infty = p(n)$ for some polynomial $p(n)$ and an $(\alpha, \eps, a)$-block encoding $U$ in \pup.
Let $\beta \leq \poly(n)$. Then for all polynomials $q(n)$ there exists a unitary $V$ in \pup such that $\norm{V\proj{0^{2n + b}}V^\dagger - \proj{\Gamma}\ot\proj{0^b}}_\infty \leq 2 c \sqrt{\eps} + 2^{-q(n)}$, where $\ket{\Gamma}$ is a $(2n)$-qubit purification of the Gibbs state $\frac{e^{\beta A}}{\Tr(e^{\beta A})}$, $b = a + \poly(n)$, and $c(\alpha, \beta, p) \leq 2^{\poly(n)}$ is a \pspace-computable function.
\end{corollary}
\begin{proof}
By \cref{lem:be-gibbs-pure}, there exists a $(1,c \sqrt{\eps} + 2^{-q(n)}, b)$-block encoding $V$ of the matrix $\ketbra{\Gamma}{0^{2n}}$, where $\ket{\Gamma}$ is a (normalised) purification of the Gibbs state $\frac{e^{\beta A}}{\Tr(e^{\beta A})}$ and $b = a + \poly(n)$.
The result then follows straightforwardly from \cref{def:be} and the triangle inequality.
\end{proof}

\section{Solving exponentially large SDPs in polynomial space} \label{sec:exp_sdp}

We will use the block encoding framework introduced in the previous section to solve SDP feasibility problems, which are specified by a Hermiticity-preserving linear map $\Phi: \linear(\cH) \to \linear(\cH)$ and a Hermitian matrix $B \in \linear(\cH)$; the goal is to find a positive semidefinite $X$ such that $\Phi(X) = B$.
In particular, we are interested in the case where the underlying Hilbert space $\cH$ has exponentially large dimension, but the instance $(\Phi,B)$ (a) has small-width and (b) is \pspace-computable. The small-width condition means that $\|B\|_\infty \leq 1$ and $\| \Phi^*(Y) \|_\infty \leq \| Y\|_\infty$. \pspace-computability of $\Phi$ and $B$ are as defined in \cref{def:pspace_comp}.
For this case we construct a quantum polynomial-space algorithm that prepares a quantum state $\rho$ whose density matrix is an approximately feasible point, i.e.~$\| \Phi(\rho)- B \|_1 \approx 0$.
In \cref{sec:state_classes}, we then use this algorithm as a quantum polynomial-space procedure to prepare the output states of $\stateQIP$-protocols by expressing such protocols as a feasibility SDP.

\subsection{Multiplicative weights update algorithm for the feasibility SDP} \label{sec:mmwu}

Our SDP feasibility algorithm will be based on the matrix multiplicative weights update (MMWU) framwork~\cite{arora2016combinatorial,kale2007efficient}.
In this section, we describe the algorithm at an abstract level using subroutines $\GibbsOracle$ and $\TraceDistOracle$, and prove that if these subroutines satisfy certain properties (see \cref{def:good_oracles}), then the algorithm is correct, i.e~produces a valid (approximately) feasible output. 
In \cref{sec:solving}, we will show how to implement these subroutines, as well as the algorithm as a whole, using the block encoding framework.

Let $B$ denote a $D \times D$ Hermitian matrix. Let $\Phi$ denote a Hermiticity-preserving linear map on the space of $D \times D$ operators.\footnote{More generally, $\Phi$ could be a linear map $\linear(\C^{D'}) \to \linear(\C^{D})$, i.e.~the input and output dimensions need not match. Indeed, this will be the case when applying our SDP algorithm in \cref{sec:state_classes}. However, we can always make the input and output dimensions match by padding, which we will do implicitly whenever necessary.} 
Suppose that $(\Phi,B)$ has small width.
The goal is to identify an \emph{approximately feasible} point of the following set: 
\[
    \{ \rho \text{ density matrix} : \Phi(\rho) = B \}~.
\]
We say that a density matrix $\rho$ is \emph{$\epsilon$-feasible for $(\Phi,B)$} if
\[
    \| \Phi(\rho) - B \|_1 \leq \epsilon
\]
where $\| \cdot \|_1$ denotes the trace norm.

Our \cref{algo:mmwu} will be able to handle not just the case where the feasibility SDP $(\Phi, B)$ has an exactly feasible point (that we are trying to approximate), but also the case where only approximately feasible solutions exist.
For this, it will be convenient to consider the following optimisation problem, which computes the minimum approximation error that we need to allow for the approximately feasible set to be non-empty:
\begin{equation}
    \label{eq:min-max}
\beta = \min_\rho \max_{H : \| H \|_\infty \leq 1} \langle \Phi(\rho) - B, H \rangle
\end{equation}
where the minimization in the first optimization problem is over all density matrices $\rho$ and the maximization is over all Hermitian matrices $H$ with spectral norm at most $1$. For fixed $\rho$, the inner maximization $\max_{H: \|H\|_\infty \leq 1} \langle \Phi(\rho) - B, H \rangle$ is simply the trace norm of $\| \Phi(\rho) - B \|_1$; here we use the variational characterization of the trace norm.
Thus, the solution to \Cref{eq:min-max} is the minimum $\beta$ such that there exists a $\beta$-feasible density matrix $\rho$ for $(\Phi,B)$. If $S$ is non-empty, then $\beta = 0$. 

We then consider \cref{algo:mmwu}, which takes as input the instance $(\Phi, B)$ and a parameter $\epsilon > 0$ and is meant to output an approximately point encoded as a quantum state.
This algorithm uses subroutines $\TraceDistOracle$ and $\GibbsOracle$, which we describe in more detail below.
The approximation error of the output will depend both on the parameter $\eps$ as well as the solution to \cref{eq:min-max} for $(\Phi, B)$ and the ``goodness'' of the subroutines  $\TraceDistOracle$ and $\GibbsOracle$ (see \cref{def:good_oracles}).

\begin{longfbox}[breakable=false, padding=1em, margin-top=1em, margin-bottom=1em]
\begin{algorithm} {\bf Matrix multiplicative weights algorithm for general SDPs} \label{algo:mmwu} \end{algorithm}
\noindent\underline{Input}  \vspace{-0.8ex}
\begin{itemize}
    \item SDP described by $(\Phi, B)$, where $B \in \linear(\C^{D})$.
    \item Parameter $\eps > 0$.
    \item Trace distance oracle $\TraceDistOracle$.
    \item Gibbs oracle $\GibbsOracle$.
\end{itemize}

\noindent\underline{Algorithm}  \vspace{-0.8ex}
\begin{enumerate}[label=\arabic*.]
    \item $T \leftarrow \lceil \frac{\ln D}{\epsilon^2} \rceil$. 
    \item $\rho_1 \leftarrow I/D$.
    \item For $t = 1,\ldots,T$:
    \begin{enumerate}
        \item $H_t \leftarrow \TraceDistOracle(\Phi(\rho_t) - B)$.
        \item $\rho_{t+1} \leftarrow \GibbsOracle(\Phi^*(H_1 + \cdots + H_t), \epsilon)$.
    \end{enumerate}
    \item Output $\rho \deq \frac{1}{T} \sum_{t=1}^T \rho_t$.
\end{enumerate}
\end{longfbox}

To analyse the correctness of \cref{algo:mmwu}, we need to specify what guarantees the subroutines $\TraceDistOracle$ and $\GibbsOracle$ are required to satisfy. They are procedures that take as input an operator in $\linear(\C^D)$ and output an operator in $\linear(\C^D)$. Intuitively, $\TraceDistOracle$ is meant to approximate the optimal $H$ in \cref{eq:min-max} and $\GibbsOracle$ is meant to approximate the Gibbs state of a given Hamiltonian. 

\begin{definition}[Good trace distance and Gibbs oracles] \label{def:good_oracles}
We call a trace distance oracle $\TraceDistOracle$ $(C, \delta)$-good if for all Hermitian operators $M \in \linear(\C^D)$ with $\norm{M}_\infty \leq C$, we have $$\norm{\TraceDistOracle(M)}_\infty \leq 2 \quad \text{and} \quad | \langle \TraceDistOracle(M), M\rangle - \norm{M}_1| \leq \delta.$$
Similarly, we call a Gibbs oracle $\GibbsOracle$ $(C, \delta)$-good if for all operators $M$ with $\norm{M}_\infty \leq C$ and $|\beta| \leq 1$, we have $$\norm{\GibbsOracle(M,\beta) - \frac{\exp(-\beta M)}{\Tr(\exp(-\beta M))} }_1 \leq \delta.$$
\end{definition} 

With this, we can prove that \cref{algo:mmwu} is correct, i.e.~outputs an approximately feasible solution.
This is formalised by the following theorem.
\begin{theorem}
\label{thm:mmwu}
Let $\Phi: \linear(\C^D) \to \linear(\C^D)$ denote a Hermiticity-preserving superoperator and let $B \in \linear(\C^D)$, and suppose that $(\Phi,B)$ has small width. 
Let $\eps, \delta > 0$.
Suppose we run \cref{algo:mmwu} with a $(2, \delta)$-good trace distance oracle and a $(\lceil 2 \ln D/\epsilon^2 \rceil,\delta)$-good Gibbs oracle.
Then, the output $\rho$ of \cref{algo:mmwu} is a $(2 \beta + 11\epsilon + 2\delta)$-feasible density matrix, where $\beta$ is the solution to \Cref{eq:min-max}. 
\end{theorem}

Note that if the feasible region for $(\Phi, B)$ is nonempty, then $\beta = 0$ and the algorithm outputs a $O(\epsilon)$-feasible solution. 
This theorem relies on the following lemma, which is proved in many forms in many papers that use the matrix multiplicative weights update framework; we include a full proof here for the sake of completeness.
\begin{lemma}
\label{lem:mmwu}
Running \cref{algo:mmwu} with the parameters specified in \cref{thm:mmwu}, the sequence of density matrices $\rho_1,\ldots,\rho_T$ generated by the algorithm satisfies
\begin{equation}
\label{eq:mmwu}
    \lambda_{\min}(\Phi^*(H_1 + \cdots + H_T)) \geq \left( \sum_{t=1}^T \langle \rho_t , \Phi^*(H_t) \rangle  \right) - \frac{\ln D}{\epsilon} - 2T (\epsilon + \delta + \sinh(2\epsilon))~.
\end{equation}
\end{lemma}
We defer the proof of this lemma and first show how \Cref{thm:mmwu} follows from \Cref{lem:mmwu}.

\begin{proof}[Proof of \Cref{thm:mmwu}]
Let $\hat{\rho}$ denote a density matrix achieving the minimum of~\Cref{eq:min-max}. We then have
\begin{align*}
    \lambda_{\min}(\Phi^*(H_1 + \cdots + H_T)) - \sum_{t=1}^T \langle B, H_t \rangle &\leq \sum_{t=1}^T \langle \hat{\rho}, \Phi^*(H_t) \rangle - \sum_{t=1}^T \langle B, H_t \rangle \\
    &= \sum_{t=1}^T \langle \Phi(\hat{\rho}) - B, H_t \rangle \\
    &= \sum_{t=1}^T \| H_t \|_\infty \cdot \Big \langle \Phi(\hat{\rho}) - B, \frac{H_t}{\| H_t \|_\infty} \Big \rangle  \\
    &\leq 2 T \beta~.
\end{align*}
The last step uses that $\norm{H_t}_\infty \leq 2$ (from the trace distance oracle normalisation condition), and that $\frac{H_t}{\| H_t \|_\infty}$ is a feasible choice for \cref{eq:min-max}, so that $\Big \langle \Phi(\hat{\rho}) - B, \frac{H_t}{\| H_t \|_\infty} \Big \rangle \leq \norm{\Phi(\hat{\rho}) - B}_1 = \beta$.

We now put this together with \Cref{lem:mmwu}. Subtracting $\sum_{t=1}^T \langle B, H_t \rangle$ from both sides of \Cref{eq:mmwu} and re-arranging, we get
\begin{align*}
\sum_{t=1}^T \langle \Phi(\rho_t) - B , H_t \rangle &\leq \lambda_{\min}(\Phi^*(H_1 + \cdots + H_T)) - \sum_{t=1}^T \langle B, H_t \rangle + \frac{\ln D}{\epsilon}+ 2T (\epsilon + \delta + \sinh(2\epsilon)) \\
&\leq 2 T \beta + \frac{\ln D}{\epsilon}+ 2T (\epsilon + \delta + \sinh(2\epsilon))~.
\end{align*}
Since the trace distance oracle is $(2,\delta)$-good and by the norm conditions on $\Phi^*$ and $B$,
\begin{align*}
    \| \Phi(\rho_t) - B \|_\infty &= \sup_{\ket{v}} \Tr \Big( \ketbra{v}{v} \, \Big( \Phi(\rho_t) - B \Big) \Big) \\
    &\leq \sup_{\ket{v}} \Tr \Big( \Phi^*(\ketbra{v}{v}) \, \rho_t \Big) + \Big| \Tr \Big( \ketbra{v}{v} \, B \Big) \Big | \\
    &\leq \sup_{\ket{v}} \| \Phi^*(\ketbra{v}{v}) \|_\infty + \| B \|_\infty \leq 2 \,,
\end{align*}
this implies that
\[
    \langle \Phi(\rho_t) - B, H_t \rangle \geq \| \Phi(\rho_t) - B \|_1 - \delta~.
\]
Thus we get that
\begin{align*}
\frac{1}{T} \sum_{t=1}^T \| \Phi(\rho_t) - B \|_1 
&\leq \frac{1}{T} \sum_{t=1}^T \Big( \langle \Phi(\rho_t) - B , H_t \rangle  + \delta \Big) \\
&\leq 2 \beta + \frac{\ln D}{T \eps} + 2 (\eps + \delta + \sinh(2\eps)) \\
&\leq 2 \beta + 11\epsilon + 2\delta
\end{align*}
where in the last line we used that $2\sinh(2\eps) \leq 8\eps$ for all $0 \leq \epsilon \leq 1$ and $T \geq \frac{\ln D}{\eps^2}$.

Let $\rho = \frac{1}{T}\sum_t \rho_t$ denote the output of the algorithm. Then $\rho$ is a $(2 \beta + 11\epsilon + 2\delta )$-feasible point for the set $S$, because
\begin{align*}
    \| \Phi(\rho) - B \|_1 = \Big \| \frac{1}{T} \sum_{t=1}^T \Phi(\rho_t) - B \Big \|_1 
    \leq \frac{1}{T} \sum_{t=1}^T \| \Phi(\rho_t) - B \|_1~.
\end{align*}
This concludes the proof of \Cref{thm:mmwu}.
\end{proof}

We now show \Cref{lem:mmwu}.
\begin{proof}[Proof of \Cref{lem:mmwu}]
First we argue that the small width property implies $\| \Phi(\rho_t) - B \|_\infty \leq 2$. This is because $\|B\|_\infty \leq 1$ and
\begin{align*}
    \| \Phi(\rho_t) \|_\infty &= \max_{A : \|A\|_1 \leq 1} \langle \Phi(\rho_t), A \rangle = \max_{A : \|A\|_1 \leq 1} \langle \rho_t, \Phi^*(A) \rangle \\
    &\leq \Tr(\rho_t) \cdot \max_{A : \|A\|_1 \leq 1} \| \Phi^*(A) \|_\infty \leq 1
\end{align*}
where in the last line we used the fact that $\|A\|_\infty \leq \|A\|_1 \leq 1$ and our norm assumption on $\Phi^*$. Combined with the fact that the trace distance oracle is $(2,\delta)$-good, we have $\| H_t \|_\infty \leq 2$. For all $t = 1,\ldots,T$ define $X_{t+1} = \exp(-\epsilon \Phi^*(H_1 + \cdots + H_t))$
and define $\tilde{\rho}_t = X_t /\Tr(X_t)$. Since the Gibbs oracle is $(2T,\delta)$-good and $\eps \| \Phi^* (H_1 + \cdots + H_t) \|_\infty \leq \sum_{j = 1}^t \| \Phi^*(H_j)\|_\infty \leq 2 T$
this means that $\| \rho_t - \tilde{\rho}_t \|_1 \leq \delta$.

By the Golden-Thompson inequality, we have
\begin{align*}
    \Tr(X_{t+1}) &\leq \Tr(X_t) \, \langle \tilde{\rho}_t, \exp( -\epsilon \, \Phi^*(H_t)) \rangle 
\end{align*}
for all $t$.
Therefore, since $\tr{X_1} = \tr{\1_D} = D$,
\begin{align}
    \Tr(X_{T+1}) &\leq D \, \prod_{t=1}^{T} \langle \tilde{\rho}_t, \exp( -\epsilon \, \Phi^*(H_t)) \rangle~. \label{eqn:mmwu1}
\end{align}
By \cite[Lemma 4.8]{vidick2016quantum}, since $\| \Phi^*(H_t) \|_\infty \leq 2$ (which follows from the small width property), we have the inequality 
\begin{align}
\langle \tilde{\rho}_t, \exp( -\epsilon \, \Phi^*(H_t)) \rangle \leq \exp \Big( -\epsilon \, \exp(-2\epsilon) \, \langle \tilde{\rho}_t, \Phi^*(H_t) \rangle \Big) \, \exp( 2\epsilon \, \sinh(2\epsilon))~. \label{eqn:mmwu2}
\end{align}
On the other hand, we have that
\begin{align}
\Tr(X_{T+1}) \geq \lambda_{\max} \Big(\exp(-\epsilon \, \Phi^*(H_1 + \cdots + H_T) \Big) = \exp \Big( - \epsilon \, \lambda_{\min}(\Phi^*(H_1 + \cdots + H_T)) \Big)~. \label{eqn:mmwu3}
\end{align}
Combining \cref{eqn:mmwu1}, \cref{eqn:mmwu2}, and \cref{eqn:mmwu3} we get
\[
    \exp \Big( - \epsilon \, \lambda_{\min}(\Phi^*(H_1 + \cdots + H_T)) \Big) \leq D \, \prod_{t=1}^{T} \Big [\exp \Big( -\epsilon \, \exp(-2\epsilon) \, \langle \tilde{\rho}_t, \Phi^*(H_t) \rangle \Big) \, \exp( 2\epsilon \, \sinh(2\epsilon)) \Big ]~.
\]
Taking natural logarithms of both sides we get
\[
    \lambda_{\min}(\Phi^*(H_1 + \cdots + H_T)) \geq \exp(-2\epsilon) \, \left( \sum_{t=1}^T \langle \tilde{\rho}_t, \Phi^*(H_t) \rangle  \right) - \frac{\ln D}{\epsilon} - 2 T\sinh(2\epsilon)~.
\]
We conclude by observing that 
\begin{align*} \exp(-2\epsilon) \, \sum_{t=1}^T \langle \tilde{\rho}_t, \Phi^*(H_t) \rangle &\geq (1 - 2\epsilon) \sum_{t=1}^T \langle \tilde{\rho}_t, \Phi^*(H_t) \rangle \\
&\geq \left( \sum_{t=1}^T \langle \tilde{\rho}_t, \Phi^*(H_t) \rangle \right) - 2\epsilon T  \\
&\geq \left( \sum_{t=1}^T \langle \rho_t, \Phi^*(H_t) \rangle  \right) - (2\epsilon + \delta)T
\end{align*}
where the second line comes from
\[
    \Big | \sum_{t=1}^T \langle \tilde{\rho}_t, \Phi^*(H_t) \rangle\Big | \leq \sum_{t=1}^T \Big | \langle \tilde{\rho}_t, \Phi^*(H_t) \rangle \Big | \leq T
\]
and the third line comes from the fact that $\| \rho_t - \tilde{\rho}_t \|_1 \leq \delta$. 
\end{proof}

\subsection{Implementing the MMWU algorithm using block encodings}
\label{sec:solving}

Having established the correctness of \cref{algo:mmwu}, we need to show that if $\Phi$ and $B$ are \pspace-computable, we can actually implement the algorithm in quantum polynomial space.
For this, we will make use of the block encoding framework.
As a first step, we show that we can implement ``good'' trace distance and Gibbs oracles on block encodings.

\subsubsection{Implementing the trace distance and Gibbs oracles}

\begin{lemma}[Trace distance oracle] \label{lem:poly_as_tdoracle}
Let $D \in \N$, $C \geq 1$, and $\delta > 0$.
Let $\kappa = \frac{\delta}{6DC}$. 
Then, for sufficiently large $d = O\left(\frac{\log(1/\kappa)}{\kappa}\right)$, the function $\TraceDistOracle(A) \deq P_{d}^{\sgn} (A/C)$ is a $(C,\delta)$-good trace distance oracle for operators $A \in \linear(\C^D)$, where $P_{d}^{\sgn}$ is the polynomial defined in \cref{lem:sign-approx}.
\end{lemma}
\begin{proof}
Let $A \in \linear(\C^D)$ denote a Hermitian matrix with spectral norm at most $C$. The operator $A/C$ has spectral norm at most $1$, and therefore by \cref{lem:sign-approx} the operator $P_{d}^{\sgn}(A/C)$ has spectral norm at most $2$.
This establishes the first property of a trace distance oracle as stated in \cref{def:good_oracles}.

Let $A = \sum_{j=1}^D \lambda_j \ketbra{v_j}{v_j}$ denote the spectral decomposition of $A$. Let $S = \{ j : |\lambda_j/C| \leq \kappa \}$ (i.e.~the set of ``small'' eigenvalues) and let $L$ denote the complement of $S$ (i.e.~the set of ``large'' eigenvalues). Then 
\begin{align*}
    \left|\langle P_{d}^{\sgn}(A/C) , A \rangle - \norm{A}_1\right| 
    &= \left| \sum_{j=1}^D P_{d}^{\sgn}(\lambda_j/C) \cdot \lambda_j - |\lambda_j| \right|\\
    &\leq \sum_{j=1}^D \left| (P_{d}^{\sgn}(\lambda_j/C) - \sgn(\lambda_j/C)) \cdot \lambda_j \right|\\
    &= \sum_{j \in S} \left| P_{d}^{\sgn}(\lambda_j/C) - \sgn(\lambda_j/C)\right| \cdot |\lambda_j| + \sum_{j \in L} \left| P_{d}^{\sgn}(\lambda_j/C) - \sgn(\lambda_j/C)\right| \cdot |\lambda_j|\,.
\end{align*}
We bound the terms separately. First, for the small eigenvalues we have
\[
    \sum_{j \in S} \left| P_{d}^{\sgn}(\lambda_j/C) - \sgn(\lambda_j/C)\right| \cdot |\lambda_j|
    \leq \sum_{j\in S} 3 |\lambda_j| \leq 3DC \kappa = \delta/2~.
\]
In the first inequality, we made use of the fact that $|P_{d}^{\sgn}(\lambda_j/C)| \leq 2$ and $|\sgn(\lambda_j/C)| = 1$, the second inequality follows from $|\lambda_j| \leq C \kappa$ and $|S| \leq D$, and the final equality holds due to our choice of $\kappa$.

For the large eigenvalues, we observe that if $j \in L$, then $|\lambda_j / C| \in [\kappa, 1]$, so by \cref{lem:sign-approx} for sufficiently large $d = O(\log(1/\kappa)/\kappa)$ we have that $\left| P_{d}^{\sgn}(\lambda_j/C) - \sgn(\lambda_j/C)\right| \leq \kappa$.
Hence
\begin{align*}
    \sum_{j \in L} \left| P_{d}^{\sgn}(\lambda_j/C) - \sgn(\lambda_j/C)\right| \cdot |\lambda_j|
    \leq \sum_{j \in L} \kappa \cdot |\lambda_j| \leq \kappa D C \leq \delta /2 \,.
\end{align*}
The second inequality follows because $|L| \leq D$ and $|\lambda_j| \leq C$ since $\norm{A}_\infty \leq C$ by assumption.
The last inequality holds due to our choice of $\kappa$.
Summing both bounds yields the desired result.
\end{proof}

Combining \cref{lem:poly_as_tdoracle} and \cref{lem:be-sign-poly}, we obtain the following result.
\begin{corollary}[Trace distance oracle with block encodings] \label{lem:be-tdoracle}
For any $\alpha \leq 2^{\poly(n)}$ and $\delta \geq 2^{-\poly(n)}$, there exists an $(\alpha, \delta)$-good trace distance oracle $\TraceDistOracle$ that satisfies the following additional property: for any Hermitian $A \in \linear(\C^D)$ with an $(\alpha, \eps, a)$-block encoding $U$ in \pup and any polynomial $q(n)$, there exists an $(\alpha', \alpha'2^{-q(n)} + \alpha'' \sqrt{\eps}, a + \poly(n))$-block encoding of $\TraceDistOracle(A)$ in \pup, where $\alpha' = \alpha'(\alpha, \delta, D) = \poly(n)$ and $\alpha'' = \alpha''(\alpha, \delta, D) \leq 2^{\poly(n)}$ are  \pspace-computable functions derived from \cref{lem:be-sign-poly}.
\end{corollary}

\begin{lemma}[Gibbs oracle] \label{lem:be-gibbs-oracle}
For any $k = \poly(n)$ and $\delta \geq 2^{-\poly(n)}$ there exists a $(k,\delta)$-Gibbs oracle \GibbsOracle that satisfies the following additional property:
for any $\eps \in (0,1)$ and $\alpha \leq 2^{\poly(n)}$, if a matrix $A \in \linear(\C^D)$ has an $(\alpha, \delta^2/(2 c^2 D^2), a)$-block encoding in \pup, there is a $(1,0,a+\poly(n))$-block encoding of $\GibbsOracle(A,\eps)$ in \pup.
Here, $c(\alpha, \eps, k)$ is the \pspace-computable function defined in \cref{lem:be-gibbs-mixed}.
\end{lemma}
\begin{proof}
It suffices to define $\GibbsOracle$ on states that have an $(\alpha, \delta^2/(2 c^2 D^2),a)$-block encoding in \pup because we will only apply $\GibbsOracle$ to such states in our block-encoding based implementation of \cref{algo:mmwu} in \cref{thm:solving}.\footnote{Note that it is a slight abuse of terminology to call what we are constructing here a Gibbs oracle because (i) it is only defined for certain matrices, and (ii) it requires as input not that matrix, but rather a block encoding of that matrix. However, in \cref{thm:solving} whenever we are applying the Gibbs oracle, the input will be a matrix whose block encoding we have access to, so this slightly weaker notion of a Gibbs oracle is sufficient for our purposes.}
For such $A$ with $\norm{A}_\infty \leq k$, by \cref{lem:be-gibbs-mixed} there exists a $(1,\delta/D,b)$-block encoding $V$ of $e^{\beta A}/\tr{e^{\beta A}}$ for $b  = a+\poly(n)$.
We define $\GibbsOracle(A,\beta) = (I \ot \bra{0^b}) V (I \ot \ket{0^b})$, where $\ket{0^b}$ acts on the ancilla register of the block encoding.
By definition, $V$ is a $(1,0, b)$-block encoding of $\GibbsOracle(A,\beta)$.
Furthermore, 
\begin{align*}
\norm{\GibbsOracle(A,\beta) - \frac{e^{\beta A}}{\tr{e^{\beta A}}}}_1 \leq D \norm{\bra{0\cdots0} V \ket{0 \cdots 0} - \frac{e^{\beta A}}{\tr{e^{\beta A}}}}_{\infty} \leq \delta 
\end{align*}
because $V$ is a $(1, \delta/D,b)$-block encoding of $\frac{e^{\beta A}}{\tr{e^{\beta A}}}$.
\end{proof}

\subsubsection{Space complexity of MMWU algorithm}

We now show that the approximate feasibility problem can be solved in $\pureUnitaryPSPACE$.
For this, we are going to instantiate \cref{algo:mmwu} with the trace distance and Gibbs oracles from the previous section.
We will then show that there exists a \pup-computable block encoding of a purification of the output state $\rho$ of \cref{algo:mmwu}.
Note that with this implementation, there are three contributions to the final error (i.e.~the distance between the block-encoded purification and the feasible set of the SDP): the tolerance $\eps$ with which we run \cref{algo:mmwu}, the error $\delta$ that we allow in the trace distance and Gibbs oracles, and a ``block-encoding error'' that arises because we can only prepare a block encoding of the final output state of \cref{algo:mmwu} to within some tolerance. 
\begin{theorem}
\label{thm:solving}
Consider a small-width \pspace-computable SDP described by $(\Phi,B)$ and let $D = \dim(B) = 2^{\poly(n)}$.
Choose any parameters
\begin{itemize}
\item $\eps \geq 1/\poly(n)$ (the tolerance $\eps$ with which we are going to run \cref{algo:mmwu}),
\item $\delta \geq 2^{-\poly(n)}$ (the error of the trace distance and Gibbs oracles with  which we are going to instantiate \cref{algo:mmwu}),
\item $2^{-r}$ for $r = \poly(n)$ (the allowed ``block-encoding error'' with which we are going to approximate a purification of the output state of \cref{algo:mmwu}).
\end{itemize}
Instantiate \cref{algo:mmwu} with error parameter $\eps$ and 
\begin{itemize}
\item the $(\alpha, \delta)$-good trace distance oracle from \cref{lem:be-tdoracle} with $\alpha = D^3 + D$ and $\delta$ as chosen above,
\item and the $(k, \delta)$-Gibbs oracle from \cref{lem:be-gibbs-oracle} with $k = 2 \ln D/\eps^2$ and $\delta$ as chosen above.
\end{itemize}

Then there exists a unitary $V$ in \pup and a purification $\ket{\psi}$ of the output state $\rho$ of this instantiation of \cref{algo:mmwu} such that 
\begin{align*}
\norm{V \proj{0\cdots 0} V^\dagger - \proj{\psi}}_1 \leq 2^{-r(n)} \,.
\end{align*}
\end{theorem}

We prove this theorem in two steps.
First, we show that given good \pup-computable block encodings of the states $\rho_1, \dots, \rho_t$ in \cref{algo:mmwu}, we can construct another good \pup-computable block encoding of $\rho_{t+1}$.
This means that we can implement the main iterative step in \cref{algo:mmwu} using the block-encoding framework.
Second, we show how to extract a purification of the final output state of \cref{algo:mmwu} from the block encodings of the states $\rho_1, \dots, \rho_T$.
We will do this by means of \cref{lem:gibbs-extract}.
We state each of these two steps as separate lemmas.

\begin{lemma} \label{lem:mmwu_iterate}
Consider the same setting as in \cref{thm:solving} and let $\rho_i$ be the intermediate states in \cref{algo:mmwu}.
Then, the following holds for all $t \geq 1$: if every $\rho_1,\dots,\rho_{t}$ has a $(1,0,a)$-block encoding in \pup, then $\rho_{t+1}$ has a $(1, 0, a+\poly(n))$-block encoding.
As before, $\poly(n)$ denotes a polynomial that can depend on our parameter choice in \cref{thm:solving}, but not on $a$.
\end{lemma}
\begin{proof}
First note that in the lemma statement we assume that all of $\rho_1, \dots, \rho_t$ have the same ancilla size $a$; this is without loss of generality since we can always add spurious ancillas, so $a$ can be set as the maximum ancilla size of $\rho_1, \dots, \rho_t$.

We will construct a $(1, 0, a+\poly(n))$-block encoding of $\rho_{t+1}$ step-by-step following \cref{algo:mmwu}.
Let $q(n)$ be a polynomial to be chosen later.
Suppose that $t \geq 1$ and $\rho_1, \dots, \rho_{t}$ have $(1,0,a)$-block encodings in \pup.
We can perform the following steps for any $i = 1, \dots, t$.
\begin{enumerate}
\item \textbf{Block encoding of $\Phi(\rho_i)$}.
Since $\rho_i$ has a $(1,0,a)$-block encoding in \pup and $\Phi$ is \pspace-computable, we can apply \cref{lem:be-superoperator} to get a \pup-computable $(\beta_1, \eta_1, b_1)$-block encoding of $\Phi(\rho_i)$ for 
\begin{align*}
\beta_1 = D^3\,, \quad \eta_1 = D^3 2^{-q(n)}\,, \quad b_1 = a+\poly(n) \,.
\end{align*}
\item \textbf{Block encoding of $B$}.
Since $B$ is \pspace-computable, by \cref{lem:be-entry} we can construct a $(\beta_2, \eta_2, b_2)$-block encoding of $B$ in \pup for
\begin{align*}
\beta_2 = D\,, \quad \eta_2 = 2^{-q(n)}\,, \quad b_2 = \poly(n) \,.
\end{align*}
\item \textbf{Block encoding of $\Phi(\rho_i) - B$}.
Applying \cref{lem:be-lin-comb} and simplifying the error term, we obtain a $(\beta_3, \eta_3, b_3)$-block encoding of $\Phi(\rho_i) - B$ in \pup for
\begin{align*}
\beta_3 = \beta_1 + \beta_2\,, \quad \eta_3 = 4 \beta_1 \eta_1 \,, \quad b_3 = a+\poly(n) \,.
\end{align*}
\item \textbf{Block encoding of $H_i$}.
Applying \cref{lem:be-tdoracle} (with $\alpha = D^3 + D = \beta_1 + \beta_2$ 
 and $\delta$ as chosen in \cref{thm:solving}), we get a $(\beta_4, \eta_4, b_4)$-block encoding in \pup of $H_i = \TraceDistOracle(\Phi(\rho_i) - B)$ for 
 \begin{align*}
\beta_4 = \alpha'\,,\quad \eta_4 = \alpha'2^{-q(n)} + 4 \alpha'' D^3 2^{-q(n)/2}\,,\quad b_4 = a+\poly(n) \,.
\end{align*}
Here, $\alpha' = \alpha'(\beta_3, \delta, D) = \poly(n)$ and $\alpha'' = \alpha''(\beta_3, \delta, D) \leq 2^{\poly(n)}$ are the \pspace-computable functions defined in \cref{lem:be-tdoracle}.
\item \textbf{Block encoding of $H_1 + \cdots + H_t$.}
Applying \cref{lem:be-lin-comb} to the block encodings of $H_i$ from the previous step, we obtain a \pup-computable $(\beta_5, \eta_5, b_5)$-block encoding of $H_1 + \dots + H_t$ for
\begin{align*}
\beta_5 = \beta_4 t\,,\quad
\eta_5 = \beta_4 t (\eta_4 + 2^{-q(n)})\,,\quad b_5 = a+\poly(n) \,.
\end{align*}
\item \textbf{Block encoding of $\Phi^*(H_1 + \cdots + H_t)$.}
It is easy to see that if $\Phi$ is \pspace-computable, so is the adjoint map $\Phi^*$.
Then, we can apply \cref{lem:be-superoperator} (and simplify the error bounds) to obtain a $(\beta_6, \eta_6, b_6)$-block encoding of $\Phi^*(H_1 + \dots + H_t)$ for 
\begin{align*}
\beta_6 = \beta_5 D^3\,,\quad 
\eta_6 = 2 \beta_5 \eta_5 D + 2 D^3 2^{-q(n)}\,,\quad 
b_6 = a+\poly(n) \,.
\end{align*}
\item \textbf{Block encoding of $\rho_{t+1}$.}
$\rho_{t+1}$ is obtained from $\Phi^*(H_1 + \cdots + H_t)$ by applying the $ \left( k, \delta \right)$-good Gibbs oracle provided by \cref{lem:be-gibbs-oracle} with $k = \frac{2 \ln D}{\eps^2}$ and $\eps, \delta$ as chosen in \cref{thm:solving}.

We first observe that by combining all of the errors from the previous steps, we obtain an error $\eta_6 \leq \alpha^* 2^{-q(n)/2}$ for some $\alpha^{*} \leq 2^{\poly(n)}$ that depends on the parameters chosen in \cref{thm:solving}, but is independent of $q(n)$.
This means that by picking a sufficiently large $q(n) = \poly(n)$, we can achieve any desired inverse exponential value for $\eta_6$.
In particular, let $c' = c(\beta_6, \eps, k)$, where $c$ is the function defined in \cref{lem:be-gibbs-oracle}, and note that $c'$ is independent of $q(n)$, too.
Then, we can pick $q(n) = \poly(n)$ large enough\footnote{When running \cref{algo:mmwu}, of course we need to make a choice for $q(n)$ before running the algorithm. However, we can simply pre-compute $\alpha^*$ and $c'$ in \pspace to pick an appropriate $q(n)$ before running the actual steps of the algorithm.} such that $\eta_6 \leq \frac{\delta^2}{2 c'^2 D^2}$.

Next we note that $t \leq T = \ln(D) / \eps^2$, since \cref{algo:mmwu} runs for $T$ iterations.
Furthermore, due to the guarantee of the trace distance oracle, $\norm{H_i}_\infty \leq 2$ for all $i$, so $\norm{H_1 + \cdots H_t}_\infty \leq 2 T$.
Since $\Phi^*$ is contracting (this is the small-width requirement on the SDP $(\Phi,B)$), this means that $\norm{\Phi^*(H_1 + \cdots H_t)}_\infty \leq 2 T \leq k$.

With this choice of $q(n)$, the block encoding of $\Phi^*(H_1 + \cdots + H_t)$ from the previous step therefore satisfies the conditions on the input to the Gibbs oracle spelled out in \cref{lem:be-gibbs-oracle}.
Hence, \cref{lem:be-gibbs-oracle} implies that there is a $(1, 0, a+\poly(n))$-block encoding of $\rho_{t+1} = \GibbsOracle(\Phi^*(H_1 + \dots + H_t))$ in \pup. \qedhere
\end{enumerate}
\end{proof}

\begin{lemma} \label{lem:mmwu_extract}
Consider the same setting as in \cref{thm:solving} and let $\rho_i$ be the intermediate states in \cref{algo:mmwu}.
Then, the following holds for all $t \geq 1$: if every $\rho_1,\dots,\rho_{t}$ has a $(1,0,\poly(n))$-block encoding in \pup, then for any polynomial $r(n)$ there exists a unitary $U_{t+1}$ in \pup such that $U\proj{0\dots 0}U^\dagger$ is $2^{-r(n)}$-close in trace distance to a purification of $\rho_{t+1}$.
\end{lemma}

\begin{proof}
From the proof of \cref{lem:mmwu_iterate} we know that if $\rho_1, \dots, \rho_t$ have $(1,0,a)$-block encodings in \pup, $\Phi^*(H_1 + \dots + H_t)$ has a $(\beta_6, \eta_6, b_6)$-block encoding in \pup with $\beta_6, \eta_6, b_6$ as in the proof of \cref{lem:mmwu_iterate}.
Furthermore, as we showed in \cref{lem:mmwu_iterate} (proof step (vii)), $\norm{\Phi^*(H_1 + \cdots H_t)}_\infty \leq k = 2 \ln D/\eps^2$.  
Then, by \cref{lem:gibbs-extract} there exists a unitary $V$ in \pup such that $\norm{V \proj{0\cdots 0} V^\dagger - \proj{\psi_{t+1}}}_\infty \leq 3 c \sqrt{\eta_6}$ for the \pspace-computable function $c(\beta_6, \eps, k) \leq 2^{\poly(n)}$ defined in \cref{lem:gibbs-extract}.
Since $V \proj{0} V^\dagger - \proj{\psi_{t+1}}$ has rank at most 2, this implies $\norm{V \proj{0} V^\dagger - \proj{\psi_{t+1}}}_1 \leq 6 c \sqrt{\eta_6}$.
Finally, as we observed in the proof of \cref{lem:mmwu_iterate}, we can choose the polynomial $q(n)$ in that proof large enough to make $\eta_6$ an arbitrarily small inverse exponential.
In particular, this means that we can choose a sufficiently large $q(n)$ (as a function of the parameter in \cref{thm:solving}) that ensures that $6 c \sqrt{\eta_6} \leq 2^{-r(n)}$ for the polynomial $r(n)$ chosen in \cref{thm:solving}.\footnote{Note that we made potentially different choices of $q(n)$ here and previously in the proof of \cref{lem:mmwu_iterate}. However, this is easily remedied by picking the maximum of the two choices of $q(n)$, for which the guarantees of \cref{lem:mmwu_iterate} and \cref{lem:mmwu_extract} hold simultaneously.}
\end{proof}

We can now combine \cref{lem:mmwu_iterate} and \cref{lem:mmwu_extract} to show \cref{thm:solving}.
\begin{proof}[Proof of \cref{thm:solving}]
Choose a sufficiently large polynomial $q(n)$ in the proof of \cref{lem:mmwu_iterate} such that \cref{lem:mmwu_iterate} and \cref{lem:mmwu_extract} both hold for the choice of parameters in \cref{thm:solving} and this choice of $q(n)$.

Let $\rho_i$ be the states defined in \cref{algo:mmwu}.
Observe that $\rho_1 = I/D$ (on $\poly(n)$ qubits) is the reduced state of a maximally entangled state of dimension $D$. This purification can be prepared efficiently without any error (using only Hadamard and CNOT gates), so by \cref{lem:be-reduced-density} $\rho_1$ has a $(1,0,\poly(n))$-block encoding in \pup.
Applying \cref{lem:mmwu_iterate} iteratively, this means that $\rho_1, \dots, \rho_T$ all have $(1,0,a(n))$-block encodings in \pup for some universal polynomial $a(n)$ that can depend on the parameters chosen in \cref{thm:solving}. (Concretely, we simply take $a(n)$ to be the largest of the polynomial ancilla sizes we get by applying \cref{lem:mmwu_iterate} $T$ times and can trivially pad all block encodings with smaller ancilla sizes.)

Then it follows from \cref{lem:mmwu_extract} that for every $t = 1, \dots, T$ there exists a unitary $U_t$ in \pup such that $\ket{\psi_t} \deq U_t \ket{0\dots 0}$ is a purification of $\rho_t$ up to an error $2^{-r(n)}$ in trace distance (where $r(n)$ is the polynomial we chose in \cref{thm:solving}).
Since $U_t \in$ \pup (as opposed to non-pure $\unitaryPSPACE$), we can use a control register ranging across $\ket{1}, \dots, \ket{T}$ to combine $U_1, \dots, U_T$ into a unitary $V$ in \pup such that $V\ket{0\dots 0} = \frac{1}{\sqrt{T}} \sum_{t = 1}^T \ket{t} \ket{\psi_t}$.
From the guarantee on the individual states $\ket{\psi_t} = U_t \ket{0\dots 0}$ it follows that this is a purification of the output state $\rho = \frac{1}{T} \sum_{i=1}^T \rho_t$ up to trace distance error $2^{-r(n)}$.
\end{proof}

\section{Application to quantum state complexity}
\label{sec:state_classes}

We now turn our attention to quantum interactive protocols, and in particular interactive protocols for state synthesis.
We begin by formally describing generic quantum interactive protocols (\cref{sec:gen_qip_protocols}) and define the class $\stateQIP$ (\cref{sec:stateqip}).
Following \cite{jain2011qip}, we then show that we can express a $\stateqip$-protocol as an SDP whose feasible points correspond to the state synthesised by the protocol (\cref{sec:sdp-qip}).
Applying our quantum polynomial-space SPD solver from \cref{sec:exp_sdp} to this SDP allows us to show our main result, $\stateQIP = \statePSPACE$ (\cref{sec:solving2}).

\subsection{Quantum interactive protocols} \label{sec:gen_qip_protocols}

Since in quantum computing the standard model of computation is the quantum circuit model (rather than quantum Turing machines), we model the verifier in a quantum interactive protocol as a sequence of \emph{verifier circuits}, one for each input length. A verifier circuit is itself a tuple of quantum circuits that correspond to the operations performed by the verifier in each round of the protocol.

More formally, an \emph{$r$-round quantum verifier circuit} $C = (C_j)_{j \in [r]}$ is a tuple of general quantum circuits that implement quantum channels $\Phi_{C_j}: \linear(\reg{M_{j}} \reg{W_{j-1}}) \to \linear(\reg{M'_{j}} \reg{W_{j}})$ for $j < r$, and $\Phi_{C_r}: \linear(\reg{M_r} \reg{W_{r-1}}) \to \linear(\reg Z \reg S \reg{W_r})$.
For each round $j$, we think of $\reg{M_j}$ as the register containing the incoming message from the prover to the verifier and $\reg{W_{j-1}}$ as the register containing the verifier's private memory before round $j$; $\reg{M'_j}$ and $\reg{W}_j$ contain the outgoing message from the verifier to the prover and the verifier's private memory after the $j$-th round, respectively.
At the end of the last round, the verifier produces an additional register $\reg{Z}$ containing a single qubit indicating whether to accept or reject the prover, and a register $\reg{S}$ containing an output state from the protocol. 
The size of a verifier circuit $C$ is the sum of the circuit sizes of the $C_j$'s.

A \emph{quantum prover} $P$ for an $r$-round quantum verifier circuit $C$ is a tuple of quantum channels $(\Phi_{P_{j}})_{j \in [r]}$, where $\Phi_{P_j}: \linear(\reg{M'_{j-1} Q_{j-1}}) \to \linear(\reg{M_{j} Q_{j}})$. By dilating, we can without loss of generality assume that the prover's channels are all unitary, i.e.~$\Phi_{P_j}(X) = U_j X U_j^\dagger$ for some unitary map $U_j$. 
\cref{fig:prot_generic} shows how the channels $\Phi_{C_i}$ and $\Phi_{P_j}$ are connected to form a quantum interactive protocol.

\begin{figure*}[t]
\centering
\resizebox{0.9\textwidth}{!}{
	\begin{tikzpicture}
			\tikzset{
				block/.style={rectangle,draw,thick,fill=gray!50,inner sep=0pt,minimum width=1.23cm,minimum height=1.25cm},
				point/.style={minimum width=1.25cm},
				every node/.style={scale=1.4,font=\large}
			}
			\coordinate (1) at (-2,3.5){};
			\coordinate (2) at (-2,0){};
			\node[block] (3) at (0,3.5){$\Phi_{P_1}$};
			\node[block] (4) at (2.75,0){$\Phi_{C_1}$};
			\node[block] (5) at (5.5,3.5){$\Phi_{P_2}$};
			\node[block] (6) at (8.25,0){$\Phi_{C_2}$};
			\node[point] (7) at (11,3.5){};
			\node[point] (8) at (11,0){};
			\node[point] (9) at (11.75,3.5){};
			\node[point] (10) at (11.75,0){};
			\node[block] (11) at (17.25,3.5){$\Phi_{P_r}$};
			\node[block] (12) at (14.5,0){$\Phi_{C_{r-1}}$};
			\coordinate (13) at (21.75,3.5){};
			\node[block] (14) at (20,0){$\Phi_{C_r}$};
			\coordinate (15) at (21.75,0){};

			\draw[-latex] (1) -- (3) node[midway,above]{$Q_0$};
			\draw[-latex] (3) -- (5) node[midway,above]{$Q_1$};
			\draw[-latex] (5) -- (7) node[midway,above]{$Q_2$};
			\draw[-latex] (9) -- (11) node[midway,above]{$Q_{r-1}$};
			\draw[-latex] (11) -- (13) node[midway,above]{$Q_r$};			
			\draw[-latex] (2) -- (4) node[midway,below]{$W_0$};
			\draw[-latex] (4) -- (6) node[midway,below]{$W_1$};
			\draw[-latex] (6) -- (8) node[midway,below]{$W_2$};
			\draw[-latex] (10) -- (12) node[xshift=10pt,midway,below left]{$W_{r-2}$};
			\draw[-latex] (12) -- (14) node[midway,below]{$W_{r-1}$};
			\draw[-latex] ([yshift=-15pt]3.east) to[out=0,in=180] node[midway,right]{$M_1$} ([yshift=15pt]4.west);
			\draw[-latex] ([yshift=15pt]4.east) to[out=0,in=180] node[midway,right]{$M'_1$} ([yshift=-15pt]5.west);
			\draw[-latex] ([yshift=-15pt]5.east) to [out=0,in=180] node[midway,right]{$M_2$} ([yshift=15pt]6.west);
			\draw[-latex] ([yshift=15pt]6.east) to [out=0,in=180] node[midway,left]{$M'_2$} ([yshift=-15pt]7.west);
			\fill (10.75,1.5) circle (2pt) (11.25,1.5) circle (2pt) (11.75,1.5) circle (2pt);
			\draw[-latex] ([yshift=-15pt]9.east) to[out=0,in=180] node[midway,right]{$M_{r-1}$} ([yshift=15pt]12.west);
			\draw[-latex] ([yshift=15pt]12.east) to[out=0,in=180] node[midway,right]{$M'_{r-1}$} ([yshift=-15pt]11.west);
			\draw[-latex] ([yshift=-15pt]11.east) to[out=0,in=180] node[midway,right]{$M_r$} ([yshift=15pt]14.west);
			\draw[-latex] (14.east) -- (15) node[right]{$S$};
			\draw[-latex] ([yshift=15pt]14.east) to[out=0,in=180] ([yshift=30pt]15) node[right]{$Z$};
			\draw[-latex] ([yshift=-15pt]14.east) to[out=0,in=180] ([yshift=-30pt]15) node[right]{$W_r$};
	\end{tikzpicture}
	}
\caption{Generic quantum interactive protocol.}
\label{fig:prot_generic}
\end{figure*}
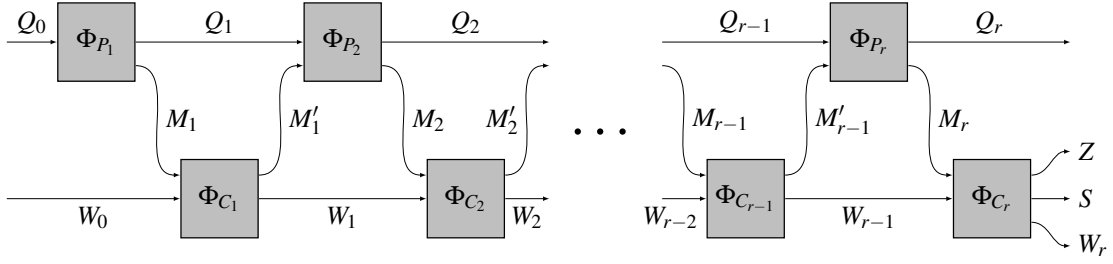

Let $x \in \{0,1\}^*$ denote a string whose length is at most the number of qubits in $W_0$. 
We write $C(x) \interact P$ to denote the interaction between the verifier circuit $C$ and the prover $P$ on input $x$, which means applying the channels $\Phi_{C_j}$ and $\Phi_{P_j}$ as pictured in \cref{fig:prot_generic} to the initial state $\ket{x, 0\dots0}_{W_0} \ket{0\dots0}_{Q_0}$.
We say that $C(x) \interact P$ accepts (resp.~rejects) if measuring the register $Z$ in the standard basis yields the outcome $1$ (resp.~$0$). 

For some fixed input $x$, we denote the intermediate states of the protocol by $\rho_{\reg R}$, where $R$ can be any set of registers that exist in the interactive protocol at the same time.
For example, $\rho_{\reg{M_2 W_1}}$ denotes the reduced state on registers $\reg{M_2}$ and $\reg{W_1}$ after the second message from the prover to the verifier.
Note that it is implicit in the choice of register $\reg R$ which step in the protocol we are referring to.
We say that the \emph{output of $C(x) \interact P$ conditioned on accepting} is the density matrix
\[
\frac{\ptr{\reg Z}{\proj{1}_\reg{Z} \rho_{\reg{ZS}}}}{\bra{1}\rho_\reg{Z} \ket{1}}\,.
\]
In other words, the output state is the reduced density matrix on register $\reg S$ of the protocol's final state conditioned on $C(x) \interact P$ accepting. (If the probability of accepting is $0$, then we leave the output undefined.) 

A \emph{quantum verifier} $V = (V_n)_{n \in \N}$ is a uniform sequence of polynomial-size and polynomial-round quantum verifier circuits.

\subsection{The class \texorpdfstring{$\class{stateQIP}$}{stateQIP}}
\label{sec:stateqip}

Next, we define quantum interactive proofs for state synthesis and their associated complexity class.

\begin{definition}[{$\stateqip$}]
	\label{def:stateQIP}
 Let $c,s,\delta:\N \to [0,1]$ be functions. The class $\stateQIP_{c,s,\delta}$ is the set of state families $R = (\rho_n)_{n}$ (where $\rho_n$ is on $n$ qubits\footnote{The assumption that $\rho_n$ is on $n$ qubits is for convenience; this definition and our results can easily be generalized to the case where each $\rho_n$ is a state on $\poly(n)$ qubits.}) for which there exists an $r(n)$-round quantum verifier $V = (V_n)_{n}$ for some polynomial $r(n)$ such that for all sufficiently large $n$ the following holds:
\begin{itemize}
	\item \emph{Completeness:} There exists a quantum prover $P$ (called an \emph{honest prover}) such that 
	\begin{equation*}
		\pr {V_n \interact P \text{ accepts}} \geq c(n) \,.
	\end{equation*}
	\item \emph{Soundness:} For all quantum provers $P$, it holds that
	\[
	   \text{if } \quad \pr {V_n \interact P \text{ accepts}} \geq s(n) \qquad \text{then} \qquad \td(\sigma, \rho_n) \leq \delta(n)~.
	\]
	where $\sigma$ denotes the output of $V_n \interact P$ conditioned on accepting.
\end{itemize}
Here the probabilities are over the randomness of the interaction. 

Finally, define $\stateQIP_\delta = \stateQIP_{1,\frac{1}{2},\delta}$, and 
\[
\stateQIP = \bigcap_{q(n)} \stateQIP_{1/q(n)}
\]
where the intersection ranges over all polynomials $q(n)$.
\end{definition}

The following lemma shows that the class $\stateQIP_\delta$ is robust under perturbation.

\begin{lemma}
\label{lem:state-qip-robust}
Let $(\rho_n)_{n \in \N} \in \stateQIP_{\delta(n)}$ for some function $\delta(n)$. Suppose $(\tilde{\rho}_n)_{n \in \N}$ is a state sequence satisfying $\td(\rho_n,\tilde{\rho}_n) \leq \eps(n)$ for another function $\eps(n)$. Then $(\tilde{\rho}_n)_{n \in \N} \in \stateQIP_{\delta(n) + \eps(n)}$. 
\end{lemma}
\begin{proof}
Let $V$ denote a $\stateQIP_{\delta(n)}$-verifier for $(\rho_n)_{n \in \N}$. Then by the soundness property, if $\pr {V_n \interact P \text{ accepts}} \geq \frac{1}{2}$, this implies that the output $\sigma$ conditioned on accepting satisfies
\[
    \td(\sigma, \tilde{\rho}_n) \leq \td(\sigma,\rho_n) + \td(\rho_n,\tilde{\rho}_n) \leq \delta(n) + \eps(n)
\]
by the triangle inequality. Thus $V$ is a $\stateQIP_{\delta(n) + \eps(n)}$-verifier for $(\tilde{\rho}_n)_{n \in \N}$.
\end{proof}

The main result of~\cite{rosenthal2022interactive} is the inclusion $(\ket{\psi_n})_{n \in \N} \subseteq \stateQIP$ for all families of \emph{pure} states $(\ket{\psi_n})_n \in \statePSPACE$. 
We will prove in \Cref{thm:purification} that $\statePSPACE$ is closed under purification, i.e.~if a family of mixed states is in $\statePSPACE$, then there exists a purification of those mixed states also in $\statePSPACE$. Putting these two facts together, we get the following:
\begin{theorem} \label{thm:statepspace_in_stateqip}
For all functions $\delta(n)$, for all polynomials $q(n)$, we have 
\[
\statePSPACE_{\delta(n)} \subseteq \stateQIP_{\delta(n) + 1/q(n)}~,
\]
\end{theorem}
We note that, technically speaking,~\cite{rosenthal2022interactive} require $\delta(n) = \exp(-p(n))$ for all polynomials $p(n)$. To see that this implies \cref{thm:statepspace_in_stateqip}, consider a state sequence in $(\rho_n)_{n \in \N} \in \statePSPACE_{\delta(n)}$ for an arbitrary $\delta(n)$ (not necessarily exponentially small). Then by definition there exists a state sequence $(\tilde{\rho}_n)_{n \in \N} \in \statePSPACE_0$ where, for all sufficiently large $n$, we have $\td(\rho_n,\tilde{\rho}_n) \leq \delta(n)$. The main theorem of~\cite{rosenthal2022interactive} combined with \cref{thm:purification} implies that $(\tilde{\rho}_n)_{n \in \N} \in \stateQIP_{1/q(n)}$ for all polynomials $q(n)$. \cref{thm:statepspace_in_stateqip} then follows from \Cref{lem:state-qip-robust}.

\subsection{An SDP for quantum interactive protocols}
\label{sec:sdp-qip}

We will now describe an SDP whose feasible points correspond to quantum states synthesised by a $\stateQIP$-protocol.
For this, consider an $r(n)$-round $\stateQIP_{c,s,\delta}$ protocol with a verifier $V = (V_n)$ for some family of states $(\rho_n)_n$ and some ($\PSPACE$-computable) $c(n),s(n),\delta(n)$ such that $c(n) - s(n) \geq 1/p(n)$ for some polynomial $p(n)$. 
For each $n$, $V_n$ implements a sequence of channels $\Phi_{C_1}, \dots, \Phi_{C_r}$.
We will characterise the set of states that can be produced by such an interactive protocol (for all prover behaviour) as the feasible set of a \pspace-computable SDP. The SDP is the same as the one presented in the survey by Vidick and Watrous~\cite[Section 4.3]{vidick2016quantum}; we in fact show that the same SDP can be solved in (quantum) polynomial space. 

Informally, this SDP has variables corresponding to intermediate states of the verifier during the protocol, and it enforces the following constraints on them:
\begin{enumerate}
\item Each intermediate state is a valid quantum state, i.e.~a positive semidefinite operator with trace 1.
\item The reduced states on the verifier's systems $\reg{W}_i$ and message registers $\reg{M}_i$ must be related to each other by the verifier's channels $\Phi_{C_i}$.
\item The prover's actions have no effect on the verifier's marginal of the state.
\item The verifier's initial state on register $\reg{W_0}$ is the all-0 state.
\item The verifier accepts with probability $c(n)$.\footnote{Note that in \cref{def:stateQIP}, the verifier is required to accept with probability at least $c(n)$. As the prover can always lower the verifier's acceptance probability by deviating from the desired behaviour with some small probability, without loss of generality we can require the prover to achieve an acceptance probability of exactly $c(n)$, which we assume is $\PSPACE$-computable.}
\end{enumerate}
It is not immediately obvious that these constraints are sufficient, but we will show that this is the case in \cref{lem:qip-as-sdp}.

To formalise these constraints, as before we denote the intermediate states of the protocol by $\rho_{\reg R}$ for some collection $\reg R$ of registers that uniquely determines which step in the protocol we are referring to.
Specifically, we will be considering the intermediate states of the protocol on the message register and the verifier's private register, i.e.~the states $\rho_{\reg{M_1 W_0}}, \rho_{\reg{M'_1 W_1}}, \dots, \rho_{\reg{M_r W_{r-1}}}, \rho_{\reg{Z S W_r}}$.\footnote{Note that these states share some registers. For example, both $\rho_{\reg{M'_1 W_1}}$ and $\rho_{\reg{M_2 W_1}}$ are defined in part on register $\reg W_1$. However, importantly these are two (a priori) independent states, i.e.~they are \emph{not} two different reduced states of some larger state $\rho_{\reg{M'_1 M_2 W_1}}$. As a result, the tensor product $\rho_{\reg{M'_1 W_1}} \ot \rho_{\reg{M_2 W_1}}$ is well-defined as a state on $\reg{M'_1 W_1 M_2 W_1}$, i.e.~it includes two copies of $\reg{W_1}$.
Whenever we refer to these intermediate states in this notation, we have to include sufficiently many registers as subscripts to make sure that the states are uniquely identified.
For example, we cannot simply write $\rho_{\reg{W_1}}$ for the reduced state on register $\reg{W_1}$ as there are two independent states, $\rho_{\reg{M'_1 W_1}}$ and $\rho_{\reg{M_2 W_1}}$, of which we could take the reduced state.
We will always have to explicitly write $\ptr{\reg{M_1}}{\rho_{\reg{M'_1 W_1}}}$ or $\ptr{\reg{M_2}}{\rho_{\reg{M_2 W_1}}}$ to avoid this ambiguity.
This notation may seem slightly confusing at first encounter, but turns out to be quite useful in this context.}
On these states, the above constraints can be formalised as follows:
\begin{enumerate}
\item $\rho_{\reg R} \geq 0$ and $\tr{\rho_{\reg R}} = 1$ for all $\reg R \in \{\reg{M_i W_{i-1}}, \reg{M'_i W_i}\}_{i = 1, \dots, r}$, where we set $\reg{M'_{r}} \deq \reg{ZS}$ to ease the notation.
\item $\rho_{\reg{M'_i W_{i}}} - \Phi_{C_i}(\rho_{\reg{M_i W_{i-1}}}) = 0$ for $i = 1, \dots, r$.
\item $\ptr{\reg{M'_{i}}}{\rho_{\reg{M'_{i} W_{i}}}} - \ptr{\reg{M_{i+1}}}{\rho_{\reg{M_{i+1} W_{i}}}} = 0$ for $i = 1, \dots, r-1$.
\item $\rho_{\reg W_0} = \proj{0\dots 0}_{\reg{W_0}}$.
\item $\tr{\proj{1}_\reg{Z}\rho_{\reg Z}} = c(n)$.
\end{enumerate}

We can write these constraints in the required form $\Phi(A) = B$ of an SDP.
For this, we consider $A \in \linear \left( \otimes_{\reg R \in \cR} \reg R \right)$, where the tensor product is over $\cR \deq \{\reg{M_i W_{i-1}}, \reg{M'_i W_i}\}_{i = 1, \dots, r}$.
Then, it is easy to see that the constraints (i)-(v) can be written as linear constraints on $A$.
More formally, there exist linear maps $\Phi^{(j)}$ and matrices $B^{(j)}$ such that for all $A \in \linear(\otimes_{\reg R \in \cR} \reg R)$ that satisfies $A \geq 0$ and $\Phi^{(j)}(A) = B^{(j)}$ for all $j$, the reduced states of $A$ are states $\rho_\reg{R}$ satisfying constraints (i)-(v).
We can combine these constraints by defining $\Phi(A) \deq \frac{1}{N} \oplus_j \Phi^{(j)}(A)$ and $B = \frac{1}{N} \oplus_j B^{j}$ to reach the desired SDP form $\Phi(A) = B$.
Here, $N$ is some normalisation factor to be chosen later.
We provide a formal description of the maps $\Phi^{(j)}$ and the matrices $B^{(j)}$ in \cref{sec:sdp_details}.

Each map $\Phi^{(j)}$ only involves taking a partial trace, applying a verifier channel $\Phi_{C_i}$, and taking the difference of such operations.
It is easy to check that each of these is a linear map $\cM$ whose adjoint is contracting, i.e.~satisfies $\norm{\cM^*(X)}_\infty \leq \norm{X}_\infty$.
Furthermore, taking the adjoint is a linear operation on superoperators.
Therefore, each $(\Phi^{(j)})^*$ is the difference of at most two contracting superoperators, so each $(\Phi^{(j)})^*$ satisfies $\norm{(\Phi^{(j)})^*(X)}_\infty \leq 2 \norm{X}_\infty$.
Since the protocol has $r(n)$ rounds, there are $O(r(n))$ different constraints $\Phi^{(j)}$.
Taking the adjoint channel of the direct sum $\Phi = \frac{1}{N} \oplus_j \Phi^{(j)}$, we see that we can write $\Phi^* = \frac{1}{N} \sum (\Phi^{(j)})^*$, so $\norm{\Phi^*(X)}_\infty \leq \frac{O(r(n))}{N} \norm{X}_\infty$.
By choosing a sufficiently large $N = O(r(n))$, we can therefore make sure that $\Phi^*$ is contracting.
Furthermore, since the verifier circuits have at most polynomial size, each $\Phi_{C_i}$ is \pspace-computable; as a result, $\Phi$ is also \pspace-computable.
Finally, it is easy to see that $\norm{B}_\infty \leq 1$ and $B$ is \pspace-computable.

We can summarise these results into the following lemma.
\begin{lemma} \label{lem:qip-as-sdp}
Fix verifier circuits $C_j$ in a $\stateQIP$ protocol.
Let $\cR \deq \{\reg{M_i W_{i-1}}, \reg{M'_i W_i}\}_{i = 1, \dots, r}$.
Then there exists a \pspace-computable superoperator $\Phi$ and a \pspace-computable matrix $B$ with $\norm{B}_\infty \leq 1$ such that $\Phi^*$ is contracting and the following holds: 
\begin{enumerate}
\item For all prover channels $\Phi_{P_i}$ that lead the prover to be accepted with probability $c(n)$, the intermediate states $(\rho_{R})_{R \in \cR}$ in the protocol satisfy $\Phi(\otimes_{\reg R \in \cR} \rho_\reg{R}) = B$.
\item For all matrices $A$ in the set
\begin{align*}
\Big \{ A \in \linear \big (\bigotimes_{R \in \cR} \reg{R} \big) \; \mid \; A \geq 0 \tand \Phi(A) = B \Big \}
\end{align*}
there exist prover channels $\Phi_{P_j}$ that lead the prover to be accepted with probability $c(n)$ and for which the intermediate states $\rho_{\reg R}$ in the protocol are the partial traces of $A$ on the corresponding registers.
\end{enumerate}
\end{lemma}
\begin{proof}
We use the map $\Phi$ and matrix $B$ described above, for which we have already shown that they are \pspace-computable and satisfy $\norm{\Phi^*(X)}_\infty \leq \norm{X}_\infty$ and $\norm{B}_\infty \leq 1$.
From the construction of $\Phi$ and $B$, condition (i) is immediate.
Condition (ii) holds because we can purify the verifier's channels $\Phi_{C_j}$ and the prover's actions so that all intermediate states in the protocol become pure states; then, the condition follows from the equivalence of purifications. We refer to~\cite[Section 4.3]{vidick2016quantum} for a detailed argument.
\end{proof}

\subsection{Solving the SDP}
\label{sec:solving2}

\begin{theorem}
\label{thm:main}
For all $c(n), s(n) \geq 1/\poly(n)$ such that $c(n) - s(n) \geq 1/\poly(n)$, for all functions $\delta(n)$, and all polynomials $q(n) \geq 0$, we have
 \[
 \stateQIP_{c, s, \delta} \subseteq \statePSPACE_{\delta + 1/q}\,.
 \]
\end{theorem}
\begin{proof}
Consider any sequence of states $(\tau_n)_{n \in \N}$ in $\stateQIP_{c,s,\delta}$.
By the definition of $\stateQIP$ there exists an $r(n)$-round verifier $V = (V_n)$ for some polynomial $r(n)$ such that for all provers $P$, if $P$ is accepted with probability at least $s(n)$, then the output state $\sigma$ conditioned on accepting satisfies $\td(\sigma, \tau_n)\leq \delta(n)$.

We now construct a $\statePSPACE$ procedure for preparing the state $\tau_n$.
For simplicity, we leave the $n$-dependence of this procedure implicit and set $\tau = \tau_n$.
Starting from the above verifier, we can apply \cref{lem:qip-as-sdp} to get a \pspace-computable SDP $(\Phi, B)$ with the following properties (where $\cR$ is as in \cref{lem:qip-as-sdp}):
\begin{enumerate}
\item $\norm{\Phi^*(X)}_\infty \leq \norm{X}_\infty$ and $\norm{B}_\infty \leq 1$.
\item The set
\begin{align}
\cS = \Big \{A \in \linear(\bigotimes_{R \in \cR} \reg{R}) \; : \; A \geq 0 \tand \Phi(A) = B \Big \} \label{eqn:feasible_set}
\end{align}
is not empty, i.e.~the SDP $(\Phi,B)$ is feasible. This follows from the existence of an honest prover for the $\stateqip$-protocol and condition (i) in \cref{lem:qip-as-sdp}.
\end{enumerate}

The SDP $(\Phi, B)$ satisfies the requirements of both \cref{thm:mmwu} and \cref{thm:solving}.
Therefore, we get that for all $\eps \geq 1/\poly(n)$, there exists a unitary $V$ in \pup such that the state $\ket{\psi'} \deq V \ket{0 \dots 0}$ is a purification of a state $\rho'$ such that $\norm{\Phi(\rho') - B}_1 \leq \eps$, i.e.~$\rho'$ is $\eps$-feasible.
This also means that $\rho' \in \statePSPACE_0$.

Intuitively, we would now like to show that our $\eps$-feasible $\rho'$ is in fact close in trace distance to an exactly feasible $\rho \in \cS$, as the definition of $\stateQIP$ only gives guarantees about exactly feasible, not approximately feasible, states.
However, it turns out that we can only show a slightly weaker statement: we can show that $\rho'$ is close in trace distance to a state $\rho$ that satisfies all constraints in $(\Phi, B)$ exactly except the last constraint ($\tr{\proj{1}_\reg{Z}\rho_{\reg Z}} = c(n)$), which will only be satisfied up to inverse polynomial error.
This is still sufficient: as long as $\tr{\proj{1}_\reg{Z}\rho_{\reg Z}} \geq c(n) - 1/\poly(n) \geq s(n)$, the guarantees of $\stateQIP$ apply.

More formally, we define $\tilde \Phi$ to be the same as $\Phi$, except that we do not include the constraint $\tr{\proj{1}_\reg{Z}\rho_{\reg Z}} = c(n)$ (i.e.~we leave out the term in the direct sum in the definition of $\Phi$ corresponding to this constraint).
We likewise define $\tilde B$.
Then, we define
\begin{align}
\tilde \cS = \Big \{A \in \linear(\bigotimes_{R \in \cR} \reg{R}) \; : \; A \geq 0 \tand \tilde \Phi(A) = \tilde B \Big \} \,. \label{eqn:feasible_set_no_c}
\end{align}

Now, we can use \cref{lem:approx_feasible_rounding}, stated and proved below, which shows that there exists a $\rho \in \tilde \cS$ such that 
\begin{align*}
\td(\rho, \rho') \leq  2 r'(n) \eps^{1/4} \qquad \tand \qquad \tr{\proj{1}_\reg{Z}\rho_{\reg Z}} \geq c(n) - O(r(n)) \eps -  \sqrt{2 r'(n)} \eps^{1/8} \,.
\end{align*}
for some polynomial $r'(n)$ that only depends on the number of rounds of the protocol and is independent from $\eps$. 
Since $r'(n)  = \poly(n)$ and $c(n) - s(n) \geq 1/\poly(n)$ are independent of $\eps$, we can now choose $\eps = 1/\poly(n)$ small enough such that $\td(\rho, \rho') \leq 1/q(n)$ and $\tr{\proj{1}_\reg{Z}\rho_{\reg Z}} \geq s(n)$.

Since $\rho \in \tilde S$ and $\tr{\proj{1}_\reg{Z}\rho_{\reg Z}} \geq s(n)$, it follows from condition (ii) of \cref{lem:qip-as-sdp} and the soundness of the $\stateQIP$-protocol for $(\tau_n)_{n \in \N}$ that 
\begin{align}
\td \left( \frac{\ptr{\reg{W_r Z}}{\proj{1}_{\reg Z} \rho_{\reg{ZSW_r}}}}{\tr{\proj{1}_\reg{Z} \rho_{\reg{ZSW_r}}}} , \tau \right) \leq \delta(n) \,. \label{lem:reduced_state_qip_cond}
\end{align}

Since $\rho' \in \statePSPACE_0$ is $1/q(n)$-close in trace distance to $\rho$, we have $\rho' \in \statePSPACE_{1/q(n)}$ by \cref{lem:state-pspace-robust}.
By measuring the $\reg{Z}$-register of $\rho$ and post-selecting on receiving outcome 1 (which we can do in polynomial space simply by repeating the preparation of $\rho'$), this also means that $\frac{\ptr{\reg{W_r Z}}{\proj{1}_{\reg Z} \rho_{\reg{ZSW_r}}}}{\tr{\proj{1}_\reg{Z} \rho_{\reg{ZSW_r}}}} \in \statePSPACE_{1/q(n)}$.
Combining this with \cref{lem:reduced_state_qip_cond} and again using \cref{lem:state-pspace-robust}, we get that $\tau \in \statePSPACE_{\delta(n) + 1/q(n)}$ as desired.
\end{proof}

We now show the remaining step in the proof of \cref{thm:main}.
For this, we will need the follow triangle-like inequality for the fidelity.
\begin{lemma} \label{lem:fidelity_triangle}
Let $\rho, \sigma, \tau$ be quantum states.
Suppose that $F(\rho, \sigma) \geq 1-\delta$ and $F(\sigma, \tau) \geq 1 - \eps$ for $\delta, \eps \in [0,1]$.
Then $F(\rho, \tau) \geq 1 - \delta - \eps - 2 \sqrt{\delta \eps} \geq 1 - \delta - 3 \sqrt \eps$.
\end{lemma}
\begin{proof}
From~\cite[Proposition 10.5]{renes2015quantum}, we have that 
\begin{align*}
F(\rho, \tau) \geq (1 - \delta)(1 - \eps) - \sqrt{(1 - (1-\eps)^2)(1 - (1 - \delta)^2)}\,.
\end{align*}
We can bound $(1 - \delta)(1 - \eps) \geq 1- \delta - \eps$ and $(1 - (1-\eps)^2)(1 - (1 - \delta)^2) \leq 4 \eps \delta$.
\end{proof}

\begin{lemma} \label{lem:approx_feasible_rounding}
Let $(\Phi, B)$ be the SDP associated with a $\stateQIP$-protocol as constructed in \cref{lem:qip-as-sdp}.
Define $\tilde \cS$ as in \cref{eqn:feasible_set_no_c}.
Then, there exists a polynomial $r'(n)$ that only depends on the number of rounds $r(n)$ such that for all $\eps > 0$ and an $\eps$-feasible $\rho'$ there exists a $\rho \in \tilde \cS$ satisfying
\begin{align*}
\td(\rho, \rho') \leq  2 r'(n) \eps^{1/4} \qquad \tand \qquad \tr{\proj{1}_\reg{Z}\rho_{\reg Z}} \geq c(n) - O(r(n)) \eps -  \sqrt{2 r'(n)} \eps^{1/8} \,.
\end{align*}
\end{lemma}
\begin{proof}
First note that we only need to show that a $\rho \in \tilde \cS$ close to $\rho'$ exists, not that it can be constructed efficiently.
Due to the structure of $(\Phi, B)$, if $\rho'$ is $\eps$-feasible, so is the tensor product of reduced states of $\rho'$ on all registers $\reg R \in \cR$.
Therefore, for the rest of the proof we will assume that $\rho'$ is a product state and construct another product state $\rho \in \tilde\cS$, i.e.~we write
\begin{align*}
\rho' = \bigotimes_{\reg R \in \cR} \rho'_{\reg R} \quad \tand \quad \rho = \bigotimes_{\reg R \in \cR} \rho_{\reg R} \,.
\end{align*}
Furthermore, at the cost of incurring at most a factor $O(r(n))$ in the error $\eps$ (which can be absorbed into $r'(n)$), we can ensure that each $\rho'_{\reg R}$ is indeed a quantum state (instead of just being approximately normalised).
By a similar argument, we can also assume that $\rho'_{\reg W_0} = \proj{0\dots 0}_{\reg{W_0}}$ without loss of generality.

Because the trace distance is monotonic under partial trace, we can expand out the definition of $\Phi$ to see that $\rho'$ being $\eps$-feasible implies that for all $i$,
\begin{enumerate}
\item $\td(\rho'_{\reg{M'_i W_{i}}}, \Phi_{C_i}(\rho'_{\reg{M_i W_{i-1}}})) \leq \eps'$,
\item $\td\left(\ptr{\reg{M'_{i}}}{\rho'_{\reg{M'_{i} W_{i}}}}, \ptr{\reg{M_{i+1}}}{\rho'_{\reg{M_{i+1} W_{i}}}}\right) \leq \eps'$
\end{enumerate}
for $\eps' = O(r(n)) \eps$, where the factor $O(r(n))$ is from the renormalisation of the SDP that we performed in \cref{lem:qip-as-sdp} to make $\Phi^*$ contracting.
We can convert these into the following fidelity requirements:
\begin{enumerate}
\item $\fidelity(\rho'_{\reg{M'_i W_{i}}}, \Phi_{C_i}(\rho'_{\reg{M_i W_{i-1}}})) \geq 1 - \eps'$,
\item $\fidelity \left(\ptr{\reg{M'_{i}}}{\rho'_{\reg{M'_{i} W_{i}}}}, \ptr{\reg{M_{i+1}}}{\rho'_{\reg{M_{i+1} W_{i}}}}\right) \geq 1 - \eps'$.
\end{enumerate}
We can now construct a feasible $\rho$ inductively as follows:
set $\rho_{\reg{M_1W_0}} = \rho'_{\reg{M_1W_0}}$.
For $i = 1, \dots, r-1$, suppose we have constructed a $\rho_{\reg{M_i W_{i-1}}}$ that satisfies 
\begin{align*}
\fidelity(\rho_{\reg{M_i W_{i-1}}}, \rho'_{\reg{M_i W_{i-1}}}) \geq 1 - \xi_i \,.
\end{align*}
Then we set $\rho_{\reg{M'_i W_i}} = \Phi_{C_i}(\rho_{\reg{M_iW_{i-1}}})$.
We can bound $\fidelity(\rho_{\reg{M'_i W_i}}, \rho'_{\reg{M'_i W_i}})$ as follows.
By monotonicity of the fidelity under quantum channels, 
\begin{align*}
\fidelity(\rho_{\reg{M'_i W_i}}, \Phi_{C_i}(\rho'_{\reg{M_i W_{i-1}}})) = \fidelity(\Phi_{C_i}(\rho_{\reg{M_i W_{i-1}}}), \Phi_{C_i}(\rho'_{\reg{M_i W_{i-1}}})) \geq 1 - \xi_i \,.
\end{align*}

By fidelity requirement (i) and \cref{lem:fidelity_triangle}, this yields $\fidelity(\rho_{\reg{M'_i W_i}}, \rho'_{\reg{M'_i W_i}}) \geq 1 - \xi_i - 3 \sqrt{\eps'}$.
Using fidelity requirement (ii), monotonicity of the trace distance under partial trace, and \cref{lem:fidelity_triangle}, this means that 
\begin{align*}
\fidelity \left(\ptr{\reg{M'_{i}}}{\rho_{\reg{M'_{i} W_{i}}}}, \ptr{\reg{M_{i+1}}}{\rho'_{\reg{M_{i+1} W_{i}}}}\right) \geq 1 - \xi_i - 6 \sqrt{\eps'} \,.
\end{align*}
By Uhlmann's theorem, there exists a $\rho_{\reg{M_{i+1} W_{i}}}$ such that 
\begin{align*}
\ptr{\reg{M'_{i}}}{\rho_{\reg{M'_{i} W_{i}}}} = \ptr{\reg{M_{i+1}}}{\rho_{\reg{M_{i+1} W_{i}}}} \tand 
\fidelity(\rho_{\reg{M_{i+1} W_{i}}}, \rho'_{\reg{M_{i+1} W_{i}}}) \geq 1 - \xi_i - 6 \sqrt{\eps'} \,.
\end{align*}
We have therefore constructed $\rho_{\reg{M'_i W_i}}$ and $\rho_{\reg{M_{i+1} W_{i}}}$ from $\rho_{\reg{M_{i} W_{i-1}}}$.
Starting this procedure from $\rho_{\reg{M_1 W_0}}$ defined above and applying this procedure iteratively, we can construct $\rho = \otimes_{\reg R \in \cR} \rho_{\reg R}$.

From the construction, it is clear that $\rho \in \tilde \cS$.
Furthermore, using the above error analysis for each step and multiplicativity of the fidelity for tensor products, we get that $\fidelity(\rho, \rho') \geq 1 - r'(n) \sqrt \eps$ for some polynomial $r'(n)$ that only depends on the number of rounds in the protocol; note that we have switched back to $\eps$ from $\eps'$ and absorbed the factor $O(r(n))$ into $r'(n)$. (Here, we have not spelled out the last step of the protocol that produces $\reg{ZSW_r}$, but this only adds an error $\eps'$ by the same argument we used for applying $\Phi_{C_i}$ above.)
By the Fuchs-van de Graaf inequality, 
\begin{align*}
\td(\rho, \rho') \leq \sqrt{1 - (1 - 2 (r'(n))^2 \sqrt \eps)^2} \leq 2 r'(n) \eps^{1/4} \,.
\end{align*}
This shows the first inequality claimed in the lemma.
To show the second, we observe that 
\begin{align*}
\tr{\proj{1}_\reg{Z}\rho_{\reg Z}} 
\geq \tr{\proj{1}_\reg{Z}\rho'_{\reg Z}} - \sqrt{2 r'(n) \eps^{1/4}} 
\geq c(n) - O(r(n)) \eps -  \sqrt{2 r'(n)} \eps^{1/8} \,.
\end{align*}
Here, the first inequality follows from H\"older's inequality and the second inequality holds because $\rho'$ is $\eps$-feasible.
\end{proof}

Recalling the definitions of $\stateQIP$ and $\statePSPACE$, we see that we can combine \cref{thm:statepspace_in_stateqip} and \cref{thm:main} to get equality of these two state complexity classes.
\begin{corollary}
$\stateQIP = \statePSPACE$.
\end{corollary}

\section{$\statePSPACE$ is closed under purification}
\label{sec:purification}

In this section we show that $\statePSPACE$ is closed under purification, answering an open question of~\cite{rosenthal2022interactive}. This is formalized as the following theorem:

\begin{theorem}
\label{thm:purification}
Let $(\rho_n)_n \in \statePSPACE_\delta$ be a family of density matrices for some function $\delta(n)$. Then there exists a pure state family $(\ket{\psi_n})_n \in \statePSPACE_{2 \sqrt \delta}$ such that for all $n$, $\ket{\psi_n}$ is a purification of $\rho_n$. Furthermore, a description of the Turing machine that outputs the circuits synthesizing $(\ket{\psi_n})_n$ can be computed in polynomial-time from the description of the Turing machine that outputs the circuits synthesizing $(\rho_n)_n$. 
\end{theorem}

To prove this theorem, we first need to show that the matrix elements of a density operator in $\statePSPACE$ are $\PSPACE$-computable (in the sense of \cref{def:pspace_comp}).
\begin{lemma}
\label{lem:tomography}
Let $(\rho_n)_n \in \statePSPACE_0$ be a family of density matrices. Then the entries of each density matrix $\rho_n$ are $\PSPACE$-computable.
\end{lemma}
\begin{proof}
We need to show that we can approximate the entries of $\rho_n$ to within error $2^{-p(n)}$ for all polynomials $p(n)$ using polynomial space.
Fix any such $p(n)$.
Let $\rho_n$ be a state on $r_n$ qubits and let $\beta(n) = 4^{-r_n} \cdot 2^{-p(n)}$. 
Consider the following polynomial-space quantum algorithm to estimate entries of $\rho_n$. On input $(1^n,x,y)$, where $x,y \in [2^{r_n}]$ are the indices specifying the matrix element we would like to compute, the algorithm does the following: for each Pauli string $P \in \{I,X,Y,Z\}^{r_n}$, compute an estimate $\alpha_P$ such that $|\alpha_P - \Tr(\rho_n P)| \leq \beta(n)$ with probability at least $1 - \beta(n)$. Output the number
\[
    c_{xy} = \frac{1}{2^{r_n}} \sum_{P \in \{I,X,Y,Z\}^{r_n}} \alpha_P \bra{x}P\ket{y}~.
\]
This algorithm uses polynomial space because it just has to keep track of the Pauli string $P$ being estimated, a running sum for $c_{xy}$, and the space usage of estimating $\alpha_P$ is $O\left(r_n + \log \frac{1}{\beta(n)}\right)$ qubits. 
This is because a single copy of $\rho_n$ can be generated each time and measured with respect to $P$ for a total of $\poly(1/\beta(n))$ times, and a counter is maintained to keep track of the estimate $\alpha_P$. 

Note if $\alpha_P = \Tr(\rho_n P)$ exactly, then $c_{xy} = \bra{x} \rho_n \ket{y}$. Thus with probability at least $1 - 4^{r_n} \beta(n)$, the estimate $c_{xy}$ satisfies
\[
    \Big | c_{xy} - \bra{x} \rho_n \ket{y} \Big| \leq 2^{r_n} \beta(n) \leq 2^{-p(n)} \,.
\]

Now we argue that there is in fact a \emph{classical} polynomial-space algorithm to compute $c_{xy}$. This follows from~\cite{watrous03complexity}, which shows that any quantum polynomial-space-computable language, where the quantum circuits\footnote{In~\cite{watrous03complexity} the model of quantum computation considered is actually quantum Turing machines, but this is equivalent to considering uniform general quantum circuits.} consist of gates with algebraic entries, is also $\PSPACE$-computable with a quadratic blow-up in space\footnote{The task of computing the numbers $c_{xy}$ can be efficiently reduced to decision problems by performing binary search on the bits of $c_{xy}$.}.
Note that the classical polynomial-space algorithm will output a deterministic answer if the quantum algorithm outputs that answer with high probability, so we do not need to take the failure probability of the quantum algorithm into account here.
\end{proof}

With this, we now proceed to prove \cref{thm:purification}.

\begin{proof}[Proof of \Cref{thm:purification}]
Let $\tilde \rho_n \in \statePSPACE_0$ be a family of density matrices such that $\td(\rho_n, \tilde \rho_n) \leq \delta(n)$.
By Uhlmann's theorem, for all purifications $\ket{\tilde \psi_n}$ of $\tilde \rho_n$ there exists a purification $\ket{\psi_n}$ such that 
\begin{align*}
\td(\proj{\psi_n}, \proj{\tilde \psi_n}) \leq \sqrt{1 - \fidelity(\rho_, \tilde \rho_n)^2} \leq \sqrt{2 \delta} \,,
\end{align*}
where we used the Fuchs-van de Graaf inequality.

Therefore, it suffices to find a purification of $\tilde \rho_n$ in $\statePSPACE_{\delta'}$ for $\delta' = (2 - \sqrt 2) \sqrt \delta$.
Let $D_n \leq 2^{\poly(n)}$ denote the dimension of $\tilde \rho_n$ and $p(n)$ a polynomial to be chosen later. By \Cref{lem:tomography} and \Cref{lem:be-entry}, there exists a \pup-computable $(D_n,2^{-p(n)}, \poly(n))$-block encoding $U_n$ of $(\tilde \rho_n)_n$.
By \Cref{lem:sqrt_be_simplified}, there exists a \pup-computable $(\alpha',c2^{-p(n)/2}, \poly(n))$-block encoding $V_n$ of $\sqrt{\tilde \rho_n}$ for \pspace-computable $\alpha', c \leq 2^{\poly(n)}$. 
(Note that here the fact that the ancilla size increases by an \emph{additive} polynomial is not relevant since we are not doing an inductive argument, so we simply observe that the number of ancillas remains polynomial.)
By \Cref{lem:be-purification}, there exists a \pup-computable $(\alpha',c2^{-p(n)/2}, \poly(n))$-block encodings $W_n$ of $(\sqrt{\tilde \rho_n} \otimes I) \ketbra{\Phi_n}{0}$, where $\ket{\Phi_n}$ is the maximally entangled state on $\C^{D_n} \otimes \C^{D_n}$. 

Note that the state
\[
    \ket{\tilde \psi_n} \deq \sqrt{D_n} (\sqrt{\tilde \rho_n} \otimes I) \ket{\Phi_n}
\]
is a purification of $\tilde \rho_n$ and a normalized quantum state. Thus $W_n$ can also be viewed as $(\sqrt{D_n} \alpha',c\sqrt{D_n}2^{-p(n)/2}, \poly(n))$-block encoding of $\ketbra{\tilde \psi_n}{0}$. 

Now consider the following algorithm: on input $1^n$, run $W_n$ on input $\ket{0,0}$ and test whether the ancilla register is in the zero state. If so, output the prepared state. Otherwise, discard all qubits, reinitialize the workspace to all zeroes, and try again for $r(n)$ number of times for some polynomial $r(n)$ to be fixed later. 
If after $r(n)$ repetitions the procedure has not terminated, output the all zeroes state. This algorithm clearly uses a polynomial number of qubits of space.

Consider first the case where the algorithm is successful (i.e.~does not output the all zeros state).
Then the state prepared by the algorithm is 
\begin{align*}
\proj{\phi_n}_S \deq (\sqrt{D_n} \alpha')^2 \bra{0}_A W_n (\proj{0}_S \ot \proj{0}_A) W_n^\dagger \ket{0}_A \,,
\end{align*}
where we have introduced the labels $S$ and $A$ for the ``encoded'' and auxiliary register of the block encoding, respectively, and $\ket{0}$ denotes the all zeroes state on the respective system.
Abbreviating $B_n \deq \sqrt{D_n} \alpha' \bra{0}_A W_n \ket{0}_A$, we can now bound 
\begin{align*}
\norm{\proj{\phi_n} - \proj{\tilde \psi_n}}_1
&= \norm{B_n \proj{0}_S B_n - \proj{\tilde \psi_n}}_1 \\
&\leq \norm{B_n \proj{0}_S B_n - \ket{\tilde \psi_n}\!\bra{0}_S \proj{0}_S B_n}_1
+ \norm{\ket{\tilde \psi_n}\!\bra{0}_S \proj{0}_S B_n - \ket{\tilde \psi_n}\!\bra{0}_S\proj{0}_S \ket{0}\!\bra{\tilde \psi_n}_S}_1 \\
&\leq \norm{B_n - \ket{\tilde \psi_n}\!\bra{0}_S}_\infty \norm{\proj{0}_S}_1 \norm{B_n}_\infty + \norm{\ket{\tilde \psi_n}\!\bra{0}}_\infty \norm{\proj{0}}_1 \norm{B_n - \ket{0}\!\bra{\tilde \psi_n}}_1 \\
&\leq 3 c \sqrt{D_n}2^{-p(n)/2} \,,
\end{align*}
where we repeatedly used H\"older's inequality and the last line follows because $W_n$ is a $(\sqrt{D_n} \alpha',c\sqrt{D_n}2^{-p(n)/2}, \poly(n))$-block encoding of $\ketbra{\tilde \psi_n}{0}$ and $\norm{B_n}_\infty \leq \norm{\ketbra{\tilde \psi_n}{0}}_\infty + \eps \leq 2$.

By definition, the probability of success in each attempt is $\frac{1}{\alpha' D_n}$, so we can choose a sufficiently large polynomial $r(n)$ such that the probability of failure after $r(n)$ repetitions is at most $c \sqrt{D_n}2^{-p(n)/2}$.
Thus the state produced by the algorithm described above is $(4 c \sqrt{D_n}2^{-p(n)/2})$-close to $\proj{\tilde \psi_n}$.
Choosing $p(n)$ large enough, we can ensure $4 c \sqrt{D_n}2^{-p(n)/2} \leq \delta'(n)$. 

The ``furthermore'' part of the theorem follows from the observation that the circuits computing each block encoding in the proof are simple, efficiently-computable functions of circuits of previous block encodings (or of the circuits synthesizing $(\rho_n)_n$). 
\end{proof}

\section{Uhlmann Transformation Problem and the complexity of optimal provers}
\label{sec:strategy}

The main result of this paper implies that the intermediate states of the verifier (i.e.~the reduced state on the verifier's private workspace register and the shared message register) of a quantum interactive protocol, interacting with an optimal prover, can be prepared in quantum polynomial space. However, this does not immediately tell us what the complexity of \emph{implementing} an optimal prover might be. As described in \Cref{sec:gen_qip_protocols}, we can model a prover as applying a unitary operator in each round of the interaction. What is the complexity of performing each unitary? 

Here, it is important that we consider a \emph{uniformly} generated sequence of verifiers $(V_n)_{n \in \N}$: for a \emph{fixed} verifier $V_n$, an optimal prover can always be implemented by a circuit $P_n$ acting on $\poly(n)$ qubits and using at most $2^{\poly(n)}$ gates  -- this is true for every unitary on $\poly(n)$ qubits~\cite[Chapter 4.5]{nielsen2000quantum}. However, this generic fact does not guarantee that all of these circuits can be \emph{uniformly} generated in either a time-efficient or a space-efficient manner. \emph{A priori}, it could be that the sequence of circuits $(P_n)_n$ can only be specified by (say) an \emph{exponential space} Turing machine; this would be a strange asymmetry between the complexity of computing the verifier's final state versus the complexity of the prover itself.

In this section, we show that the complexity of optimal provers in an arbitrary quantum interactive protocol (not necessarily a $\stateQIP$ protocol) with efficient uniform verifiers
is the same as computing the intermediate states of the verifiers (which are computable in quantum polynomial space).

\begin{theorem}[Complexity of optimal provers]
\label{thm:prover-complexity}
Let $(V_n)_{n \in \N}$ denote a family of $r(n)$-round quantum verifiers for some polynomial $r(n)$. Let $\omega^*_n$ denote the optimal acceptance probability over all provers $P_n$ that interact with the verifier $V_n$. Let $q(n)$ be a polynomial and let $(U_{n,j})_{1 \leq j \leq r(n)}$ be the unitary operators of a prover $P$ (who applies $U_{n,j}$ in round $j$) satisfying
\[
    \pr {V_n \interact P \text{ accepts}} \geq \omega_n^* - \frac{1}{q(n)}~.
\]
Then the family of unitaries $(U_n)_{n \in \N}$ is contained in $\unitaryPSPACE_{1/q(n)}$, where
\[
    U_n = \sum_{j=1}^{r(n)} \ketbra{j}{j} \otimes U_{n,j}~.
\]
\end{theorem}

We prove \Cref{thm:prover-complexity} as a special case of a more general computational task that we call the \emph{Uhlmann Transformation Problem}, which we describe next.

\subsection{Uhlmann's theorem and the Uhlmann Transformation Problem}

The well-known Uhlmann's theorem states that the fidelity between two mixed states is equal to the largest overlap between two purifications of those mixed states~\cite{uhlmann1976transition}.
Using that any two purifications are related by a unitary (or, more generally, a partial isometry) acting only on the purifying system, we can state Uhlmann's theorem in the following way.

\begin{theorem}[Uhlmann's theorem] \label{thm:uhlmann_std}
Let $\ket{\psi}_{\reg{AB}}$ and $\ket{\varphi}_{\reg{AB}}$ be pure states on registers $\reg{AB}$ and denote their reduced states on register $\reg{A}$ by $\rho_\reg{A}$ and $\sigma_\reg{A}$, respectively.
Then, there exists a unitary $U_{\reg{B}}$ acting only on register $\reg{B}$ such that 
\begin{align*}
F(\rho_{\reg{A}}, \sigma_{\reg{A}}) = \bra{\varphi}_{\reg{AB}} (\Id_\reg{A} \ot U_\reg{B}) \ket{\psi}_{\reg{AB}} \,.
\end{align*}
\end{theorem}

Uhlmann's theorem motivates the following algorithmic question: how difficult is it to actually implement the \emph{Uhlmann unitary} $U_\reg{B}$?

Naturally, this will depend on the states $\ket{\varphi}$ and $\ket{\psi}$.
We can formalise this into an abstract computational task, which we call the \emph{Uhlmann Transformation Problem} and which is parameterised by two pure state families $(\ket{\psi_n})_{n \in \N}$ and $(\ket{\varphi_n})_{n \in \N}$.

\begin{definition}[Uhlmann Transformation Problem] \label{def:uhlmann_tp}
Let $(\ket{\psi_n})_{n \in \N}$ and $(\ket{\varphi_n})_{n \in \N}$ be families of pure states such that for each $n$, the states $\ket{\psi_n},\ket{\varphi_n}$ have the same number of qubits and the qubits can be divided into two registers $\reg{A}_n, \reg{B}_n$.
The \emph{Uhlmann Transformation Problem for $(\ket{\psi_n})_{n \in \N}$ and $(\ket{\varphi_n})_{n \in \N}$} is the following: given a classical description of these state families (e.g.~in terms of circuit families generating these states), implement the Uhlmann unitaries $(U_n)_{n \in \N}$ such that 
\begin{align*}
\fidelity(\rho_n, \sigma_n) = \bra{\varphi_n} (\Id_{\reg{A}_n} \ot U_n) \ket{\psi_n} \,,
\end{align*}
where $\rho_n$ and $\sigma_n$ are the reduced states on register $\reg{A}_n$ of $\ket{\varphi_n}$ and $\ket{\psi_n}$, respectively.
\end{definition}

In \cref{def:uhlmann_tp}, we have defined the Uhlmann Transformation Problem for any two sequences of pure states $(\ket{\psi_n})_{n \in \N}$ and $(\ket{\varphi_n})_{n \in \N}$.
For the rest of this section, we focus on the case where these two state families are in $\statePSPACE$.
For this case, we will show that we can solve the Uhlmann Transformation Problem (i.e.~implement the unitaries $(U_n)_{n \in \N}$) in $\unitaryPSPACE$.
Specifically, we show the following Algorithmic Uhlmann's Theorem, which solves the Uhlmann Transformation Problem for the case of state sequences in $\statePSPACE$.

\begin{theorem}[Algorithmic Uhlmann's Theorem]
\label{thm:uhlmann}
Let $\delta(n)$ be a function and let $(\ket{\psi_n})_n, (\ket{\varphi_n})_n \in \statePSPACE_\delta$ be pure state families. Suppose that for each $n$ the states $\ket{\psi_n},\ket{\varphi_n}$ have the same number of qubits and the qubits can be divided into two registers $\reg{A}_n, \reg{B}_n$. Then for all polynomials $q(n)$ there exists a sequence of unitaries $( K_n )_n$ in $\unitaryPSPACE$ such that $K_n$ acts on registers $\reg{B}_n$ and an ancilla register $\reg{R}_n$ and satisfies
\[
    \Big \| I_{\reg{A}_n} \otimes   K_n \ket{\varphi_n}_{\reg{A}_n \reg{B}_n} \ket{0}_{\reg{R}_n} - \ket{\psi_n}_{\reg{A}_n \reg{B}_n} \ket{0}_{\reg{R}_n} \Big \|^2 \leq 2(1 - \fidelity(\rho_n,\sigma_n)) + O(\delta(n)) + 2^{-q(n)}\,.
\]
where $\rho_n$ and $\sigma_n$ are the reduced density matrices on register $\reg{A}_n$ of $\ket{\psi_n}$ and $\ket{\varphi_n}$, respectively. 

Furthermore, the description of the Turing machine that outputs the circuits implementing the unitaries $(K_n)_n$ can be computed in polynomial time from the descriptions of the Turing machines that output the circuits for preparing the states $(\ket{\psi_n})_n$ and $(\ket{\varphi_n})_n$.
\end{theorem}

We will prove \cref{thm:uhlmann} in \cref{sec:algo_uhlmann}.
Before doing so, we show how this general result can be used to analyse the complexity of optimal provers in quantum interactive protocols.

\subsection{From the Uhlmann Transformation Problem to optimal quantum provers}

The task of implementing optimal provers in quantum interactive protocols is a special case of the Uhlmmann Transformation Problem.
As a result, we can use \cref{thm:uhlmann} to prove \cref{thm:prover-complexity} as follows.
Let $V_n$ denote an $r(n)$-round quantum verifier in a quantum interactive protocol.
One example is of course a verifier in a $\stateQIP$ protocol, but \cref{thm:prover-complexity} is not restricted to this case.
We will need that the intermediate states on the verifier and message registers in the protocol have purifications in $\statePSPACE$.
This follows straightforwardly by combining the proof of \cref{thm:main} with \cref{lem:be-purification} and we formalise it as \cref{lem:intermediate-states-pspace} below.
In contrast to \cref{sec:sdp-qip}, here we use the following simpler (but less general) notation for the intermediate states of the protocol:
\begin{enumerate}
\item $\rho^{(j)}_{\reg{M^{(n)}} \reg{W^{(n)}}}$ denotes the state of the message register $\reg{M^{(n)}}$ and the private workspace $\reg{W^{(n)}}$ of verifier $V_n$ at the beginning of the verifier's $j$'th turn.
\item $\sigma^{(j)}_{\reg {M^{(n)}W^{(n)}}}$ denotes the state of the message register and the verifier's private workspace at the end of the verifier's $j$'th turn (i.e.~after the verifier has applied its channel $\Phi_{C_j}$). 
\end{enumerate}
Note that here we assume without loss of generality that the workspace and message registers ($\reg{W^{(n)}}$ and $\reg{M^{(n)}}$, respectively) of verifier $V_n$ are identical in all rounds of the protocol; this can always be achieved by padding with ancilla qubits.
\begin{lemma}
\label{lem:intermediate-states-pspace}
Let $V_n$ denote an $r(n)$-round quantum verifier (that receives no input state) with optimal acceptance probability $\omega_n^*$.
For all polynomials $q(n)$ and for all $n \in \N$, there exists a prover $P_n$ that is accepted with probability at least $\omega_n^* - \frac{1}{q(n)}$ for which the following additional property holds: 
there are families of pure states 
\begin{align*}
( \ket{\psi_{n,j}}_{\reg{M^{(n)} W^{(n)} Q^{(n)}}} )_{n,j}, \,\, (\ket{\varphi_{n,j}}_{\reg{M^{(n)} W^{(n)} Q^{(n)}}})_{n,j} \in \statePSPACE_{1/q(n)}
\end{align*}
for some purifying registers $\reg{Q^{(n)}}$ that are purifications of the intermediate states $\rho^{(j)}_{\reg{M^{(n)}} \reg{W^{(n)}}}$ and $\sigma^{(j)}_{\reg{M^{(n)}} \reg{W^{(n)}}}$ of the verifier $V_n$ interacting with the prover $P_n$.
\end{lemma}
\begin{proof}
Note that $V_n$ need not be a verifier for a $\stateQIP$-protocol, but can be a verifier for any interactive quantum protocol (without an input state, i.e.~the verifier's starting state can be fixed to all-0 without loss of generality).
In this more general case, the corresponding SDP we constructed in \Cref{sec:sdp-qip} is still well-defined as it did not use any particular properties of state synthesis protocols.\footnote{Technically speaking the SDP in \Cref{sec:sdp-qip} is defined to search for provers that are accepted with probability $c(n)$, the completeness parameter, which is not necessarily the acceptance probability of an optimal prover.
However, since the maximum acceptance probability $\omega_n^*$ is a $\PSPACE$-computable quantity~\cite{jain2011qip}, we can set $c(n) = \omega^*_n$ and still have a \pspace-computable SDP.}
Following the same steps as in the proof of \cref{thm:main}, this means that there exists a prover who succeeds with probability at least $\omega^* - 1/q(n)$ and for which the intermediate states $\rho^{(j)}_{\reg{M^{(n)}} \reg{W^{(n)}}}$ and $\sigma^{(j)}_{\reg{M^{(n)}} \reg{W^{(n)}}}$ are in $\statePSPACE_{1/(4 q^2(n))}$.
The lemma then follows from \cref{thm:purification}.
\end{proof}

Having shown that the intermediate states of the protocol have purifications in $\statePSPACE$, we can now use our Algorithmic Uhlmann's Theorem (\cref{thm:uhlmann}) to implement a $\unitaryPSPACE$-prover that performs the Uhlmann unitaries connecting these purifications.
Since this prover produces exactly the same states from the verifier's point of view as the prover $P_n$ from \cref{lem:intermediate-states-pspace} and that prover was accepted by the verifier with probability $\omega^* - 1/q(n)$, it follows that the $\unitaryPSPACE$-prover that implements the Uhlmann transformations is accepted with the same probability (up to an additional factor of the number of rounds $r(n)$ due to the error accumulating in each round).
More formally, we show the following.
\begin{proof}[Proof of \cref{thm:prover-complexity}]
For a given polynomial $q(n)$, define $q'(n) = q(n) r(n)$ and let 
\begin{align*}
( \ket{\psi_{n,j}}_{\reg{M^{(n)} W^{(n)} Q^{(n)}}} )_{n,j}, \,\, (\ket{\varphi_{n,j}}_{\reg{M^{(n)} W^{(n)} Q^{(n)}}})_{n,j} \in \statePSPACE_{1/q'(n)}
\end{align*}
be the purifications constructed in \cref{lem:intermediate-states-pspace} for a prover $P_n$ that succeeds with probability $\omega^* - 1/q'(n)$.
We can also view these as purifications of the reduced states $\rho^{(j)}_{\reg{W^{(n)}}}$ and $\sigma^{(j)}_{\reg{W^{(n)}}}$ on the verifier registers only.
Recall that by construction $\rho^{(j)}_{\reg{W^{(n)}}}$ and $\sigma^{(j)}_{\reg{W^{(n)}}}$ are the \emph{exact} reduced states of the verifier $V_n$ interacting with the prover $P_n$.
Therefore, $F(\sigma^{(j)}_{\reg{W^{(n)}}}, \rho^{(j+1)}_{\reg{W^{(n)}}}) = 1$,
since the prover does not act on register $\reg{W^{(n)}}$.
Then, \cref{thm:uhlmann} implies that  there exist unitaries $\{K_{n,j}\}_{n,j} \in \unitaryPSPACE$ acting on $\reg{M^{(n)}W^{(n)}}$ such that 
\begin{align*}
\norm{(\id_{\reg{W^{(n)}}} \ot K_{n,j}) \ket{\varphi_{n,j}}_{\reg{W^{(n)} M^{(n)} Q^{(n)}}} - \ket{\psi_{n,j+1}}_{\reg{W^{(n)} M^{(n)} Q^{(n)}}}}_2 \leq O(1/q'(n)) \,.
\end{align*}
Therefore, a prover $P^*_n$ that, in each round $j \in \{1, \dots, r(n)\}$ of the protocol, applies the unitary $K_{n,j}$, will be accepted with probability $\omega^* - O(r(n)/q'(n)) = \omega^* - O(1/q(n))$ by the triangle inequality.
\end{proof}

\subsection{Proof of the Algorithmic Uhlmann's Theorem} \label{sec:algo_uhlmann}

In this section, we prove our Algorithmic Uhlmann's Theorem (\cref{thm:uhlmann}).
As a first step, we show \cref{lem:uhlmann}, which gives an explicit form for the Uhlmann unitary in terms of the involved states $\ket{\varphi}$ and $\ket{\psi}$. 
For this, we need the following piece of notation: for a matrix $A$ with a singular value decomposition $A = U\Sigma V^\dagger$ and a function $f: \R \to \R$ such that $f(0) = 0$, we write
\begin{equation}
    \label{eq:singular-val-func}
    f(A) = U f(\Sigma) V^\dagger~.
\end{equation}
This is well-defined because the span of the left singular vectors (resp. right singular vectors) corresponding to a nonzero singular value of a matrix is unique. 

\begin{lemma}[Explicit Uhlmann unitary]
\label{lem:uhlmann}
Let $\ket{\psi}_{\reg{A} \reg{B}}$ and $\ket{\varphi}_{\reg{A} \reg{B}}$ denote pure states on registers $\reg{A}, \reg{B}$. Let $\rho$ and $\sigma$ denote the reduced density matrices on register $\reg{A}$ of $\ket{\psi}$ and $\ket{\varphi}$, respectively. Define the operator
\[
    W = \sgn(\Tr_{\reg{A}} (\ketbra{\psi}{\varphi}))
\]
acting on register $\reg{B}$. 
Then $W$ is a partial isometry satisfying
\[
    \fidelity(\rho,\sigma) =  \bra{\psi} (I_{\reg{A}} \otimes W) \ket{\varphi}~.
\]
\end{lemma}
\begin{proof}
Assume without loss of generality that the dimensions of registers $\reg{A}$ and $\reg{B}$ are equal. Let $\ket{\Phi}_{\reg{A} \reg{B}}$ denote the unnormalized maximally entangled state between registers $\reg{A}$ and $\reg{B}$. 
By the equivalence of purifications, there exist unitary operators $X,Y$ acting on register $\reg{B}$ such that
\[
    \ket{\psi} = \sqrt{\rho} \otimes X \ket{\Phi} \qquad \text{and} \qquad \ket{\varphi} = \sqrt{\sigma} \otimes Y \ket{\Phi}~.
\]
Thus 
\begin{align*}
    \Tr_{\reg{A}} (\ketbra{\psi}{\varphi})
    &= \Tr_{\reg{A}} (\sqrt{\rho} \otimes X \ketbra{\Phi}{\Phi} \sqrt{\sigma} \otimes Y^\dagger )  \\
    &= (I_{\reg{A}} \otimes X) \Tr_{\reg{A}} \Big( (\sqrt{\rho} \otimes I_{\reg{B}}) \ketbra{\Phi}{\Phi} (\sqrt{\sigma} \otimes I_{\reg{B}}) \Big) (I_{\reg{A}} \otimes Y^\dagger) \\
    &= X (\sqrt{\sigma} \sqrt{\rho})^\top Y^\dagger
\end{align*}
where the transpose is with respect to the basis $\{\ket{i}\}$ in which $\ket{\Phi}$ is equal to $\sum_i \ket{i} \ket{i}$. Then we have
\begin{align*}
    W &= \sgn( X (\sqrt{\sigma} \sqrt{\rho})^\top Y^\dagger ) \\
    &= X \sgn( (\sqrt{\sigma} \sqrt{\rho})^\top ) Y^\dagger \\ 
    &= X \sgn(\sqrt{\sigma} \sqrt{\rho} )^\top Y^\dagger~.
\end{align*}
The operator $W$ is a partial isometry, meaning that $W^\dagger \, W$ is a projection onto its support:
\begin{align*}
    W^\dagger \, W &= Y \overline{\sgn(\sqrt{\sigma} \sqrt{\rho} )} \sgn(\sqrt{\sigma} \sqrt{\rho} )^\top Y^\dagger \\
    &= Y \overline{\Big(\sgn(\sqrt{\sigma} \sqrt{\rho} ) \sgn(\sqrt{\sigma} \sqrt{\rho} )^\dagger \Big) } Y^\dagger \\
    &= Y \overline{U} \sgn(\Sigma)^2 \overline{U}^\dagger Y^\dagger
\end{align*}
where $\overline{A}$ denotes taking the complex conjugate of the entries of $A$ in the standard basis, and the singular value decomposition of $\sqrt{\sigma} \sqrt{\rho}$ is $U \Sigma V^\dagger$. Note that since $\Sigma \geq 0$, we have $\sgn(\Sigma)^2 = \sgn(\Sigma)$, which is an orthogonal projection and since the singular value decomposition of $W$ is $X \overline{V} \sgn(\Sigma) \overline{U}^\dagger Y^\dagger$ it follows that $W^\dagger W$ is the projection onto the support of $W$, as desired.

We now calculate
\begin{align*}
    \bra{\psi} (I_{\reg{A}} \otimes W) \ket{\varphi} &= \bra{\Phi} (\sqrt{\rho} \otimes X^\dagger) ( I_{\reg{A}} \otimes X \sgn(\sqrt{\sigma} \sqrt{\rho})^\top \,  Y^\dagger) (\sqrt{\sigma} \otimes Y) \ket{\Phi} \\
    &= \bra{\Phi} \sqrt{\rho} \sqrt{\sigma} \otimes \sgn(\sqrt{\sigma} \sqrt{\rho})^\top \ket{\Phi} \\
    &= \bra{\Phi} \sqrt{\rho} \sqrt{\sigma} \, \sgn(\sqrt{\sigma} \sqrt{\rho}) \otimes I_{\reg{B}} \ket{\Phi} \\
    &= \Tr \Big( \sqrt{\rho} \sqrt{\sigma} \, \sgn(\sqrt{\sigma} \sqrt{\rho}) \Big) \\
    &= \Tr( |\sqrt{\rho} \sqrt{\sigma}|) \\
    &= \fidelity(\rho,\sigma)
\end{align*}
where the second-to-last line follows from the fact that $\Tr(K \sgn(K^\dagger) ) = \Tr(|K|)$ for all matrices $K$, and the last line follows by definition of the fidelity function.
\end{proof}

We now state a ``robust'' version of \Cref{lem:uhlmann}.
\begin{lemma}[Robust version of \Cref{lem:uhlmann}]
\label{lem:uhlmann-robust}
Let $\kappa \in (0,1)$, $\alpha\geq 1$, and let $P_d^{\sgn}$ denote the odd degree-$d = O \Big( \frac{\log 1/\kappa}{\kappa} \Big)$ polynomial approximation to the sign function from \Cref{lem:sign-approx}. 
Let $\ket{\psi}_{\reg{A} \reg{B}}$, $\ket{\varphi}_{\reg{A} \reg{B}}$ and $\rho, \sigma$ be as in \Cref{lem:uhlmann}. 
Let
\[
    \widetilde{W} = P_d^{\sgn} (\Tr_{\reg{A}}(\ketbra{\psi}{\varphi})/\alpha)
\]
where $P_d^{\sgn}$ is applied in the sense of~\Cref{eq:singular-val-func}. 
Then 
\[
    \Big | \bra{\psi} I_{\reg{A}} \otimes \widetilde{W} \ket{\varphi} - \fidelity(\rho,\sigma) \Big| \leq O(\dim \reg{A} \cdot \alpha \cdot \kappa)~.
\]
\end{lemma}
\begin{proof}
First note that since $P_d^{\sgn}$ is odd, $P_d^{\sgn}(0) = 0$, so the application of $P_d^{\sgn}$ in the sense of \cref{eq:singular-val-func} is well-defined.
Then, the proof is nearly identical to that of \Cref{lem:uhlmann}, but the $\sgn$ function is replaced with the polynomial $P_d^{\sgn}$. 
The only difference is that we have to bound the difference
$$
\Big | \Tr \Big( \sqrt{\rho} \sqrt{\sigma}  \, \sgn(\sqrt{\sigma} \sqrt{\rho}) \Big) - \Tr \Big( \sqrt{\rho} \sqrt{\sigma} \, P_d^{\sgn} (\sqrt{\sigma} \sqrt{\rho} /\alpha) \Big) \Big |
$$
Let $K = \sqrt{\rho} \sqrt{\sigma}$ and suppose it has singular value decomposition $U \Sigma V^\dagger$. Let $( \lambda_1,\ldots,\lambda_t)$ denote the nonzero diagonal entries of $\Sigma$. Let $S = \{ i : \lambda_i/\alpha \leq \kappa \}$ denote the ``small'' singular values and let $L$ denote the complement (the ``large'' singular values). Then this difference is equal to
\begin{align*}
    \Big | \Tr \Big( \Sigma ( \sgn(\Sigma) - P_d^{\sgn} (\Sigma/\alpha) ) \Big) \Big| &\leq \sum_i \lambda_i \, | 1 - P_d^{\sgn}(\lambda_i/\alpha) | \\
    &= \sum_{i \in S} \lambda_i \, | 1 - P_d^{\sgn}(\lambda_i/\alpha) | + \sum_{i \in L} \lambda_i \, | 1 - P_d^{\sgn}(\lambda_i/\alpha) | \\
    &\leq O(t \alpha \kappa) + O(t \kappa) = O(t \alpha \kappa)
\end{align*}
where to bound the sum over $S$ we used that $|1 - P_d^{\sgn}(x)| \leq O(1)$ for all $x \in [-1,1]$ and to bound the sum over $L$ we used that $|\lambda_i| \leq 1$ for all $i$, and the fact that for $\kappa \leq x \leq 1$ we have $|1 - P_d^{\sgn}(x)| \leq \kappa$ by definition. Since $t$ (the number of singular values of $K$) is at most $\dim \reg{A}$, the lemma follows.
\end{proof}

In the proof of \cref{thm:uhlmann}, we will need the following utility lemma that allows us to embed an arbitrary square matrix $R$ into a Hermitian matrix $Q$ in a convenient way.
\begin{lemma}
\label{lem:singular-value-func}
Let $R$ denote a square matrix with singular value decomposition $U \Sigma V^\dagger$. Let $Q$ denote the Hermitian matrix
\[
Q = \ketbra{0}{1} \otimes R^\dagger + \ketbra{1}{0} \otimes R~.
\]
Let $f:\R \to \R$ be an odd function. Then 
\[
    f(Q) = \ketbra{0}{1} \otimes f(R)^\dagger + \ketbra{1}{0} \otimes f(R)
\]
where $f(R) = U f(\Sigma) V^\dagger$.
\end{lemma}
\begin{proof}
    We can rewrite $Q$ as 
    \[
        Q = A (X \otimes \Sigma) A^\dagger
    \]
    where $A = \ketbra{0}{0} \otimes V + \ketbra{1}{1} \otimes U$ (which is unitary) and $X = \begin{pmatrix} 0 & 1 \\ 1 & 0 \end{pmatrix}$ is the single-qubit bitflip operator. Then $f(Q) = A f(X \otimes \Sigma) A^\dagger$. Since $f$ is an odd function and the eigenvalues of $X \otimes \Sigma$ are simply the diagonal entries of $\Sigma$ and their negations, we have $(X \otimes \Sigma) = X \otimes f(\Sigma)$. This implies that 
    \[
        f(Q) = A (X \otimes f(\Sigma))A^\dagger = \ketbra{0}{1} \otimes f(R)^\dagger + \ketbra{1}{0} \otimes f(R)
    \]
    as desired.
\end{proof}

Another ingredient that we will require is a certain ``robust'' version of oblivious amplitude amplification~\cite{berry2014exponential}, which is based on~\cite[Theorem 3.2.5]{gilyen2019quantum}.
To state this, we first make the following definition.

\begin{definition}[Approximate isometry on a subspace]
\label{def:approx-isometry}
Let $F$ be an operator with singular value decomposition $U \Sigma V^\dagger$, and let $\Delta_{1 \pm \kappa}$ denote the span of columns of $V$ (i.e.~the right singular vectors) such that the associated singular values are in the interval $[1 - \kappa,1 + \kappa]$. Then we say that $F$ is a \emph{$\kappa$-approximate isometry on a subspace $\Gamma$} if $\Gamma$ is a subspace of $\Delta_{1 \pm \kappa}$. If $\kappa = 0$, then we say that $F$ is an \emph{exact isometry} on $\Gamma$. 
\end{definition}
In other words, $F$ is a $\kappa$-approximate isometry on a subspace $\Gamma$ if for every normalised $\ket{\psi} \in \Gamma$, $\bra{\psi}F^\dagger F \ket{\psi} \in [1-\kappa, 1+\kappa]$.

The final ingredient we will need is a way to ``reset'' the post-selection factor $\alpha$ of a block encoding of an approximate isometry $F$ on a subspace $\Gamma$ to $\alpha \approx 1$.
This can be accomplished using oblivious amplitude amplification~\cite{berry2014exponential}.
However, since $F$ is an approximate, not an exact, isometry, and $F$ only behaves this way on a subspace, we cannot use the result of \cite{berry2014exponential} directly.
However, we can show that the following robust version of oblivious amplitude amplification holds for block encodings of approximate isometries.
We defer the proof to \cref{app:oaa}.

\begin{lemma}[Robust oblivious amplitude amplification] \label{lem:oaa_for_approx_iso}
Let $J$ be an $(\alpha,\eps,a)$-block encoding of an operator $\widetilde{W}$ that is a $\kappa$-approximate isometry on a subspace $\Gamma$.
Let $\Pi = \Id \otimes \ketbra{0^a}{0^a}$ denote the projector onto the ancilla qubits of the block encoding being zero, let $L = 2\Pi - \Id$ and let $S = -J L J^\dagger L$. Then for all $\ell \in \N$ and all $\ket{\phi} \in \Gamma$, 
\[
    \Big \| S^\ell J \ket{0^a}\ket{\phi} - \sin((2\ell + 1)\theta) \, \ket{0^a}  \widetilde{W} \ket{\phi} \Big \| \leq \cos((2\ell + 1)\theta) + O\left(\ell \sqrt{\kappa + \eps}\right)
\]
where $\theta = \arcsin(\alpha^{-1})$.
\end{lemma}

With these ingredients at hand, we can now prove our Algorithmic Uhlmann's Theorem.
\begin{proof}[Proof of \cref{thm:uhlmann}]
Assume for now that the state families $(\ket{\psi_n})_n,(\ket{\varphi_n})_n$ are in $\statePSPACE_0$, i.e.~they can be synthesized exactly in quantum polynomial space; we will handle the approximate case at the end. Let $D_n = 2^{\poly(n)}$  denote the dimension of $\ket{\psi_n},\ket{\varphi_n}$, and let $D_n^{(1)}$ and $D_n^{(2)}$ denote the dimensions of the first and second registers, respectively, so that $D_n = D_n^{(1)} D_n^{(2)}$. 

By \Cref{lem:tomography}, the amplitudes of the states $\ket{\psi_n},\ket{\varphi_n}$ are $\PSPACE$-computable, and as a result the entries of the matrices $\ketbra{\psi_n}{0}, \ketbra{0}{\varphi_n}$ are \pspace-computable, too. Therefore, by \cref{lem:be-entry} for all polynomials $p(n)$ there are $(\alpha_1,\eps_1, \poly(n))$-block encodings $A_n$ of $\ketbra{\psi_n}{0}$ and $B_n$  of $\ketbra{0}{\varphi_n}$ that are computable in $\pureUnitaryPSPACE$, where $\alpha_1 = D_n$ and $\eps_1 = 2^{-p(n)}$. \Cref{lem:be-product} implies a $(\alpha_2,\eps_2, \poly(n))$-block encoding $E_n$ of $\ketbra{\psi_n}{\varphi_n}$ that is computable in $\pureUnitaryPSPACE$, where $\alpha_2 = \alpha_1^2$ and $\eps_2 = 2\alpha_1 \eps_1$. By \Cref{lem:be-partial-trace}, there exists an $(\alpha_3,\eps_3, \poly(n))$-block encoding $F_n$ of 
\begin{align*}
R_n \deq \Tr_{\reg{A}}(\ketbra{\psi_n}{\varphi_n})
\end{align*}
computable in $\pureUnitaryPSPACE$ where $\alpha_3 = D_n^{(1)} \alpha_2$ and $\eps_3 = 2 D_n^{(1)} \eps_2$. 

At this point, we would like to apply the $\sgn$ function (or rather its polynomial approximation) to the block encoding of $R_n$.
Unfortunately, our tools for applying polynomials (in particular, \cref{lem:be-sign-poly}) require a block encoding of a Hermitian matrix, but $R_n$ is not Hermitian.
We therefore employ the construction from \cref{lem:singular-value-func} and define the unitary
\[
    G_n = \ketbra{0}{1} \otimes F_n^\dagger + \ketbra{1}{0} \otimes F_n~.
\]
It is straightforward to see that $G_n$ is a $(1,\eps_4, \poly(n))$-block encoding of the matrix
\begin{equation}
    \label{eq:uhlmann-qn}
    Q_n = \frac{1}{\alpha_3} \Big( \ketbra{0}{1} \otimes  R_n^\dagger + \ketbra{1}{0} \otimes R_n \Big).
\end{equation}
for $\eps_4 = 2\eps_3/\alpha_3$. Let $\kappa \geq 2^{-\poly(n)}$ to be chosen later and let $P_d^{\sgn}$ denote the degree-$d = O \Big ( \frac{\log 1/\kappa}{\kappa} \Big)$ polynomial approximation to the sign function from \Cref{lem:sign-approx}. 
Using \Cref{lem:be-sign-poly}, we obtain a $\pureUnitaryPSPACE$-computable $(\alpha_5,\eps_5, \poly(n))$-block encoding $H_n$ of $P_d^{\sgn}(Q_n)$ where $\alpha_5 = O(\log d)$ and $\eps_5 = O(\alpha_5 d \sqrt{\eps_4})$. 

We can connect a block-encoding of $P^{\sgn}_d(Q_n)$ to a block-encoding of $P^{\sgn}_d(R_n)$ (where $P^{\sgn}_d$ is again applied to the singular values of $R_n$, see \Cref{eq:singular-val-func}) by means of \Cref{lem:singular-value-func}, which implies that $H_n$ is a $(\alpha_5,\eps_5, \poly(n))$-block encoding of
\[
\ketbra{0}{1} \otimes \widetilde{W}_n^\dagger  +  \ketbra{1}{0} \otimes  \widetilde{W}_n \,,
\]
where $\widetilde{W}_n = P_d^{\sgn}(R_n/\alpha_3)$. Now let $J_n = (X \otimes I) H_n$ where $X$ is the single-qubit bit-flip operator acting on the first qubit; thus $J_n$ is a $(\alpha_5,\eps_5, \poly(n))$-block encoding of
\[
    (X \otimes I) P_d^{\sgn}(Q_n) =  \ketbra{0}{0} \otimes \widetilde{W}_n  +  \ketbra{1}{1} \otimes  \widetilde{W}_n^\dagger
\]
where $\widetilde{W}_n = P_d^{\sgn}(R_n/\alpha_3)$. Viewing the first qubits as part of the ancilla register of the block encoding, we get that $J_n$ is a $(\alpha_5,\eps_5,\poly(n))$-block encoding of $\widetilde{W}_n$, and furthermore $J_n$ is $\pureUnitaryPSPACE$-computable.

We now want to apply \cref{lem:oaa_for_approx_iso} to $J_n$.
For this, we define $\Gamma_n$ to be the span of right singular vectors of $R_n/\alpha_3$ whose corresponding singular value is at least $\kappa$.
We then claim that $\widetilde{W}_n$ is a $\kappa$-approximate isometry on the subspace $\Gamma_n$. 
To see that this is the case, let $U \Sigma V^\dagger$ denote the singular value decomposition of $R_n/\alpha_3$. Since $\widetilde{W}_n = U P^{\sgn}_d(\Sigma) V^\dagger$, the right singular vectors of $R_n/\alpha_3$ whose singular values are at least $\kappa$ will now have associated singular values in the interval $[1 - \kappa,1+\kappa]$ due to the guarantees of the polynomial $P_d^{\sgn}$ given by \Cref{lem:sign-approx}.

This means that $J_n$ is a \pup-computable $(\alpha_5,\eps_5,\poly(n))$-block encoding of the $\kappa$-approximate isometry $\tilde W_n$ on the subspace $\Gamma_n$.
Observe that we can turn this into a \pup-computable $(\alpha_5',\eps_5,\poly(n))$-block encoding of $\widetilde W$ with $\alpha_5' = 1/\sin \left( \frac{\pi}{2(2\ell+1)} \right)$ for some integer $\ell \in \N$.
To see this, note that by tracking $\alpha_5$ through the proof, if follows that $\alpha_5$ is \pspace-computable;
therefore we can classically compute the smallest $\alpha_5' \geq \alpha$ such that $\alpha_5' = 1/\sin \left( \frac{\pi}{2(2\ell+1)} \right)$ for some $\ell \in \N$.
Then, we can modify our block encoding $J_n$ by appending an additional ancilla qubit on which $J_n$ acts as a (very small) single-qubit rotation $R$ such that $\bra{0}R\ket{0} = \alpha_5/\alpha_5' \leq 1$.
To ease the notation, for the rest of the proof we will simply assume without loss of generality that $J_n$ is a \pup-computable $(\alpha_5,\eps_5,\poly(n))$-block encoding of $\widetilde W$ with $\alpha_5 = 1/\sin \left( \frac{\pi}{2(2\ell+1)} \right)$.

For these values of $\alpha_5$, $\ell$, and $\theta = \arcsin(\alpha^{-1})$ as in \cref{lem:oaa_for_approx_iso}, $\sin((2\ell + 1)\theta) = 1$ and $\cos((2 \ell + 1)\theta) = 0$, so it follows from \cref{lem:oaa_for_approx_iso} that for any $\ket{\psi} \in \Gamma_n$, 
\begin{align}
\Big \| K_n \ket{0^a}\ket{\phi} - \ket{0^a}  \widetilde{W} \ket{\phi} \Big \|_2 \leq \eps_6 \label{eqn:wk_bound}
\end{align}
for $K_n = S^\ell J$, $a' = \poly(n)$ and $\eps_6 = O\left(\ell \sqrt{\kappa + \eps_5}\right)$ as in \cref{lem:oaa_for_approx_iso}.
Since $J$ and $L$ are \pup-computable, it follows that $K_n$ is \pup-computable, too.

Our $\unitaryPSPACE$ algorithm for the Uhlmann Transformation Problem associated with $\{\ket{\psi_n}_{\reg{A_n B_n}} \}$ and $(\ket{\varphi_n}_{\reg{A_n B_n}})$ is the family of unitaries $\{ K_n \}_n$.
We have already shown that this family is \pup-computable and from the construction is is clear that $K_n$ does not act on $\reg{A_n}$ as required.
To conclude the proof, we need to evaluate how well $K_n$ solves the Uhlmann Transformation Problem with respect to the state pair $(\ket{\varphi_n},\ket{\psi_n})$, i.e.~we need to find an upper bound on 
\begin{align*}
&\Big \| I_{\reg{A}_n} \otimes   K_n \ket{\varphi_n}_{\reg{A}_n \reg{B}_n} \ket{0^{a'}}_{\reg{R}_n} - \ket{\psi_n}_{\reg{A}_n \reg{B}_n} \ket{0^{a'}}_{\reg{R}_n} \Big \|^2_2 \\
&\qquad = 2 - 2 \Re \left( \bra{\psi_n} \bra{0^{a'}} (I_{\reg{A}_n} \otimes   K_n) \ket{\varphi_n} \ket{0^{a'}}  \right) \\
&\qquad\leq 2 - 2 F(\rho_n,\sigma_n) + 2 \Big | \bra{\psi_n} \bra{0^{a'}} (I_{\reg{A}_n} \otimes   K_n) \ket{\varphi_n} \ket{0^{a'}} -  \fidelity(\rho_n,\sigma_n) \Big | \,, \numberthis \label{eqn:norm_fidelity_bound}
\end{align*}
where $\reg{R_n}$ contains the $a' = \poly(n)$ qubits used by $K_n$ as ancillas.

By the triangle inequality: 
\begin{align*}
&\Big | \bra{\psi_n} \bra{0^{a'}} (I_{\reg{A}_n} \otimes   K_n) \ket{\varphi_n} \ket{0^{a'}} -  \fidelity(\rho_n,\sigma_n) \Big | \\
&\leq \Big | \bra{\psi_n} \bra{0^{a'}} (I_{\reg{A}_n} \otimes   K_n) \ket{\varphi_n} \ket{0^{a'}} - \bra{\psi_n} \bra{0^{a'}} (I_{\reg{A}_n} \otimes   \widetilde W_n) \ket{\varphi_n} \ket{0^{a'}}\Big| + \Big| \bra{\psi_n} \bra{0^{a'}} (I_{\reg{A}_n} \otimes   \widetilde W_n) \ket{\varphi_n} \ket{0^{a'}} - \fidelity(\rho_n,\sigma_n) \Big | \\
&\leq \Big | \bra{\psi_n} \bra{0^{a'}} (I_{\reg{A}_n} \otimes K_n - I_{\reg{A}_n} \otimes   \widetilde W_n) \ket{\varphi_n} \ket{0^{a'}} \Big| + O(D_n^{(1)} \alpha_3 \kappa) \\
&\leq \Big | \bra{\psi_n} \bra{0^{a'}} (I_{\reg{A}_n} \otimes K_n - I_{\reg{A}_n} \otimes \widetilde W_n) \Pi_{\Gamma_n} \ket{\varphi_n} \ket{0^{a'}} \Big| + \Big | \bra{\psi_n} \bra{0^{a'}} (I_{\reg{A}_n} \otimes K_n - I_{\reg{A}_n} \otimes \widetilde W_n) \bar \Pi_{\Gamma_n} \ket{\varphi_n} \ket{0^{a'}} \Big| + O(D_n^{(1)} \alpha_3 \kappa)
\end{align*}
where the second inequality follows from \cref{lem:uhlmann-robust}, and for the second inequality we have introduced the projectors $\Pi_{\Gamma_n}$ onto the subspace $\Gamma_n$ and $\bar \Pi_{\Gamma_n}$ onto the orthogonal complement and used the triangle inequality.

We bound the first and second term in the last line separately.
For the first term, define $\ket{\tilde \phi_n} = \Pi_{\Gamma_n} \ket{\varphi_n} / \norm{\Pi_{\Gamma_n} \ket{\varphi_n}}_2$. 
Since $\norm{\Pi_{\Gamma_n} \ket{\varphi_n}}_2 \leq 1$, 
\begin{align*}
\Big | \bra{\psi_n} \bra{0^{a'}} (I_{\reg{A}_n} \otimes K_n - I_{\reg{A}_n} \otimes \widetilde W_n) \Pi_{\Gamma_n} \ket{\varphi_n} \ket{0^{a'}} \Big| 
&\leq \Big | \bra{\psi_n} \bra{0^{a'}} (I_{\reg{A}_n} \otimes K_n - I_{\reg{A}_n} \otimes \widetilde W_n) \ket{\tilde \varphi_n} \ket{0^{a'}} \Big| \\
&\leq \norm{(I_{\reg{A}_n} \otimes K_n - I_{\reg{A}_n} \otimes \widetilde W_n) \ket{\tilde \varphi_n} \ket{0^{a'}}}_2 \leq \eps_6
\end{align*}
by \cref{eqn:wk_bound}.
For the second term, we observe that $(I_{\reg{A}_n} \otimes K_n - I_{\reg{A}_n} \otimes \widetilde W_n) \bar \Pi_{\Gamma_n}$ does not act act register $\reg{A_n}$ and $\Pi_{\Gamma_n}$ only acts on register $\reg{B_n}$, so 
\begin{align*}
\Big | \bra{\psi_n} \bra{0^{a'}} (I_{\reg{A}_n} \otimes K_n - I_{\reg{A}_n} \otimes \widetilde W_n) \bar \Pi_{\Gamma_n} \ket{\varphi_n} \ket{0^{a'}} \Big|
&\leq \Big | \tr{( K_n - \widetilde W_n) \left(\bar \Pi_{\Gamma_n} \ptr{A_n}{\ket{\varphi_n}\!\bra{\psi_n}} \ot \proj{0^{a'}}_{\reg{R_n}} \right)} \Big| \\
&\leq \norm{K_n - \widetilde W_n}_\infty \cdot \norm{\bar \Pi_{\Gamma_n} \ptr{A_n}{\ket{\varphi_n}\!\bra{\psi_n}}}_1
\end{align*}
by H\"older's inequality~\cite[Corollary IV.2.6]{bhatia}.
Since $K_n$ is unitary and $\widetilde{W}_n$ has spectral norm $O(1)$ (due to the boundedness of $P_d^{\sgn}$), the first factor is $O(1)$. The second factor is equal to $ \alpha_3 \| R_n \overline{\Pi}_{\Gamma_n} \|_1$ which by definition is equal to the sum of the singular values of $R_n$ that are less than $\alpha_3 \kappa$. Since there are at most $D_n^{(2)}$ such singular values, we have that this is at most $O \Big( D_n^{(2)} \alpha_3 \kappa \Big)$.

Inserting these two bounds, we get that
\[
\Big | \bra{0} \bra{\psi_n} (I \otimes K_n) \ket{0} \ket{\varphi_n} -  \fidelity(\rho_n,\sigma_n) \Big | \leq O((D_n^{(1)} + D_n^{(2)})\alpha_3 \kappa) + \eps_6 \leq O(D_n\alpha_3 \kappa) + \eps_6~.
\]
We can simplify the error terms by noting that for sufficiently large $\alpha_5$ we have $\ell = O(\alpha_5)$
Furthermore, $\alpha_5 = O(\log(1/\kappa))$.
Choosing $\kappa = \eps_5 \geq 2^{-\poly(n)}$, we get that 
\begin{align*}
\eps_6 = O(\ell \sqrt{\kappa + \eps_5}) = O(\sqrt{\eps_5} \log(1/\eps_5)) \,.
\end{align*}
Similarly, recalling that 
\begin{align*}
\alpha_3 = D_n^{(1)} \alpha_2 = D_n^{(1)} \alpha_1^2 = D_n^{(1)} D_n^2 \leq D_n^3 \,,
\end{align*}
we get that $O(D_n \alpha_3 \kappa) = O(D_n^4 \eps_5)$.
Tracking $\eps_5$ through the proof, we see that for any given polynomial $q(n)$, we can find a (larger) polynomial $p(n)$ such that setting $\eps_1 = 2^{-p(n)}$ yields a value of $\eps_5$ such that 
\begin{align*}
O(\sqrt{\eps_5} \log(1/\eps_5)) + O(D_n^4 \eps_5) \leq 2^{-q(n) - 1}
\end{align*}
for sufficiently large $n$.

Inserting this into \cref{eqn:norm_fidelity_bound}, we get that 
\begin{align*}
    \Big \| I_{\reg{A}_n} \otimes   K_n \ket{\varphi_n}_{\reg{A}_n \reg{B}_n} \ket{0^{a'}}_{\reg{R}_n} - \ket{\psi_n}_{\reg{A}_n \reg{B}_n} \ket{0^{a'}}_{\reg{R}_n} \Big \|^2_2
    &= 2 (1 - \fidelity(\rho_n,\sigma_n)) + 2^{-q(n)} \,.
\end{align*}
This completes the proof for the case where the states $\ket{\psi_n},\ket{\varphi_n}$ are in $\statePSPACE_0$.

To handle the case when the states $\ket{\psi_n},\ket{\varphi_n}$, can only be synthesized up to error $\delta(n)$, we perform the same analysis for the family $(\ket{\psi_n'})_n,(\ket{\varphi_n'})_n \in \statePSPACE_0$ such that 
$\ket{\psi_n}$ and $\ket{\psi_n'}$ (resp. $\ket{\varphi_n}$ and $\ket{\varphi_n'}$) are $\delta(n)$-close, and then use the triangle inequality to deduce that 
\[
    \Big \| I \otimes K_n \ket{0} \ket{\varphi_n} - \ket{0} \ket{\psi_n} \Big \|^2 \leq \Big \| I \otimes K_n \ket{0} \ket{\varphi_n'} - \ket{0} \ket{\psi_n'} \Big \|^2 + O(\delta(n)). \qedhere
\]
\end{proof}

\bibliographystyle{alpha}
\bibliography{refs}

\appendix

\section{Robust oblivious amplitude amplification}
\label{app:oaa}

In this section we prove \Cref{lem:oaa_for_approx_iso}. Before doing so, we first present a utility lemma about embedding \emph{almost unitary} matrices in unitary matrices.

\begin{lemma}[Embedding almost-unitary matrices]
\label{lem:unitary-embed} 
Let $A \in L(\C^d)$ be such that $(1-\delta)\1 \leq \Sigma \leq \1$. Let $A = U \Sigma V^\dagger$ denote its singular value decomposition. Let $B = U \sqrt{1 - \Sigma^2} V^\dagger$. Then the matrix
$$
    A_0 = \begin{pmatrix} 
    A & -B \\
    B & A
    \end{pmatrix} = \ketbra{0}{0} \otimes A - \ketbra{0}{1} \otimes B + \ketbra{1}{0} \otimes B + \ketbra{1}{1} \otimes A
$$
is unitary and furthermore satisfies for all $\ket{\phi} \in \C^d$ 
\begin{gather*}
    \| \ket{0} A \ket{\phi} - A_0 \ket{0} \ket{\phi} \|_2  \leq O(\sqrt{\delta}) \\
    \| \ket{0} A^\dagger \ket{\phi} - A_0^\dagger \ket{0} \ket{\phi} \|_2 \leq O(\sqrt{\delta})~.
\end{gather*}
\end{lemma}
\begin{proof}
It is easy to verify that $A_0$ is unitary. To show the ``furthermore'' part, we calculate
\begin{align*}
    \| \ket{0} A \ket{\phi} - A_0 \ket{0} \ket{\phi} \|_2 &= \| \ket{1} \otimes B \ket{\phi} \|_2 \leq \| B \|_\infty \leq \| \sqrt{1 - \Sigma^2} \|_\infty~.
\end{align*}
The minimum singular value of $\Sigma$ is at least $1-\delta$. 
Thus $\| \sqrt{1 - \Sigma^2} \|_\infty$ is at most $\sqrt{1 - (1 - \delta)^2} \leq O(\sqrt{\delta})$ as desired. The calculation for the $A^\dagger$, $A_0^\dagger$ case is identical.
\end{proof}

Next, we recall the \emph{non}-robust version of oblivious amplitude amplification from~\cite{berry2014exponential}.
The following is essentially a rephrasing of \cite[Lemma 3.6]{berry2014exponential} in terms of block encodings.
Here we write the ancilla qubits of the block encoding as the first register to match the notation of \cite{berry2014exponential}.
\begin{lemma}[Exact oblivious amplitude amplification]
\label{lem:oaa-exact}
Let $J_0$ be an $(\alpha,0,b)$-block encoding of an operator $W$ that is an (exact) isometry on a subspace $\Gamma$. Let $\Pi_0 = \ketbra{0^b}{0^b} \ot \Id$ denote the projector onto the ancilla qubits being zero, let $L_0 = 2\Pi_0 - I$, and let $S_0 = -J_0 L_0 J_0^\dagger L_0$. Then for all $\ell \in \N$ and all $\ket{\phi} \in \Gamma$, 
\[
    \Big \| S_0^\ell J_0 \ket{0^b}\ket{\phi} - \sin((2\ell + 1)\theta) \, \ket{0^b}  W \ket{\phi} \Big \| = \cos((2\ell + 1)\theta)
\]
where $\theta = \arcsin(\alpha^{-1})$.
\end{lemma}
\begin{proof}
The proof is nearly identical to that \cite[Lemma 3.6]{berry2014exponential}, but we briefly describe the modifications one needs to make to the proof of \cite[Lemma 3.6]{berry2014exponential} for the sake of completeness.
For any state $\ket{\psi} \in \Gamma$, we have that 
\begin{align*}
J_0 \ket{0^b} \ket{\psi} = \frac{1}{\alpha} \ket{0^b} W \ket{\psi} + \sqrt{1 - \frac{1}{\alpha^2}} \ket{\Phi^{\perp}}
\end{align*}
for $W \ket{\psi}$ a normalised state (because $W$ is an isometry on $\Gamma$ and $\ket{\psi} \in \Gamma$) and a state $\ket{\Phi^{\perp}}$ such that $\Pi_0 \ket{\Phi^{\perp}} = 0$.
The only difference to the starting point in \cite{berry2014exponential} is that the above equation only holds for $\ket{\psi} \in \Gamma$, not all $\ket{\psi}$.

With this, we can follow the steps of the proof of \cite[Lemma 3.7]{berry2014exponential} unchanged, except that we need the additional property that the subspace $\Gamma$ is invariant under the operator $Q$ defined in \cite[Eqn.~(12)]{berry2014exponential}.
This is easily seen to hold because in our case, we have $Q = \frac{1}{\alpha^2} W^\dagger W$ by the definition of block encodings.
With this, the calculation in \cite[Eqn.~(13)]{berry2014exponential} allows us to conclude that $Q$ acts as a scalar multiple of the identity on the subspace $\Gamma$.
Then, the remainder of the proofs of \cite[Lemmas 3.6 and 3.7]{berry2014exponential} goes through unchanged.
\end{proof}

We can now prove \Cref{lem:oaa_for_approx_iso} via reduction to the non-robust version \Cref{lem:oaa-exact}. 
\begin{proof}[Proof of \cref{lem:oaa_for_approx_iso}]
Let $J$ be the $(\alpha,\epsilon,b)$-block encoding of an operator $\widetilde{W}$. Let $\Delta_{1 \pm \kappa}$ denote the span of the right singular vectors of $\widetilde{W}$ whose corresponding singular values are in $[1-\kappa,1+\kappa]$. Let $\Gamma$ denote a subspace of $\Delta_{1 \pm \kappa}$. Thus $\widetilde{W}$ is a $\kappa$-approximate isometry on $\Gamma$. 
Let $\widetilde{W} = U \Sigma V^\dagger$ denote its singular value decomposition and write
\[
    \widetilde{W} = U (\Sigma_{in} + \Sigma_{out}) V^\dagger
\]
where $\Sigma_{in}$ denotes only the diagonal entries of $\Sigma$ that are in the interval $[1-\kappa,1+\kappa]$, and $\Sigma_{out} = \Sigma - \Sigma_{in}$. 

Define the matrix
\[
    W = U(\sgn(\Sigma_{in}) + \Sigma_{out})V^\dagger
\]
where $\sgn(\Sigma_{in})$ denotes rounding all positive diagonal entries to $1$ and keeping the zeros at zero. It is easy to see that 
\begin{align}
\| W - \widetilde{W} \|_\infty \leq \kappa\,,\label{eqn:rob_wbound}
\end{align} 
and furthermore, $W$ is an (exact) isometry on the subspace $\Gamma$. 

Define the matrix $J'$ to be $J$ except the top left corner is replaced by $W$. In other words:
\[
    J' = \frac{1}{\alpha} \ketbra{0^a}{0^a} \otimes W + \Big(J - \Pi J \Pi \Big)
\]
where $\Pi = \ketbra{0^a}{0^a} \otimes I$. Note that
\begin{align*}
    \Big \| J - J' \Big \|_\infty &= \Big \| \frac{1}{\alpha} W - \bra{0^a} J \ket{0^a} \Big \|_\infty \\
    &\leq \frac{1}{\alpha} \Big \| W - \widetilde{W} \Big \|_\infty + \Big \| \frac{1}{\alpha} \widetilde{W} - \bra{0^a} J \ket{0^a} \Big \|_\infty \\
    &\leq (\kappa + \eps)/\alpha~.
\end{align*}
Define $J'' = \frac{1}{1 + \beta} J'$ for $\beta = (\kappa + \eps)/\alpha$. 
Using the triangle inequality and $1 - \frac{1}{1+\beta} \leq \beta$, we get that $\| J - J'' \|_\infty \leq \beta + \beta (1 + \beta) \leq 3 \beta$ for $\beta \leq 1$.
Since $J$ is unitary, its singular values are 1, so the singular values of $J''$ are in the interval $[1 - 3 \beta, 1]$.

We then apply \Cref{lem:unitary-embed} to obtain a unitary $J_0$ that has $J''$ at its top left corner:
\[
    (\bra{0} \ot I) J_0 (\ket{0} \ot I) = J''~.
\]
Note that $J_0$ is a $(\gamma,0,b)$-block encoding of $W$, where $\gamma = \alpha(1+\beta) = \alpha + \kappa + \eps$ and now the ancilla space is $b = a +1$ qubits. We apply \Cref{lem:oaa-exact} with $b = a+1$ to get that for all $\ket{\phi} \in \Gamma$, 
\begin{align*}
\Big \| S_0^\ell J_0 \ket{0^b}\ket{\phi} - \sin((2\ell + 1)\theta_0) \, \ket{0^b}  W \ket{\phi} \Big \| = \cos((2\ell + 1)\theta_0)
\end{align*}
where $\theta_0 = \arcsin(\gamma^{-1})$.
We can bound the difference between $\cos((2\ell + 1)\theta_0)$ and $\cos((2\ell+1)\theta)$ as follows. 
Using that the derivative of cosine is at most $1$ everywhere,
\begin{align*}
    |\cos((2\ell+1)\theta) - \cos((2\ell+1)\theta_0)| &\leq O(\ell \cdot |\theta - \theta_0|)~.
\end{align*}
Since $\theta = \arcsin(\alpha^{-1})$ and $\theta_0 = \arcsin((\alpha + \kappa + \eps)^{-1}) = \arcsin(\alpha^{-1} - \eta)$ where $\eta = \alpha^{-1} - (\alpha + \kappa + \eps)^{-1} = O((\kappa + \eps)/\alpha^2) = O(\beta/\alpha)$, we get that $|\theta - \theta_0| \leq O(\beta/\alpha)$. 
Therefore, we obtain 
\begin{align}
\Big \| S_0^\ell J_0 \ket{0^b}\ket{\phi} - \sin((2\ell + 1)\theta_0) \, \ket{0^b}  W \ket{\phi} \Big \| = \cos((2\ell + 1)\theta) + O(\ell \beta/\alpha)~.
\label{eqn:rob_p1}
\end{align}

As a second step, we will show that 
\begin{align}
\Big \| S^\ell_0 J_0 \ket{0^b} \ket{\phi} - \ket{0} S^\ell J \ket{0^a} \ket{\phi}  \Big \| \leq O(\ell \sqrt{\beta}) \,. \label{eqn:rob_p2}
\end{align}
For this, we employ a hybrid argument.
By the triangle inequality:
\begin{multline*}
\Big \| S^\ell_0 J_0 \ket{0^b} \ket{\phi} - \ket{0} S^\ell J \ket{0^a} \ket{\phi}  \Big \| 
\leq \Big \| S^\ell_0 J_0 \ket{0^b} \ket{\phi} - S^\ell_0 \ket{0^b} J \ket{\phi} \Big \| \\
+ \sum_{k = 0}^{\ell-1}\Big \| S_0^{\ell - k} \ket{0} S^k J \ket{0^a} \ket{\phi} - S_0^{\ell - k - 1} \ket{0} S^{k+1} J \ket{0^a} \ket{\phi} \Big \|~.
\end{multline*}
Since $S_0^\ell$ is a unitary, the first term can be bounded as follows:
\begin{align*}
\Big \| S^\ell_0 J_0 \ket{0^b} \ket{\phi} - S^\ell_0 \ket{0^b} J \ket{\phi} \Big \| = \Big \| J_0 \ket{0^b} \ket{\phi} - \ket{0} J \ket{0^a} \ket{\phi} \Big \| \leq \Big \| J_0 \ket{0^b} \ket{\phi} - \ket{0} J'' \ket{0^a} \ket{\phi}  \Big \| + \| J'' - J \|_\infty ~.
\end{align*}
The first term on the right hand side is at most $O(\sqrt{\beta})$ by \Cref{lem:unitary-embed}, and the second term is at most $O(\beta)$ as argued previously.

To bound the terms in the sum, we first use that $S_0^{\ell - k - 1}$ is unitary, so 
\begin{align*}
\Big \| S_0^{\ell - k} \ket{0} S^k J \ket{0^a} \ket{\phi} - S_0^{\ell - k - 1} \ket{0} S^{k+1} J \ket{0^a} \ket{\phi} \Big \| = \Big \| (S_0\ket{0} - \ket{0} S) S^k J \ket{0^a} \ket{\phi}\Big \| \,.
\end{align*}
(Note that despite what this compact notation might suggest, of course $S_0$ does not just act on the first qubit, but the entire state.)
Abbreviating $\ket{\theta} = S^k J \ket{0^a} \ket{\phi}$, this can be bounded as
\begin{align*}
    \Big \| (S_0 \ket{0} - \ket{0} S) \ket{\theta} \Big \| &= \Big \| (J_0 L_0 J_0^\dagger L_0 \ket{0} - \ket{0} J L J^\dagger L) \ket{\theta}\Big \| \\
    &= \Big \| (J_0 L_0 J_0^\dagger\ket{0} - \ket{0} J L J^\dagger) L \ket{\theta}\Big \| \\
    &\leq \Big \| (J_0 L_0 J_0^\dagger\ket{0} - \ket{0} J'' L (J'')^\dagger) L \ket{\theta}\Big \| + O(\sqrt{\beta})
\end{align*}
where in the second line we used that $L_0 = \ketbra{0}{0} \otimes L$, and in the third line we used that $\| J - J'' \|_\infty \leq O(\sqrt{\beta})$. 
Again using the triangle inequality and unitary invariance of the norm, we have
\begin{align*}
    &\Big \| (J_0 L_0 J_0^\dagger\ket{0} - \ket{0} J'' L (J'')^\dagger) L \ket{\theta}\Big \| \\
    &\quad \leq \Big \| (J_0^\dagger\ket{0} - \ket{0} (J'')^\dagger) L \ket{\theta}\Big \| + \Big \| (L_0 \ket{0}  - \ket{0} L) (J'')^\dagger L \ket{\theta}\Big \| \\
    &\qquad\qquad + \Big \| (J_0 \ket{0} - \ket{0} J'') L (J'')^\dagger L \ket{\theta}\Big \| = O(\sqrt{\beta})
\end{align*}
because the first and third terms are bounded by $\| J - J'' \|_\infty \leq O(\sqrt{\beta})$, and the second term is equal to $0$ because $L_0 \ket{0} = \ket{0} L$.
This show that \cref{eqn:rob_p2} holds.

As a third step, we bound 
\begin{align}
\Big \| \sin((2\ell + 1)\theta_0) \, \ket{0^b}  W \ket{\phi} - \sin((2\ell + 1)\theta) \, \ket{0^b}  \widetilde{W} \ket{\phi} \Big \| \leq O(\ell \beta/\alpha) + \kappa \,. \label{eqn:rob_p3}
\end{align}
This holds because we can first switch $\theta_0$ to $\theta$ inside the sine function using the same argument we used for \cref{eqn:rob_p1}, incurring an error $O(\ell \beta/\alpha)$.
Then we can switch $W$ to $\tilde W$, incurring an error $\kappa$ due to \cref{eqn:rob_wbound}.

Combining \cref{eqn:rob_p1}, \cref{eqn:rob_p2}, and \cref{eqn:rob_p3} using the triangle inequality and keeping the asymptotically dominant error term, we find that 
\begin{align*}
\Big \| S^\ell J \ket{0^a}\ket{\phi} - \sin((2\ell + 1)\theta) \, \ket{0^a}  \widetilde{W} \ket{\phi} \Big \| 
&\leq \cos((2\ell + 1)\theta) + O(\ell \sqrt{\beta}) + \kappa \\
&= \cos((2\ell + 1)\theta) + O\left(\ell \sqrt{\kappa + \eps}\right)
\end{align*}
\end{proof}

\section{Details of SDP for $\stateQIP$ protocols} \label{sec:sdp_details}
In this appendix, we give a formal description of the SDP for $\stateQIP$ protocols described in \cref{sec:sdp-qip}.
We focus on the primal form of the SDP as neither our \cref{algo:mmwu} nor our analysis requires the dual form.
As in \cref{sec:sdp-qip}, we will be considering the intermediate states of the protocol on the message register and the verifier's private register, i.e.~the states $\rho_{\reg{M_1 W_0}}, \rho_{\reg{M'_1 W_1}}, \dots, \rho_{\reg{M_r W_{r-1}}}, \rho_{\reg{M'_r W_{r}}}$, where we set $\reg{M'_{r}} \deq \reg{ZS}$ to ease the notation.
Recall from \cref{sec:sdp-qip} that our SDP needs to incorporate the following constraints (for some $c = c(n)$, where we drop the $n$-dependence as we consider $n$ fixed for the purposes of constructing the SDP):
\begin{enumerate}[label=(\arabic*)]
\item $\rho_{\reg R} \geq 0$ and $\tr{\rho_{\reg R}} = 1$ for all $\reg R \in \{\reg{M_i W_{i-1}}, \reg{M'_i W_i}\}_{i = 1, \dots, r}$.
\item $\rho_{\reg{M'_i W_{i}}} - \Phi_{C_i}(\rho_{\reg{M_i W_{i-1}}}) = 0$ for $i = 1, \dots, r$.
\item $\ptr{\reg{M'_{i}}}{\rho_{\reg{M'_{i} W_{i}}}} - \ptr{\reg{M_{i+1}}}{\rho_{\reg{M_{i+1} W_{i}}}} = 0$ for $i = 1, \dots, r-1$.
\item $\rho_{\reg W_0} = \proj{0\dots 0}_{\reg{W_0}}$.
\item $\tr{\proj{1}_\reg{Z}\rho_{\reg Z}} = c$.
\end{enumerate}

We need to write these constraints in the required form $\Phi(A) = B$ of an SDP.
For this, we need to combine that different states $\rho_{\reg{M_1 W_0}}, \rho_{\reg{M'_1 W_1}}, \dots, \rho_{\reg{M_r W_{r-1}}}, \rho_{\reg{M'_r W_{r}}}$ into one large operator $A$.
There are two ways to do this: we could take either a direct sum or a tensor product of all the intermediate states.
Both options work.
The direct sum will result in a smaller SDP (as the dimensions of the intermediate states are added, not multiplied), while the tensor product will be slightly more convenient to work with.
Since the size of our SDP is anyway exponential, an additional exponential factor in the size does not bother us, so we combine the states with a tensor product.

Concretely, we consider $A \in \linear \left( \otimes_{\reg R \in \cR} \reg R \right)$, where the tensor product is over $\cR \deq \{\reg{M_i W_{i-1}}, \reg{M'_i W_i}\}_{i = 1, \dots, r}$.
As explained in \cref{sec:sdp-qip}, this will contain two copies of most of the individual registers involved (e.g.~two copies of $\reg{M_1, W_1}$, etc).
We use the notation $\ptr{\setminus \reg R}{A}$ to denote the partial trace over everything except registers $\reg R$ for $\reg R \in \{\reg{M_i W_{i-1}}, \reg{M'_i W_i}\}$.\footnote{As in \cref{sec:sdp-qip}, if we use this notation $\reg R$ always has to include sufficiently many registers so that it uniquely specifies one of the intermediate states, and it is understood that the partial trace then returns the corresponding intermediate state. In particular, this is the reason for the slightly awkward notation in the implementation of constraint (3) below.}
We now need to write the constraints (1) - (5) as linear constraints on $A$.
To this end, we define the following linear maps and matrices, which implement the constraints (1) - (5) above:
\begin{enumerate}[label=(\arabic*)]
\item $\Gamma^{(1)}(A) = \tr{A}$ and $\tilde B^{(1)} = 1$ (interpreted as a $1\times 1$ matrix).
Note that the condition $\rho_{\reg R} = \ptr{\setminus \reg R}{A} \geq 0$ is implied by $A \geq 0$ (which is a condition on $A$ in any SDP feasibility problem) and the constraint $\tr{A} = 1$ implies the constraints $\tr{\ptr{\setminus \reg R}{A}} = 1$ for all $\reg R \in \cR$, so $\Gamma^{(1)}(A) = \tilde B^{(1)}$ enforces constraint (1) from above.
\item $\Gamma^{(2)}_i(A) = \ptr{\setminus (\reg{M'_i W_i})}{A} - \Phi_{C_i}\left(\ptr{\setminus (\reg{M_i W_{i-1}})}{A}\right)$ and $\tilde B^{(2)}_i = 0_{\reg{M'_iW_i}}$ for $i = 1, \dots, r$. Here, $0_{\reg{M'_iW_i}}$ denotes the square all-0 matrix of size $\dim(\reg{M'_iW_i}) \times \dim(\reg{M'_iW_i})$.
\item $\Gamma^{(3)}_i(A) = \ptr{\reg{M'_i}}{\ptr{\setminus(\reg{M'_iW_i})}{A}} - \ptr{\reg{M_{i+1}}}{\ptr{\setminus(\reg{M_{i+1}W_i})}{A}}$ and $\tilde B^{(3)}_i = 0_{\reg{W_i}}$ for $i = 1, \dots, r-1$.
\item $\Gamma^{(4)}(A) = \ptr{\reg{M_1}}{\ptr{\setminus \reg{M_1 W_0}}{A}}$ and $\tilde B^{(4)} = \proj{0\dots 0}_{\reg{W_0}}$.
\item $\Gamma^{(5)}(A) = \tr{(\id_{\reg{W_r S}} \ot \proj{1}_{\reg{Z}}) \ptr{\setminus (\reg{M'_r W_r})}{A}}$ (recalling that $\reg{M'_r} \deq \reg{ZS}$) and $B^{(5)} = c$ (interpreted as a $1 \times 1$-matrix).
\end{enumerate}
By construction, it is clear that these linear constraints of the form $\Gamma(A) = \tilde B$ together implement the desired constraints (1)-(5).

We now need to combine them into one linear operator $\Phi$ and one matrix $B$.
We can do so simply by taking a direct a sum.
To ease the notation, for this we relabel the linear constraints above as $\Phi^{(j)}$ and $B^{(j)}$ (according to some arbitrary numbering scheme) so that 
\begin{align*}
\{\Phi^{(j)}\}_{j = 1, \dots, 2r + 2} &= \{\Gamma^{(1)}, \Gamma^{(2)}_i, \Gamma^{(3)}_{i'}, \Gamma^{(4)}, \Gamma^{(5)}\}_{i = 1, \dots, r; i' = 1, \dots, r-1} \\
\tand \{B^{(j)}\}_{j = 1, \dots, 2r + 2} &= \{\tilde B^{(1)}, \tilde B^{(2)}_i, \tilde B^{(3)}_{i'}, \tilde B^{(4)}, \tilde B^{(5)}\}_{i = 1, \dots, r; i' = 1, \dots, r-1} \,.
\end{align*}
These are the linear maps and matrices referred to in \cref{sec:sdp-qip}.
We can combine these constraints by defining 
\begin{align*}
\Phi(A) \deq \frac{1}{N} \oplus_j \Phi^{(j)}(A) \tand B = \frac{1}{N} \oplus_j B^{j}
\end{align*}
for some normalisation factor $N = O(r)$ that is large enough so that $\Phi^*$ is contracting. (As shown in \cref{sec:sdp-qip}, a factor of size $O(r)$ will suffice for this.)
Then, the feasibility SDP $(\Phi, B)$ will satisfy the conditions of \cref{lem:qip-as-sdp} as shown in \cref{sec:sdp-qip}.

\end{document}

%% file: headers.tex
\usepackage[in]{fullpage}

\usepackage[shortlabels]{enumitem}
\usepackage{amsfonts}
\usepackage{amsmath}
\usepackage{mathtools,amsthm,amssymb} 
\usepackage{xcolor}
\usepackage{graphicx}
\usepackage{qtree}
\usepackage{tree-dvips}
\usepackage{float}
\definecolor{linkblue}{HTML}{001487}
\usepackage[colorlinks=true,allcolors=linkblue]{hyperref}
\usepackage[nameinlink,capitalize,noabbrev]{cleveref}
\usepackage{braket}
\usepackage{mathrsfs}
\usepackage{tikz}
\usepackage{qcircuit}
\usepackage{xspace}
\usepackage{dsfont}
\usepackage{enumitem} 
\setenumerate[0]{label={\normalfont (\roman*)}}
\usepackage{longfbox}
\usepackage{csquotes}
\usepackage{authblk}

\theoremstyle{plain}
\newtheorem{theorem}{Theorem}[section]
\newtheorem{lemma}[theorem]{Lemma}

\Crefname{claim}{Claim}{Claims}
\newtheorem*{lemma*}{Lemma}

\newtheorem{corollary}[theorem]{Corollary}

\newtheorem{definition}[theorem]{Definition}

\newtheorem{algorithm}{Algorithm}
\numberwithin{equation}{section}

\newcommand\numberthis{\addtocounter{equation}{1}\tag{\theequation}}

\usepackage[T1]{fontenc}
\usepackage[UKenglish]{babel}
\usepackage{helvet}
\usepackage{mathptmx}
\DeclareMathAlphabet{\mathcal}{OMS}{ntxm}{m}{n}
\let\mathbb\relax
\let\mathbb\mathds
\usepackage{stmaryrd} 
\DeclareSymbolFont{greekletters}{OML}{ntxmi}{m}{it}
\DeclareMathSymbol{\alpha}{\mathord}{greekletters}{"0B}
\DeclareMathSymbol{\beta}{\mathord}{greekletters}{"0C}
\DeclareMathSymbol{\gamma}{\mathord}{greekletters}{"0D}
\DeclareMathSymbol{\delta}{\mathord}{greekletters}{"0E}
\DeclareMathSymbol{\epsilon}{\mathord}{greekletters}{"0F}
\DeclareMathSymbol{\zeta}{\mathord}{greekletters}{"10}
\DeclareMathSymbol{\eta}{\mathord}{greekletters}{"11}
\DeclareMathSymbol{\theta}{\mathord}{greekletters}{"12}
\DeclareMathSymbol{\iota}{\mathord}{greekletters}{"13}
\DeclareMathSymbol{\kappa}{\mathord}{greekletters}{"14}
\DeclareMathSymbol{\lambda}{\mathord}{greekletters}{"15}
\DeclareMathSymbol{\mu}{\mathord}{greekletters}{"16}
\DeclareMathSymbol{\nu}{\mathord}{greekletters}{"17}
\DeclareMathSymbol{\xi}{\mathord}{greekletters}{"18}
\DeclareMathSymbol{\pi}{\mathord}{greekletters}{"19}
\DeclareMathSymbol{\rho}{\mathord}{greekletters}{"1A}
\DeclareMathSymbol{\sigma}{\mathord}{greekletters}{"1B}
\DeclareMathSymbol{\tau}{\mathord}{greekletters}{"1C}
\DeclareMathSymbol{\upsilon}{\mathord}{greekletters}{"1D}
\DeclareMathSymbol{\phi}{\mathord}{greekletters}{"1E}
\DeclareMathSymbol{\chi}{\mathord}{greekletters}{"1F}
\DeclareMathSymbol{\psi}{\mathord}{greekletters}{"20}
\DeclareMathSymbol{\omega}{\mathord}{greekletters}{"21}
\DeclareMathSymbol{\varepsilon}{\mathord}{greekletters}{"22}
\DeclareMathSymbol{\vartheta}{\mathord}{greekletters}{"23}
\DeclareMathSymbol{\varpi}{\mathord}{greekletters}{"24}
\DeclareMathSymbol{\varrho}{\mathord}{greekletters}{"25}
\DeclareMathSymbol{\varsigma}{\mathord}{greekletters}{"26}
\DeclareMathSymbol{\varphi}{\mathord}{greekletters}{"27}

\newcommand{\class}[1]{\mathsf{#1}}
\newcommand{\poly}{\mathrm{poly}}

\newcommand{\statePSPACE}{\class{statePSPACE}}
\newcommand{\stateQIP}{\class{stateQIP}}
\newcommand{\QIP}{\class{QIP}}
\newcommand{\IP}{\class{IP}}
\newcommand{\PSPACE}{\class{PSPACE}}
\newcommand{\pspace}{$\class{PSPACE}$\xspace}
\newcommand{\stateqip}{\class{stateQIP}}
\newcommand{\pup}{$\class{pureUnitaryPSPACE}$\xspace}
\newcommand{\unitaryPSPACE}{\class{unitaryPSPACE}}
\newcommand{\pureUnitaryPSPACE}{\class{pureUnitaryPSPACE}}

\newcommand*{\interact}{\mathord{\leftrightarrows}}

\newcommand{\eps}{\epsilon}

\newcommand{\ketbra}[2]{\ket{#1}\!\!\bra{#2}}
\renewcommand{\cal}[1]{\mathcal{#1}}
\newcommand{\R}{\mathbb{R}}
\newcommand{\C}{\mathbb{C}}
\newcommand{\N}{\mathbb{N}}

\newcommand{\Tr}{\mathrm{Tr}}

\newcommand{\reg}[1]{\mathsf{#1}}

\newcommand{\Id}{I}

\newcommand{\td}{\mathrm{td}}

\newcommand{\fidelity}{\mathrm{F}}

\newcommand{\setft}[1]{\textnormal{#1}}
\newcommand{\id}{\setft{id}}

\newcommand{\bits}{\ensuremath{\{0, 1\}}}
\newcommand{\linear}{\mathrm{L}}

\newcommand{\ot}{\ensuremath{\otimes}}
\newcommand{\deq}{\coloneqq}
\newcommand{\tr}[1]{\mathrm{Tr}\!\left[ #1 \right]}
\newcommand{\ptr}[2]{\mbox{\rm Tr}_{#1}\!\left[ #2 \right]}
\newcommand{\pr}[1]{{\rm Pr}\!\left[ #1 \right]}

\newcommand{\norm}[1]{\left\lVert#1\right\rVert}

\DeclareMathOperator{\sgn}{sgn}
\DeclareMathOperator{\erf}{erf}
\let\1\relax
\newcommand{\1}{\mathds{1}}

\newcommand{\tand}{\;\textnormal{~and~}\;}

\newcommand{\GibbsOracle}{\textsc{GibbsOracle}\xspace}
\newcommand{\TraceDistOracle}{\textsc{TraceDistanceOracle}}

\newcommand{\proj}[1]{\ket{#1}\!\!\bra{#1}}

\newcommand{\cH}{\ensuremath{\mathcal{H}}}

\newcommand{\cM}{\ensuremath{\mathcal{M}}}

\newcommand{\cR}{\ensuremath{\mathcal{R}}}
\newcommand{\cS}{\ensuremath{\mathcal{S}}}